\documentclass[12pt,letterpaper]{article}

\usepackage{amssymb,amsthm,amsmath,multirow,amsfonts,enumerate,setspace,pdflscape,lscape}
\usepackage[T1]{fontenc}
\usepackage{times}
\usepackage[margin=1in]{geometry}
\usepackage[small]{titlesec}

\usepackage[round,authoryear]{natbib}
\usepackage[colorlinks=true,linkcolor=red,urlcolor=blue,citecolor=blue,anchorcolor=blue,
pdftex,breaklinks,pdfencoding=auto,psdextra,
bookmarksopenlevel=1,bookmarksopen=true]{hyperref}

\usepackage{booktabs}
\usepackage{graphicx,epstopdf,color}
\graphicspath{{graphics/}}
\usepackage{subcaption} 
\makeatletter
\DeclareRobustCommand\citepos
{\begingroup\def\NAT@nmfmt##1{{\NAT@up##1's}}%
	\NAT@swafalse\let\NAT@ctype\z@\NAT@partrue
	\@ifstar{\NAT@fulltrue\NAT@citetp}{\NAT@fullfalse\NAT@citetp}}
\makeatother

\makeatletter
\renewcommand*{\eqref}[1]{\hyperref[{#1}]{\textup{\tagform@{\ref*{#1}}}}}
\def\thm@space@setup{%
	\thm@preskip=0.2\baselineskip plus 0.05\baselineskip minus 0.05\baselineskip
	\thm@postskip=\thm@preskip
}
\makeatother

\usepackage[dvipsnames]{xcolor}

\colorlet{redby}{black}

\usepackage[inline,shortlabels]{enumitem}

\expandafter
\def \expandafter \normalsize \expandafter{\normalsize \setlength \abovedisplayskip{5pt plus 2pt minus 5pt}}
\expandafter
\def \expandafter \normalsize \expandafter{\normalsize \setlength \abovedisplayshortskip{0pt plus 2pt}}
\expandafter
\def \expandafter \normalsize \expandafter{\normalsize \setlength \belowdisplayskip{5pt plus 2pt minus 5pt}}
\expandafter
\def \expandafter \normalsize \expandafter{\normalsize \setlength \belowdisplayshortskip{5pt plus 2pt minus 3pt}}
\def\Var{\mathop{\hbox{\rm Var}}}
\theoremstyle{plain}
\newtheorem{theorem}{Theorem}
\newtheorem{lemma}{Lemma}

\theoremstyle{definition}
\newtheorem{assumption}{Assumption}
\theoremstyle{example}
\newtheorem{example}{Example}
\newtheorem{rmk}{Remark}
\newenvironment{remark}{\begin{rmk}}{\hfill $\square$ \end{rmk}}
\DeclareMathOperator{\ran}{ran}
\DeclareMathOperator{\cl}{cl}

\DeclareMathOperator{\tr}{tr}

\def\dto{\overset{d}\rightarrow}
\def\pto{\overset{p}\rightarrow}
\def\hto{\hspace{-0.2em}\to\hspace{-0.2em}}
\numberwithin{equation}{section}

\titlespacing{\section}{0pt}{2pt}{0pt}% this reduces space between (sub)sections to 0pt, for example
 \titlespacing{\subsection}{0pt}{0pt}{0pt} % this reduces space between (sub)sections to 0pt, for example
 
 \titlespacing{\subsubsection}{0pt}{2pt}{0pt}

\newtheorem{proposition}{Proposition}

\makeatletter
\usepackage{xr}

\newcommand{\violet}{\textcolor{black}}
\renewcommand{\hat}{\widehat} 

\newcommand{\op}{\text{op}}
\newcommand{\HS}{\text{HS}}
\begin{document}
	\title{Binary response model with many weak instruments\thanks{This paper originated as part of my PhD dissertation at the University of California, Davis, but has since undergone substantial revisions. \violet{I am deeply grateful for the invaluable feedback received from the Co-Editor, three anonymous referees,} as well as from Shu Shen, James G. MacKinnon, A. Colin Cameron, Takuya Ura, Morten \O. Nielsen, Won-Ki Seo, and all seminar participants at Queen's University, the University of California, Davis, National Taiwan University, the University of Sydney, the University of Melbourne, the University of Technology Sydney, Monash University, Tilburg University, the Higher School of Economics (ICEF), Vrije Universiteit Amsterdam, Academia Sinica, Korea Development Institute, the 2019 All California Econometrics Conference, the 2021 Korea's Allied Economic Associations Annual Meeting, and the 4th International Conference on Econometrics and Statistics.}}
	\author{Dakyung Seong\\ \small School of Economics, University of Sydney
		\\ \small \texttt{dakyung.seong@sydney.edu.au} }
	\date{\small{\today}} 
 	\maketitle
	\vspace{-2em}
	\begin{abstract} 
		This paper considers an endogenous binary response model with many weak instruments. \violet{We employ a control function approach and a regularization scheme to obtain better estimation results for the endogenous binary response model in the presence of many weak instruments.} Two consistent and asymptotically normally distributed estimators are provided, each of which is called a regularized conditional maximum likelihood estimator (RCMLE) and a regularized nonlinear least squares estimator (RNLSE). Monte Carlo simulations show that  the proposed estimators outperform the existing ones when there are many weak instruments. We use the proposed estimation method to examine the effect of family income on college completion.
		
		
	\end{abstract}\clearpage

\section{Introduction}\label{sec:intro} 
The conventional  two-stage least squares (TSLS) estimation method has been widely used in empirical studies in economics and other fields of social science, even when the dependent variable is binary  (see, e.g., \citealp*{Miguel2004}; \citealp{Norton2008}; \citealp*{Nunn2014}; \citealp{Bastian2018}). In spite of it not being recommended from a theoretical perspective, the empirical popularity of  TSLS in binary response models is attributed not only to its ease of interpretation, but also to the fact that some essential statistical properties of estimators for {\it{nonlinear}} binary response models have only been studied in limited settings.  %From a theoretical view, this approach is not recommended. Nonetheless, it has been popularly employed not only because of the ease of interpretation, but also because of the fact that some essential statistical properties of estimators for {\it{nonlinear}} binary response models have only been studied in limited settings. 
For example, in contrast to the growing body of work on instrumental variable estimators for linear  models with many instruments, little attention has been given to the statistical properties of estimators for binary response models in a similar context. % statistical properties of estimators for binary response models in such a scenario have not been discussed yet.
 Moreover, only a little is known when there are weak or nearly weak instruments, see e.g., \cite{Magnusson2010} and \cite{Dufour2018}.

\violet{This paper aims to develop a valid estimation and inference method for the endogenous binary response model when the dimension of the instrumental variable (or sometimes called the number of instruments) is very large.} %the authors show that the TSLS estimation approach results in a loss of efficiency even though the second-stage nonlinearity, caused by the binary nature of the outcome variable, is taken into account. As an alternative, they%
 In principle, using a large number of instruments is helpful for improving estimation efficiency. However, in spite of this advantage,  this is not always recommended in finite samples.\label{r1p3} \violet{This is because the use of many instruments can result in a less stable inverse of the covariance (see Examples~\ref{example2} and~\ref{example}) and/or introduce bias (see \citealp{Bekker1994}). In such a case, the usual asymptotic approximation becomes less reliable. Therefore, in practice, a different approach is necessary to effectively utilize the information from a large number of instruments.} For example, one may suggest using  a few instruments selected by a certain criterion, see, e.g., \cite{donald2001choosing}, \violet{\cite*{Belloni}}, and \cite{Caner2014}. However, in the current study, we allow each element of a large number of instruments to be \textsl{nearly weak} (see \citealp{Newey2009a}). This consideration makes the consistency results of   variable selection procedures in the aforementioned papers  no longer valid even in linear models. Therefore, to handle the scenario with \textsl{``{many weak instruments}''}, we adopt a regularization scheme similar to those considered by \cite*{Hausman2011}, \cite{Carrasco2012}, \cite{Hansen2014}, \cite{Carrasco2015}, and \cite{Han2019}. Our regularization strategy is preferred to variable selection procedures for several reasons, which will be detailed in Section~\ref{sec:lasso}.% \textcolor{redby}{and it corrects the many instruments bias when every instruments are strong.} %We believe that all the theoretical developments made in this paper enable practitioners to better apply the endogenous binary response model.

  Alongside the regularization scheme, we consider estimators similar to the two-stage conditional maximum likelihood estimator (2SCMLE) proposed by \cite{rivers1988limited}, which is one of the most popular estimators for endogenous binary response models (see, e.g., \citealp{Wooldridge2010}).  Specifically, as  detailed in Section~\ref{sec:est}, we use the first-stage residual, obtained using a regularization method, as an additional regressor in the second stage. Then, we define two estimators, %each of which is associated with its own second-stage objective function. They 
 each of which is called the regularized conditional maximum likelihood estimator (RCMLE) and the regularized nonlinear least squares estimator (RNLSE). Their asymptotic properties are studied under a parametric assumption similar to that in \cite{rivers1988limited}. \label{r2p4}\violet{While this parametric approach may not be satisfactory from a theoretical standpoint, a large number of citations to \cite{rivers1988limited} suggest a clear preference for such an approach in the literature, at least when compared to semi- and non-parametric approaches. In addition, the parametric approach simplifies the asymptotic analysis of our estimators. This is partly because we do not need to estimate an infinite-dimensional parameter, such as the link function, under the parametric assumption. In this paper, we generalize the existing approach to allow for a possibly infinite-dimensional instrument, which is associated with an infinite-dimensional parameter in the first stage. Thus, a parametric assumption simplifies our analysis by reducing the number of infinite-dimensional parameters.}

\label{r2p1aa}\violet{We consider nearly weak instruments in the framework of \cite{Newey2009a}.} These instruments are commonly observed in empirical studies involving a binary dependent variable. For example,  \cite{Miguel2004}, \cite{Norton2008}, \cite*{Nunn2014}, \cite*{Frijters2009}, and \cite*{Bastian2018} report small values of the first-stage F test statistic that is known to be closely related to the weakness of instruments in the linear model (\citealp*{Staiger1997}; \citealp{Stocka}). Although this statistic is not a perfect measure of instrument strength in binary response models (e.g., \citealp*{frazier2020weak}), such small F test statistics suggest that their instrumental variables might be weak. However, in spite of its practical importance, only a few studies concern the issue of (nearly) weak instruments in endogenous binary response models. \cite{frazier2020weak} study a way to test the weakness of instruments using a distorted J test under a parametric assumption. \cite{Andrews2013} investigate asymptotic properties of the generalized method of moments (GMM) estimator under weak identification and apply their estimation approach to an endogenous probit model. However, even in these articles, the case with many weak instruments has not been  studied.

  This paper is technically different from   previous studies on many weak instruments. Specifically, we allow the instrumental variable to be function-valued; examples of function-valued random variables include  the age-specific fertility rate (\citealp{Florence2015}), the skill-specific share of immigrants (\citealp{seong2021}), and  the continuum of moments of an exogenous variable (\citealp{Carrasco2012}).  Functional instrumental variables have received recent attention in the econometrics literature, but have mostly been  used in linear models (e.g., \citealp{Carrasco2012}; \citealp{Florence2015}; \citealp{Carrasco2015a, Carrasco2015}; \citealp{Benatia2017}; \citealp{seong2021}). However, in contrast to the cases considered therein, our moment condition is nonlinear and not additively separable from the error term. Moreover, to address the endogeneity, we employ the control function approach. The limiting distributions of our estimators turn out to depend on the limiting distribution of the first-stage estimator that is given by an operator acting on a possibly infinite-dimensional space in our setting (see Section~\ref{sec:general}). These differences complicate our asymptotic analysis and make it technically distinguished from those of estimators for linear  models with many weak instruments (e.g., \citealp{Hansen2014}; \citealp{Carrasco2012}; \citealp{Carrasco2015}; \citealp{Benatia2017}) and   GMM estimators associated with an additively separable second stage (e.g., \citealp{Hausman2011}; \citealp{Andrews2013}; \citealp{Han2019}). %Given these differences, the current paper uses an approach to asymptotic analysis that differs from the previous studies.% Our approach is based on \citepos*{Chen2003} methodology, which is detailed in Section~\ref{sec:asy.pro}.  

This paper is relevant to practitioners who need to control for endogeneity, using many weak instruments when the outcome is binary. We revisit the work by \cite{Bastian2018} and use our estimators to study the effect of family income in childhood and adolescent years on college completion in young adulthood. Following the authors, we first use three measures of Earned Income Tax Credit (EITC) exposure to address the endogeneity of family income, and then employ interactions of the EITC measures with 47 state-of-residence dummies as additional instruments. \label{r1mp3}\violet{Similar to the bias of TSLS estimates toward the OLS observed in the literature on many weak instruments in linear models, the 2SCMLE proposed by \cite{rivers1988limited} tends to be biased toward the naive probit estimator that does not account for the endogeneity underlying the model. In contrast, our estimates appear to be robust to the inclusion of additional instruments.}

%The second empirical example examines the effect of children's development on maternal employment status. As described in \cite{Frijters2009}, a child's handedness is not directly correlated with parental socioeconomic status, but would provide some information on their brain development. Based on the arguments made in \cite{Frijters2009}, we use the children's handedness in writing as a basic instrumental variable and introduce additional instrumental variables by taking interactions between the basic instrumental variable and a set of exogenous dummy variables. In both examples, we provide evidence that the estimator in \cite{rivers1988limited} tends to be biased toward the naive probit estimator as the number of instrumental variables gets larger.

The paper is organized as follows. In Section~\ref{sec:model}, we discuss the model of interest. In Section~\ref{sec:est}, we propose our estimators and present their asymptotic properties. Section~\ref{sec:sim} summarizes Monte Carlo simulation results. The empirical application is given in Section~\ref{sec:emp}. Section~\ref{sec:con} concludes. All proofs of the results given in this paper  are provided in the Online Supplement. We include an additional simulation result and empirical application in the same supplement.

\section{The Model}\label{sec:model}
We consider the following model with independent and identically distributed (iid) data. \begin{equation}
y_{i}   = 1 \{  Y_{2i} ' \beta_0 \geq u_i    \}    \label{model:1eq}\quad\text{  and } \quad Y_{2i}  = \Pi_{n} Z (x_i)  + V_i , %\label{model:2eq} 
\end{equation}
where $1\{  A \}$ is the indicator function that takes 1 if $ A$ is true and 0 otherwise. Thus, the outcome variable $y_i$ is binary. The $d_e$-dimensional explanatory variable $Y_{2i}$ consists of endogenous and exogenous variables; if a row of $Y_{2i}$ is exogenous, the corresponding row of $V_i$ is equal to zero. \violet{The instrument $Z_i$ ($=Z(x_i)$) is  a function of $x_i \in \mathbb{R}^{d_x}$ satisfying $\mathbb E[u_i | x_i]  = 0$ and~$\mathbb E \left[V_i | x_i \right] = 0$.} 

The above model is similar to the standard endogenous binary response model considered by e.g., \cite{rivers1988limited}, \cite{Rothe2009}, and \cite{Blundell_Powell2004}, except that it allows for a general class of random variables, such as a function-valued random variable, to be considered as the instrument. %The instrument $Z_i$ is a  function of a $d_x$-dimensional exogenous variable $x_i$. %and a deterministic index $t$. 
\violet{The technical details of $Z_i$ for the general case will be discussed  in Section~\ref{sec:inst}.}\label{r1ptmp} Instead, in this section, we provide examples of various types of instruments that can be considered in our setting. First, if $x_i =(x_{1i}, \ldots, x_{d_x i})'$, then $Z_i $ may be given by $ x_i$, a standard $d_x$-dimensional vector of different instruments (see, e.g., Examples~\ref{example2} and~\ref{example}). If $x_i$ is a scalar, we may let $Z_i$ be $(1, x_i ,x_i^2,\ldots, x_i^K)'$.  If so, our first-stage estimator can be understood as a non-parametric estimator of $\mathbb E  [Y_{2i} | x_i ] $ in some cases, see \citet[Section 2.3]{Carrasco2012} and \cite{Antoine2014}. \violet{Lastly, in Section~\ref{sec:general}, we allow for a function-valued random variable to be considered as the instrument $Z_i$, as discussed in recent articles including \cite{Carrasco2012}, \cite{Florence2015} and \cite{Benatia2017}.} Examples of function-valued random variables include (but are not limited to) the density of $x_i$, the continuum of moments of $x_i$ (\citealp{Carrasco2012}), the rainfall growth curve (Example~\ref{example:miguel}), and the skill-specific share of immigrants (\citealp{seong2021}).

\label{r2p3}\label{r4p1}\label{r4p3}\violet{The potential weakness of instruments is characterized by the first-stage parameter $\Pi_n$, whose magnitude may depend on $n$. This will be further detailed in Assumption~\ref{ass1:finite}.  A similar setting in linear models can be found in \cite{Chao2005}, \cite{Hansen2014}, and \cite{Carrasco2015}, among many others. The parameter $\Pi_n$ will be represented by a matrix or a linear operator with rank $d_e$, depending on the instrumental variable used in the analysis. That is, regardless of the dimension of $Z_i$, $\Pi_n Z_i$ is a $d_e$-dimensional vector that  combines information from the instrumental variable of a large, and possibly infinite, dimension. This will be further detailed in Section~\ref{sec:est}.}

Let \label{r2pppp} \violet{$\eta_i =u_i - \mathbb E[u_i |x_i, V_i] = u_i+ V_i ' \psi_0$}. Then, the outcome equation can be written as follows:\begin{equation}
	y_{i}  = 1 \{  Y_{2i} ' \beta_0 + V_i ' \psi_0 \geq \eta_i  \} . \label{model:3eq}
\end{equation}
\violet{If $\eta_i |(x_i , V_i )\sim \mathcal N(0,\sigma_\eta ^2)$\footnote{\label{r2p6}\violet{The distributional assumption is satisfied if $(u_i, V_i')'|x_i \sim \mathcal N(0, \mathbf{\Sigma}  )$, as in \cite{rivers1988limited}.}} and the variable  $V_i$ is available, \eqref{model:3eq} can be estimated by maximum likelihood estimation (MLE).\label{r1p1}} In this regard, $V_i$ is often called the control function in the literature. However, in practice, $V_i$ is not observable and has to be replaced by its estimate.  \cite{rivers1988limited} suggested replacing it with the residual from  first-stage linear regression, denoted  $\widehat{V}_i$. Then, the parameters in \eqref{model:3eq} are estimable by maximizing the conditional likelihood of $y_i | (x_i, \widehat V_i)$ under a parametric assumption. This approach has become one of the most popular estimation methods for endogenous binary response models and was extended by \cite{Rothe2009} to the semi-parametric case. %where the conditional distribution of $u_i|(x_i, V_i)$ is unknown.

Although the aforementioned estimators generally have good asymptotic properties, it is not likely to be the case in our setting. This is because, to ensure such a good property, the difference between $\widehat{V}_i$ and $ V_i$ should be asymptotically negligible. %, but  %, or, as a primitive condition, the first-stage least square estimate for the coefficient of $Z_i$ converges to its population counterpart with a $n^{1/2}$ rate. In fact, these conditions
%this condition is not  satisfied in our case. To see this in detail, 
\violet{This requirement is not likely to be satisfied, especially when $Z_i$ is of a large dimension or has a nearly singular covariance matrix,}\label{r1mp5a} see, e.g., Examples~\ref{example2} and \ref{example}. In these cases, the first-stage estimator may be asymptotically biased, resulting in $\widehat{V}_i$ being asymptotically different from $V_i$. Another concern on \citepos{rivers1988limited} approach is that the asymptotic distribution of their estimator depends on the asymptotic distribution of the first-stage estimator. \label{r1p2}\violet{In Section~\ref{sec:general}, the first-stage parameter and its estimator are given by linear operators from a possibly infinite-dimensional Hilbert space to $\mathbb R^{d_e}$, whereas in \cite{rivers1988limited}, they are given by matrices. This difference makes it difficult to  apply their asymptotic results  in our setting. }%Therefore, the asymptotic distribution cannot be obtained as in \cite{rivers1988limited}, which prevents a naive application of their approach.  

 Below, we provide examples where  \citepos{rivers1988limited} approach may not be appropriate due to the nearly singular covariance of $Z_i$, even with a  large sample size. \violet{In the examples, the instrument $Z_i$ is given by $(W_i' , \widetilde{Z}_i ' )'$, where $W_i$ and  $\widetilde{Z}_i$ denote the included and excluded exogenous variables, respectively.  $Y_{2i}$ in \eqref{model:1eq} is understood as a vector that includes $W_i$ and endogenous variables (whose notations are specific to each context).}\label{r1mp5} We use the condition number (the ratio of the largest to smallest eigenvalues) of a covariance matrix as a measure of how close the matrix is to being singular. 
\begin{example}[]\label{example2}\normalfont
	\cite{Bastian2018} studied the impact of family income in childhood and adolescent years on academic achievement in early adulthood using the linear probability model. However, from a theoretical perspective, their result may be improved by using a binary response model; this is because the linear model often produces (i) the predicted probability outside the interval $[0,1]$, (ii) constant marginal effects of explanatory variables on the probability of outcome occurring, regardless of the values of the explanatory variables, and (iii) less efficient estimation results due to heteroscedasticity. Moreover, when endogeneity is present, theoretical results developed for the linear model do not always remain valid in the binary response model (\citealp{li2019binary}; \citealp{frazier2020weak}).  Hence, as an alternative, one may consider the following model.\begin{equation*} y_i =  {1} \{  \beta_0 + I_{i}'\beta_1  + W_{1i} ' \beta_2 + W_{2i,t} '\beta_3 \geq u_i  \} ,\quad I_{i} = \pi_{0} + \widetilde Z_i ' \pi_1 + W_{1i} ' \pi_{2} + W_{2i,t} ' \pi_{3} + V_{i} , 
	\end{equation*}
	where $y_i$ is 1 if individual $i$ is a college graduate and 0 otherwise. $I_i = ( I_{i,\text{(0-5)}}, I_{i,\text{(6-12)}},$ $I_{i,\text{(13-18)}})'$, where $I_{i,(a)}$ is the family income for individual $i$ at age interval $a$. $W_{1i}$ and $W_{2i}$ are vectors of exogenous variables detailed in Section~\ref{sec:emp}. As instruments, the authors suggested using measures of EITC exposure for individual $i$ in three age intervals: 0-5, 6-12, and 13-18. That is, $\widetilde Z_i = (\text{EITC}_{i,(\text{0-5})},\text{EITC}_{i,(\text{6-12})}\text{EITC}_{i,(\text{13-18})})'$. In Section~\ref{sec:emp}, we use additional instruments generated by interacting the EITC measures with state-of-residence dummies. This approach is expected to improve estimation efficiency by reducing variances of estimators. It also enables us to account for heterogeneous effects of EITC exposure on family income that may be influenced by the state of residence, thereby mitigating  potential bias attributed to, e.g., state-specific costs of living.  The maximum number of instruments is 144, and the sample size is 2,654. \label{r1p3a}The largest and smallest eigenvalues of the sample covariance are approximately 1.94 and $9.35 \cdot 10^{-18}$, indicating a very large condition number. \violet{This large condition number suggests that the inverse of the covariance is less stable, and the existing asymptotic approximation using this inverse may not perform well. Here, regularization will improve estimation results by stabilizing the inverse of the covariance.}%I is suggested.  Hence, in spite of the considerably large sample size, the sample covariance of instruments has a very large condition number, which suggests that the matrix is nearly singular. } 
\end{example}
\begin{example}[]\label{example}\normalfont
	\cite{Angrist1991} has been revisited in numerous articles, particularly in the many-instruments literature (e.g., \citealp{donald2001choosing}; \citealp{Carrasco2012}; \citealp{Hansen2014}). Suppose we are interested in the effect of educational achievement on employment probability, rather than  on earnings as studied by \cite{Angrist1991}. One may consider the model below, where  $y_i = 1$ if individual $i$ is employed for a full year and 0 otherwise.\begin{equation*} 
	y_i  = 1\{ \beta_0 + \text{Educ}_i \beta_1 + W_i ' \beta_2 \geq u_i  \},\quad\quad 
	\text{Educ}_i  = \widetilde Z_{i} ' \Pi_1 + W_i ' \Pi_2 + V_i.
\end{equation*}
The variable $\text{Educ}_i$ is the years of schooling of individual $i$. $W_i$ consists of a constant, 9 year-of-birth dummies (YoB), and 50 state-of-birth dummies (SoB). As in \cite{Angrist1991}, let $\widetilde Z_i$ be the 180-dimensional vector  of 3 quarter-of-birth dummies and their interactions with YoB and SoB. The condition number of the covariance of  $ (\widetilde{Z}_i ' , W_i ') ' $ is approximately 4,941, even after standardization. This suggests near singularity, even with a large sample size of 329,509.
%the condition number of the matrix, which is given by the ratio of the largest and smallest eigenvalues, is often used as a measure of how singular a matrix is. In this particular example, the number is approximately 4998. This implies that the sample covariance matrix of $Z_i$ is nearly singular, and thus, in this example, an estimation result based on \citepos{rivers1988limited} approach may not be accurate.
\end{example}

\subsection{Regularization with many weak instruments\label{sec:lasso}}

We use regularization to resolve the aforementioned issues related to many, and nearly weak, instruments. %Regularization method is not novel especially to the econometric literature on many instruments, see, e.g., \cite{Carrasco2012}, \cite{Hausman2011}, \cite{Belloni}, \cite{Hansen2014}, and \citeauthor{Carrasco2015} (\citeyear{Carrasco2015}, \citeyear{Carrasco2015a}). 
There are various regularization methods available in the literature. A popular approach is reducing the number of instruments by (i) choosing the optimal number of instruments that minimizes a model selection criterion (e.g., \citealp{donald2001choosing}) or (ii) using a variable selection procedure such as Lasso (e.g., \citealp{Belloni}) or Elastic-net (e.g., \citealp{Caner2014}). This approach works particularly well if the signal from instruments is (i) concentrated on a few of them (called the sparsity condition) and (ii) strong enough to consistently select a set of instruments with non-zero first-stage coefficients. However, despite its popularity, this is not always the best approach to use many instruments. For example, the sparsity condition may not hold when the variables of interest are highly correlated with each other. In such a case, there could be a model selection mistake of choosing irrelevant instruments. \label{r1p4}\violet{The selection mistake is not desirable especially in finite samples, because the resulting estimators could be less efficient. For  additional concerns related to the selection mistake, see \citet[p.\ 293]{Hansen2014}.} Hence, it would be advisable to avoid using the variable selection procedure if such a selection mistake is likely to occur. In this regard, the variable selection procedure may not be appropriate in our setting, because (i) we consider  many (nearly) weak instruments and (ii) in such a scenario, variable selection procedures tend to select no instrument or almost randomly choose a set of instruments (see, e.g., \citealp{Belloni}).
%as an example, the simulation results in \cite{Belloni} show that their Lasso procedure tends to select no instrument or almost randomly choose a set of instruments in the presence of many weak instruments, which is obviously not desirable in practice.\footnote{\cite{Belloni} proposes an inference procedure that is valid and readily available in this case.}  
Furthermore, as in Examples~\ref{example2} and~\ref{example}, if instruments are given by a set of dummies, it is not recommended to choose a few of them; this is because (i) a selected set of instruments will depend on how the variables are encoded and (ii)   the categorical variables are not likely to be represented by only a few of them. %A more serious concern is that, as is noted by \cite{Hansen2014}, if within-sample correlations between the instruments and the first-stage error term are high, an instrumental variable with zero population coefficient may be chosen. This, in turn, may result in the violation of the exogeneity condition. 

In this paper, we regularize the covariance of instruments. Our regularization scheme, formally introduced in Section~\ref{sec:est}, is preferred to the above-mentioned variable selection procedures for several reasons. Firstly, our approach does not need a sparsity condition, and secondly,  is free from the aforementioned concern on selection mistakes. Moreover, in practice, valid instruments are likely to be (nearly) weak. In this case, it has been shown that the asymptotic properties of the estimators obtained by regularizing the covariance of instruments remain valid (e.g., \citealp{Hausman2011}; \citealp{Hansen2014}; \citealp{Carrasco2015a}), while  the variable selection procedures no longer produce consistent estimation results. In contrast to the popularity of regularizing the covariance in linear models (\citealp{Carrasco2012}; \citealp{Hansen2014}; \citealp{Carrasco2015}), its application to binary response models has not been fully explored. %The current paper studies how regularization resolves the issue of many weak instruments in the binary response model.

\section{Estimator and Asymptotic Properties\label{sec:est}}

  We first consider the simple case $Z_i \in \mathbb R^{d_z}$ for $d_e \leq d_z <\infty$. In this case, regularization is useful if the covariance of $Z_i$ is nearly singular, as in  Examples~\ref{example2} and~\ref{example}.  In Section~\ref{sec:general}, we further extend the discussion to accommodate function-valued $Z_i$. In that case, not only is regularization required to address the singularity of the covariance, but a novel asymptotic analysis is also needed to establish the asymptotic properties of our estimators. In the following sections, $Z_i$ is assumed to be centered without loss of generality.
\subsection{Simple case: a finite-dimensional instrumental variable \label{sec:finite}}
\subsubsection{Estimator \label{sec:est1:finite}} 
Let $\mathcal K$ and $\mathcal K_n$ denote the covariance of $Z_i$ and its sample counterpart where  $\mathcal K = \mathbb E[Z_iZ_i']$ and $\mathcal K_n =n^{-1}\sum_{i=1} ^n Z_i Z_i' $. In Section~\ref{sec:general}, we will adapt these definitions for the general case.   
\violet{\label{r1p5}Given the model \eqref{model:1eq} and the exogeneity of $Z_i$ with respect to $V_i$ ($\mathbb E[V_i Z_i']=0$), we have the following first-stage moment condition:}\begin{equation}
\violet{\mathbb E \left[  Y_{2i} Z_i '  \right]= 	\Pi _{n} \mathbb E \left[Z_i   Z_i' \right]  =\Pi_n \mathcal K .} \label{mod2:eq1:finite}
\end{equation}
 In many econometric studies, it has been assumed that the covariance $\mathcal K$ is invertible to solve \eqref{mod2:eq1:finite}. However, this assumption is not likely to be satisfied when the dimension of $Z_i$ is large, as in Examples~\ref{example2} and~\ref{example}. Hence, instead of the inverse of $\mathcal K$, we consider its regularized inverse. Here, we focus on Tikhonov regularization (in Section~\ref{sec:general}, we extend it to various regularization methods satisfying certain conditions). Specifically, for a  regularization parameter $\alpha >0$ and the identity matrix of dimension $d_z$, denoted $\mathcal I_{d_z}$,  $\mathcal K_\alpha ^{-1}$ denotes the regularized inverse of $\mathcal K$ such that 	\begin{equation*}    \mathcal K_\alpha ^{-1} = \mathcal K (\mathcal K^2 + \alpha \mathcal I_{d_z})^{-1}    . \end{equation*} The regularized solution to \eqref{mod2:eq1:finite} (denoted $\Pi_\alpha$) and its sample counterpart (denoted $\widehat \Pi_{\alpha}$) are given by \begin{equation}
	\Pi_{\alpha} = \mathbb E  [  Y_{2i}    Z_i'   ] \mathcal K_\alpha  ^{-1}  \hspace{.5em}\text{ and }\hspace{.5em}\widehat \Pi_{\alpha} = (n^{-1} \sum_{i=1} ^n Y_{2i}Z_i'  ) {\mathcal K}_{n\alpha} ^{-1} =  (n^{-1} \sum_{i=1} ^n Y_{2i}Z_i'  ) \   {\mathcal K}_{n }  (  {\mathcal K}_n ^2 +\alpha \mathcal I_{d_z})^{-1}   .  \label{eq:first-est:finite}
\end{equation} 
Using $\widehat{\Pi}_\alpha$ in \eqref{eq:first-est:finite}, we compute the predicted value of $Y_{2i} $, denoted $\widehat{\gamma}_{i,\alpha}$, and the first-stage residual  $\widehat{V}_{i,\alpha}$, i.e., $\widehat{\gamma}_{i,\alpha} = \widehat \Pi_{\alpha} Z_i $ and $\widehat{V}_{i,\alpha} = Y_{2i} -\widehat \Pi_{\alpha} Z_i$. We use them as regressors in the second stage for mathematical convenience.

%; specifically, as is detailed in Theorem~\ref{prop4:finite}, the coefficients of $\gamma_{i,\alpha}$ and $V_{i,\alpha}$ are respectively $\beta_0$ and $\beta_0 + \psi_0 $

%In the following, the second-stage objective function is considered as a function of  $\widehat{g}_{i,\alpha} = (\widehat{\gamma}_{i,\alpha}', \widehat V_{i,\alpha}')'$ and  $\theta_0  =  ( \beta_0  ' , \beta_0 '+ \psi_0 ' )'$, because  $\beta_0$ and $\beta_0+\psi_0$ have different convergence rates (see ).    %$g_{i}$ and $\widehat g_{i,\alpha}$ can be understood as a function of $\Pi $ evaluated at $\Pi_{\star}$ and $\widehat \Pi_\alpha$ respectively, i.e.\ $g_{i} = \mathcal G_i  (\Pi_{\star} )$ and $\widehat g_{i,\alpha} = { \mathcal G}_i  (\widehat\Pi_\alpha )$, and thus, $g_{i}$ would be understood as a population counterpart of $\widehat g_{i,\alpha}$. 

Suppose $(u_i -\mathbb E[u_i|x_i, V_i])| (x_i, V_i) = u_i +V_i'\psi_0 | (x_i, V_i)  \sim  \mathcal {N} (0,\sigma_\eta ^2)$.  Here, what benefits us is not the conditional normality of $u_i$, but rather the fact that the endogeneity of $Y_{2i}$ affects the second stage only through the conditional expectation of $u_i$. The normality assumption is employed only to facilitate our discussion, and may be replaced by another distributional assumption. We then normalize the conditional variance $\sigma_\eta ^2 $ to $1$ to achieve the identification of the second-stage parameters, see \cite{Manski1988} for details. Lastly, let $ g_{i} =  ( \gamma_{i} ' ,  V_{i} '  )'  = ( (\Pi_n Z_i)' , (Y_{2i} - \Pi_n Z_i)')'$,  $\widehat{g}_{i,\alpha} = (\widehat{\gamma}_{i,\alpha}', \widehat V_{i,\alpha}')'$ and  $\theta  =  ( \beta  ' , \beta '+ \psi ' )'$.   In contrast to $\widehat{g}_{i,\alpha}$, the vector $g_i$ is not defined with the regularized inverse. The RCMLE (denoted $\widehat{\theta}_M$)  and the RNLSE (denoted $\widehat{\theta}_N$) are the estimators of the true coefficients $\theta_{0M}$ and $\theta_{0N}$, respectively.  $\widehat \theta_{ j }$ and $ \theta_{ 0j }$ are defined as follows: for $j \in \{ M,N\}$, \begin{align}
\widehat \theta_{ j } := \underset{\theta \in \Theta }{\arg\max} \mathcal Q_{jn}  (\theta , \widehat g_{i,\alpha} ) \quad\quad	&\text{and}\quad\quad  \theta_{0j} := \underset{\theta \in \Theta}{\arg\max} \mathcal Q_j  (\theta , g_{i} ), \label{model:5eq}\\
\text{where }
	\mathcal Q_{jn}  ( \theta , \widehat  g_{i,\alpha}  ) := \frac{1}{n} \sum_{i=1} ^n m_j  (\theta , \widehat g_{i,\alpha} )   \quad\quad&\text{and} \quad\quad\mathcal Q_j  ( \theta , g_{i}  ) :=\mathbb E  [   m_{j}  ( \theta , g_i  )  ], \label{model:eq3}
\end{align}
and the function $m_{j}  (\theta , g_i  ) $ is  $ y_i\log\Phi (g_i ' \theta ) +  (1-y_i )\log (1-\Phi (g_i '\theta ) )$ (resp.\ $  -\frac{1}{2}  (y_i - \Phi (g_{i} ' \theta ) )^2$) if $j=M$ (resp.\ $j=N$). The above  expectation is taken with respect to the distribution of $g_i$ and $y_i$ for sample size $n$, but we suppress the index $n$ unless it is necessary. Under the identification conditions in Section~\ref{sec:asy.pro:finite}, we have  $\theta_{0M} = \theta_{0N} = \theta_0$.

\begin{remark}\normalfont \label{rem1}
   $\mathbb E  [y_i | x_i ,V_{i,\alpha} ] $ is in general different from $\mathbb E  [y_i | x_i ,V_{i} ] $. To see this in detail, let $\gamma_{i,\alpha} = \Pi_\alpha Z_i$, and $V_{i,\alpha} = Y_{2i} - \gamma_{i,\alpha}$. Then, under the aforementioned assumptions, we have
	\begin{equation}
		\mathbb E  [y_i | x_i ,V_{i,\alpha} ] = \Phi  (  \gamma_{i,\alpha} ' \beta_0 + V_{i,\alpha} '  (\beta_0 + \psi_0)-   \psi _0' \Pi_{n}(\mathcal I_{d_z} - \mathcal K_\alpha^{-1} \mathcal K)   Z_i    ). \label{model:6eq}
	\end{equation}
	  If $\alpha>0$ is fixed, the conditional expectation in \eqref{model:6eq} depends not only on the error $V_{i,\alpha}$ but also on a bias term associated with $\Pi_{n}(\mathcal I_{d_z} - \mathcal K_\alpha^{-1} \mathcal K)  $. %Should $\alpha$ be fixed, \citepos{rivers1988limited} estimator becomes inconsistent even in the case where $\mathcal K$ is invertible.
	Thus, to mitigate its potential effect on the asymptotic properties of our estimators, $\alpha$ is assumed to shrink to zero as $n \hto \infty$ throughout the paper.
\end{remark}

\begin{remark}\normalfont
	In this paper, we apply the regularization approach to the conditional maximum likelihood and nonlinear least squares (NLS) estimation methods. These methods are chosen because of their ease of implementation and popularity in empirical research. While other estimation methods, such as the instrumental variable probit and the limited information maximum likelihood in \cite{rivers1988limited}, are also applicable, we do not pursue them further in the current paper.
\end{remark}

\subsubsection{Asymptotic Properties \label{sec:asy.pro:finite}}

This section presents the asymptotic properties of our estimators. Hereafter, we use $\lambda_{\min}(A)$ to denote the minimum eigenvalue of a matrix $A$. %In this section,  $\langle a_1, a_2 \rangle $ (resp.\ $\Vert a_1 \Vert$) for any $a_1, a_2 \in \mathbb R^{n}$ denotes $a_1'a_2$ (resp.\ $(a_1'a_1)^{1/2}$), and $\Vert \cdot \Vert_{\HS}$ indicates $ (\tr(A'A))^{1/2}$ for a matrix $A$.\footnote{\violet{When the argument of $\Vert \cdot \Vert_{\HS}$ is a linear operator from, it denotes its Hilbert-Schmidt norm, i.e., for a linear operator acting on $\mathcal H$, $ \Vert \mathcal A \Vert _{\HS} = (\sum_{j=1} ^\infty \langle\mathcal A ^{\ast }\mathcal A \zeta_{j}, \zeta_{j}\rangle_{\mathcal H}^2)^{1/2}$ where $\{  \zeta_j\}_{j \geq 1}$ is a set of orthonormal basis of $\mathcal H$.}\label{r1mp2c}}
\begin{assumption}\label{ass1:finite}  For $\Pi_{n}$ satisfying  \eqref{mod2:eq1:finite}, there exists a bounded matrix $\Pi_0 : \mathbb R^{d_z} \hto \mathbb R^{d_e}$ such that $\Pi_{n} = {\Lambda}_n \Pi_0 /\sqrt{n}$ where ${\Lambda}_n = \widetilde {\Lambda} \text{diag}(\mu_{1n} , \ldots, \mu_{d_{e}n})$. The matrix $\widetilde {\Lambda}: \mathbb R^{d_e} \to \mathbb R^{d_e} $ is bounded, and $\lambda_{\min}(\widetilde {\Lambda} \widetilde {\Lambda} ' )$ is bounded away from zero. For each $j$, either $\mu_{jn} = \sqrt{n}$ (strong) or $\mu_{jn} / \sqrt{n} \hto0$ (weak)\violet{, and $\mu_{m,n} = \underset{1 \leq j \leq d_{e}}{\min} \mu_{jn} \hto +\infty$ as $n \hto+ \infty$}.
\end{assumption}
  Hereafter, we let $f_i = \Pi_0 Z_i$.   %In Section~\ref{sec:general}, we allow the instrument $Z_i$ to be function-valued. In that case, $f_i$ can be viewed as the (infeasible) optimal instrument that summarizes all the signals from the infinite-dimensional random variable $Z_i$.  
  \label{r2p1a} \violet{The above assumption is similar to Assumption~1 in \citet[p.\ 698]{Newey2009a}, under which  many weak moment asymptotics are studied for GMM models with linear and nonlinear moment conditions. Considering that  the GMM class is general enough to encompass the MLE and the NLS (\citealt[p.\ 2116]{NEWEY1994}), this framework would be suitable  for studying  consequences of many weak instruments in our setting.}  A similar assumption can be found in the literature on many weak instruments in linear models, including  \cite{Chao2005}, \cite{Hansen2014}, and \cite{Carrasco2015}, to mention only a few. 
   
   Assumption~\ref{ass1:finite} essentially  rules out weak instruments in the sense of  \cite{Staiger1997}. Specifically, under the assumption, each row of $f_i$ could be either (nearly) weak or strong, depending on the value of $\mu_{jn}$. If $\mu_{jn} = \sqrt{n}$ for all $j$, then every element of $f_i$ is strong, and regularization will play a role of reducing the many instruments bias discussed in \cite{Bekker1994}. The above assumption is also related to the nearly weak moment condition in \cite{AR2009}. %We complement the article by  exploring the use of possibly infinite-dimensional $Z_i$ in the binary response model.

\begin{assumption}\begin{enumerate*}[(i)]
	\item $\{ u_i, V_{i}', x_i'   \}_{i=1} ^n$ is iid and \violet{$(u_i, V_i')'|x_i $ follows the joint normal distribution with mean zero and finite positive definite covariance $\mathbf \Sigma_{uV}$, and satisfies that $\mathbb E[u_i|x_i, V_i] = -\psi_0'V_i$ and $\Var(u_i  |x_i, V_i) =1$;} \label{r2p6a}\label{ass32}
	%	\item The minimum eigenvalue of  is bounded away from zero;\label{ass33}
		\item  \violet{There exists $c>0$ such that   $\mathbb E [\Vert Z_i \Vert^2]\leq c $ and $\lambda_{\min}(\mathbb E  [f_i f_i '  ]  ) \geq 1/c$.}\label{ass35}
	\end{enumerate*}\label{ass3:finite}
\end{assumption}
\begin{assumption}\label{ass4}
	The parameter space of $\theta$, denoted $ \Theta$, is compact. $\theta_{0}$ is the unique point satisfying Assumption~\ref{ass3:finite}.\ref{ass32} and the model \eqref{model:1eq}, and is in the interior of $\Theta$. 
\end{assumption}
\label{r2p6b}  \violet{Assumption~\ref{ass3:finite} is similar to Assumption~1 in \cite{rivers1988limited}. As they mentioned, this assumption is stronger than what we need to control for the endogeneity of $Y_{2i}$, and the assumption can be weaken by assuming that $(u_i-\mathbb E[u_i|x_i, V_i])| (x_i, V_i) =(u_i+\psi_0 ' V_i)| (x_i, V_i)   \sim\mathcal N  ( 0, 1  )$;} that is, the endogeneity of $Y_{2i}$ arises  through $\mathbb E[u_{i}| x_i,V_i]$, which enables us to address the endogeneity via the inclusion of $V_i$. We normalize the conditional variance of  $u_i| (x_i,V_i )$ to 1 for identification (see \citealp{Manski1988}). A different normalization  could have been chosen, but   we select the current normalization because it makes the distribution of  $u_i |(x_i,V_i )$  free from nuisance parameters. 
  
   \violet{Assumptions \ref{ass3:finite} and \ref{ass4} ensure the identification  of $\theta_0$ (\citealt[Section 2.2]{NEWEY1994}), and similar conditions can be found in \cite{frazier2020weak}  which concerns the endogenous binary response model with weak instruments.}\label{r2p2}  Under the conditions, $\theta_{0M} = \theta_{0N} = \theta_0$.

We then discuss a condition on $\Pi_0$; in the condition below, $\ran \mathcal K$ denotes the range of $\mathcal K$.% and $\pi_{\ell,0}$ denotes $\ell$-th row of $\Pi_0$. %Therefore, for any $\mathcal A:\mathcal H \to \mathbb R^{d_e}$, $\mathcal A ^{\ast } e_\ell$ turns out to be an element in $\mathcal H$.  
 \begin{assumption}\label{ass7:finite} Let $\pi_{\ell,0}$ be  the $\ell$th row of $\Pi_0$. Then, for $\ell = 1,\ldots, d_e $, $\pi_{\ell,0}\in \ran \mathcal K$.
\end{assumption} 
  As  detailed later in Section~\ref{sec:general}, a primitive condition for Assumption~\ref{ass7:finite} is that $(\pi_{\ell,0} ' \varphi_j  )^{2}$ decreases at a rate slower than the decaying rate of the square of eigenvalues of $\mathcal K$, where $\{\varphi_j\}_{j \geq 1}$ denotes the eigenvectors of $\mathcal K$. This assumption is crucial for identifying     $\Pi_0$ without assuming the invertibility of $ \mathcal K$ (see \citealp{Carrasco2007}; \citealp{Benatia2017} for details) and for controlling the bias discussed in Remark~\ref{rem1}. \label{r1p12}\violet{Moreover,  point identification of $\Pi_0$ simplifies our asymptotic analysis by providing an explicit form of  the regularization bias and the asymptotic distribution of the first-stage estimator, which are directly related to the asymptotic distributions of our estimators.} 
 
The theorem below presents the consistency of $\widehat{\theta}_j$ under the aforementioned assumptions.
\begin{theorem}\label{thm1:consistency:finite}\normalfont
	Suppose that $\alpha \hto 0$ and $\mu_{m,n} ^2 \alpha^{1/2}\hto \infty$ as $n \hto \infty$, and Assumptions \ref{ass1:finite} to \ref{ass7:finite} are satisfied. Then, $
 	\widehat \theta_{ j} - \theta_{0}  \pto 0  
	$ as $n \hto \infty$ for $j \in\{ M,N \}$.
\end{theorem}
  The conditions on $\alpha$ and $\mu_{m,n}$ in Theorem~\ref{thm1:consistency:finite} are related to the weakness of $f_i$. If $\alpha = O(n^{-a})$ and $\mu_{m,n} = O(n^{\ell})$ for some positive $a$ and $\ell\leq 1/2$, then the condition on $\mu_{m,n}^2\alpha^{1/2} $ suggests  $0<a< 2$. This is a necessary condition for $\alpha$ to ensure the consistency of $\widehat{\theta}_j$. A similar condition can be found in \citet[Proposition 2]{Carrasco2015} and \citet[Assumption 2]{Hansen2014}, where in the latter, the inverse of the number of instruments plays a role similar to $\alpha$ in the current study. The convergence rate of $\alpha$ allows us to control for the bias associated with the nearly weak signal of $f_i$ and the regularized inverse simultaneously.  %We note that the consistency result is obtained simply by showing that $n^{-1} \sum_{i=1} ^n \Vert \widehat{g}_{i,\alpha} - g_{i} \Vert ^2 =o_p(1)$, see Appendix~\ref{sec:pf} for details. %As we mentioned, in practice, the normal and logistic distribution functions are widely used as the conditional distribution of $u_i|(x_i, V_i)$. Since their inverse Mill's ratio is Lipschitz continuous, the consistency result remains valid for them.

   Next, we discuss the asymptotic distribution of $\widehat \theta_j$.  We follow  \cite{rivers1988limited} and address the endogeneity by using  the estimated control function. As a result, the limiting distribution of $\widehat\theta_j$ turns out to depend on that of the first-stage estimator. However, in our setting, obtaining the limiting distribution of the first-stage estimator is not straightforward because of the use of a regularized inverse and the bias discussed in Remark~\ref{rem1}. %This in turn complicates the application of \citepos{rivers1988limited} approach. 
  %Moreover,  the use of a regularized inverse  also requires us to address the regularization bias in Remark~\ref{rem1}. 
  Although a similar scenario has been considered for linear models (\citealp{Carrasco2012}; \citealp{Hansen2014}; \citealp{Carrasco2015a,Carrasco2015}), the asymptotic approaches therein cannot be applied in our context because of some distinct features of our models, such as nonlinearity and the use of the control function approach. To resolve these issues and to encompass the general case in Section~\ref{sec:general}, we adapt the asymptotic approach in \cite*{Chen2003}.  To this end, we employ the following assumption, in which  a.s.n.\  denotes   almost surely for $n$ large enough.
\begin{assumption}\label{ass9:finite} \begin{enumerate*}[(i)]
	%	  \item \label{ass9-1}$\mathbb E  [  \Vert V_i  \Vert ^4 | x_i ] <\infty $;
		    \item \label{ass9-2} For all $\theta \in \Theta$, $0< \inf_i\Phi (g_i ' \theta )\leq \sup_i\Phi (g_i ' \theta )<1$ a.s.n.;
		    \item  There exists $c>0$ such that  $  \sup_i \Vert  f_i\Vert \leq c  $ a.s.n.\ %, $ \mathbb E[\Vert V_i\Vert ^4 |x_i] \leq c$ 
		    and  $ \mathbb E[\Vert Z_i\Vert ^4 ] \leq c$.%, $\sup_i \Vert {\Lambda}_n ^{-1} \widehat \Pi_{\alpha} Z_i \Vert \leq c$ 		 
 	\end{enumerate*}
\end{assumption} 
Assumption~\ref{ass9:finite}  is necessary to  bound a particular quantity in our proof, but can be relaxed by using a trimming term  as in \cite{Rothe2009}. Thus, the condition may not be restrictive in practice.

{We then present the asymptotic normality result. Let $\mathcal S_n^{-1}$ denote $\text{diag}(\sqrt{n} {\Lambda}_n ^{-1} , \mathcal I_{d_e}) $ and $g_{0i}= \mathcal S_n ^{-1} g_i = (f_i ' , V_i ' )'$. Moreover, let $\sigma_{\psi_0} ^2=\psi_0 ' \mathbb E[V_iV_i'|x_i]\psi_0 =  \psi_0 ' \mathbf \Sigma_V\psi_0$,  $m_{1M} ^2 (\theta, g_i)  = - \dot{m}_{2M} (\theta, g_i)$, $ \dot{m}_{2M} (\theta, g_i) =  - \left((y_i - \Phi (g_i'\theta))\phi  (g_i ' \theta )/\left( \Phi (g_i ' \theta )(1-\Phi (g_i ' \theta )) \right)\right)^2 $,   $
m_{1N} ^2 (\theta , g_i) =$ $ \left((y_i - \Phi(  g_i ' \theta ))\phi\left(g_i ' \theta\right)\right)^2 $, and $ \dot{m}_{2N} (\theta, g_i) =\phi ^2  (g_i ' \theta )$.  The matrix $\mathcal C_j $ is the solution to $\mathcal K^{1/2}\mathcal C_j' = \mathbb E[  \dot{m}_{2j}(\theta_0,g_i)Z_i g_{0i}'    ]$, and    $e_{\ell}$ denotes  the unit vector whose $\ell$th element is equal to one. Lastly,  $\Gamma _{1,0j}=\mathbb E[  \dot{m}_{2j}(\theta_0,g_i)g_i g_{i}'    ]$ and $\mathcal J_j =\mathcal J_{1,j} + \mathcal J_{2,j}$ where $\mathcal J_{1,j} =  \mathbb E  [    {m}_{1j} ^2 ( \theta_{0} , g_i  )      g_{0i} g_{0i} '    ] $ and $\mathcal J_{2,j}$ is a $2d_e \times 2d_e$ matrix whose $(\ell,k)$th element is given by $\sigma_{\psi_0} ^2 e_\ell ' \mathcal C_j \mathcal C_j' e_k$.}    %where   \violet{We further let $m_{1j} ^2  ( \theta,  g_{i} )$ be $\dot{m}_{2M} (\theta , g_i )$  $\left(y_i - \Phi\left(  g_i ' \theta \right)\right)^2 \phi^2\left(g_i ' \theta\right)$ }   The matrix $\mathcal C_j $ is the solution to $\mathcal K^{1/2}\mathcal C_j' = \mathbb E[  \dot{m}_{2j}(\theta_0,g_i)Z_i g_{0i}'    ]$ where  $\dot{m}_{2j} (\theta_0, g_i)$ is  $\phi ^2  (g_i ' \theta_{0} )/ (\Phi (g_i ' \theta_{0} )\Phi (-g_i ' \theta_{0} ) )^2$ (resp.\ $$) if $j=M$ (resp.\ $j=N$). \violet{The matrix $\Gamma _{1,0j}$ is given by  $\mathbb E[  \dot{m}_{2j}(\theta_0,g_i)g_i g_{i}'    ]$. } %We note that $\mathcal C_j$ always exists regardless of the dimension of $Z_i$, see  \citet[Theorem 1]{baker1973joint}.

\begin{theorem}\label{prop4:finite}\normalfont
	Suppose that $\alpha \to 0$, $\mu_{m,n} ^2 \alpha\to\infty$ and $\sqrt{n} \alpha \to0$ as $n \hto\infty$, and  Assumptions~\ref{ass1:finite} to~\ref{ass9:finite} hold.  Then,  $\mathcal W_{0j} ^{-1/2} \sqrt{n} \mathcal S_n ^\prime \left(\widehat \theta_j - \theta_{0}  \right) \dto \mathcal N \left( 0,\mathcal  I_{2d_e}  \right) $ for $j \in \{M,N\}$, where $  \mathcal W_{0j} =  (\mathcal S_n ^{-1}\Gamma_{1,0j}\mathcal S_{n}^{\prime -1} ) ^{-1} \mathcal J_j (\mathcal S_n ^{-1}\Gamma_{1,0j} \mathcal S_n ^{\prime -1} ) ^{-1}$.  
\end{theorem}  

The conditions in Theorem~\ref{prop4:finite} are closely related to  $\mu_{m,n}$  in Assumption~\ref{ass1:finite}. Specifically, if $\mu_{m,n} = O(n ^{\ell})$, the theorem remains valid when $1/4 <\ell \leq1/2$. Similar results can be found in \cite{AR2009}  and \cite{HR2020}; these articles show that GMM estimators are asymptotically normally distributed if their moment conditions shrink to zero at such a rate. \label{r1p7}\violet{In linear models, $\mu_{m,n}$ is closely related to the convergence rate of the \textit{concentration parameter}, and conditions on  the number of instruments or the regularization parameter, in relation to $\mu_{m,n}$, have been employed in the literature (see \citealt[Table 2]{Carrasco2015} for a comprehensive review). To the best of our understanding, the conditions are not directly testable, because $\mu_{m,n}$ is neither observable nor estimable. Consequently, the condition on $\mu_{m,n} ^2 \alpha$ should be interpreted as a theoretical requirement for $\alpha$ to shrink slowly depending on the weakness of instruments.} The second condition $\sqrt{n}\alpha\to 0$  suggests a certain convergence rate of $\alpha$. This condition is required to mitigate a potential effect of the asymptotic bias of $\widehat{\Pi}_\alpha $ on the limiting distribution of $\hat{\theta}_j$. %In Section~\ref{sec: mse}, we  will discuss how this condition is related to the mean squared error of our estimators. 
%The conditions in Theorem~\ref{prop4:finite} are closely related to  $\mu_m$ appearing in Assumption~\ref{ass1:finite}. To be specific, suppose that $\mu_m = cn ^{\ell}$ for positive constants $c$ and $\ell$, then the theorem remains valid when $1/4 <\ell <1/2$. A similar result can be found in \cite{AR2009}  and \cite{HR2020}; these articles show that GMM estimators are asymptotically normally distributed if their moment conditions tend to zero at the rate of $n^{\ell}$ where $ 1/4 < \ell <1/2$. 

%$\mathcal W_{0j}$ in Theorem~\ref{prop4:finite} is nonsingular regardless of the invertibility of $\mathcal K$. In addition, if all instruments are strong and regularization is not used, $\mathcal W_{0M}$ is reduced to a linear transformation of the asymptotic covariance of \citepos{rivers1988limited} 2SCMLE.

\begin{remark}\normalfont\label{rem:tmp1}
\label{r1p10}\violet{In Theorem~\ref{prop4:finite}, $\Lambda_n$ in $\mathcal S_n$ can be understood as the weight determining the convergence rate of  $\widetilde{\Lambda} \widehat{\beta}_j$, as in \cite{Chao_et_al}. This weighting matrix is unfortunately neither estimable nor observable, because $\Lambda_n$ is not available in practice. Furthermore, the above asymptotic normality result tells us that if there exists a weak instrument (i.e., $\mu_{m,n} /\sqrt{n} = o(1)$), the asymptotic covariance in \cite{rivers1988limited} ($\mathcal S_n^{\prime -1} \mathcal W_{0j} \mathcal S_n^{-1}$ in our notation)  is divergent. However, the matrix $\mathcal S_n$ does not play any role in the Wald-based testing procedure, thereby allowing  us to conduct inference on $\theta_0$. 	For example, let $H_0: \theta_0 = \theta$ be the null hypothesis of interest, and for each $j \in \{M,N\}$, suppose there exists  $\widetilde{\mathcal W}_j$ satisfying $  \mathcal S_n' \widetilde{\mathcal W}_{j} \mathcal S_n  = \mathcal W_{0j} + o_p(1)$. That is, $\widetilde{\mathcal W}_j$ properly normalized by $\mathcal S_n$ is a consistent estimator of $\mathcal W_{0j}$. (An example of $\widetilde{\mathcal W}_j$ can be found in Section~\ref{sec:asy.var}.) Then, as in  \citet[Theorem 5.1]{ANTOINE2012350} and \citet[p.\ 22]{frazier2020weak}, the weight is canceled out in the Wald testing procedure; specifically,  $  n(\widehat{\theta}_j - \theta)' \widetilde{\mathcal W}_j ^{-1} (\widehat{\theta}_j - \theta) = n (\widehat{\theta}_j - \theta)' \mathcal S_n (\mathcal S_n '  \widetilde{\mathcal W}_{j}\mathcal S_n)^{-1} \mathcal S_n ' (\widehat{\theta}_j - \theta)\to_d\chi^2 (2d_e)    $ under $H_0$ where $\chi^2 (2d_e)$ is the chi-square distribution with $2d_e$ degrees of freedom. }
\end{remark}
\begin{remark}\normalfont \label{rem:tmp2}
	If there exist non-zero $r_n $ and $\varsigma_\ast \in \mathbb R^{2d_e}\backslash \{0\}$ such that $ r_n n^{-1/2}\mathcal S_n ^{-1} \varsigma \to \varsigma_\ast $ for some  $\varsigma \in \mathbb R ^{2d_e}\backslash \{0\}$, then $  (\varsigma'\widetilde{\mathcal W}_{j}\varsigma) ^{-1/2}\sqrt{n}  \varsigma'(\widehat{\theta}_j  - \theta_0)  \dto N(0,1) $\violet{; as in Remark~\ref{rem:tmp1}, the scaling terms are canceled out and thus $r_n $ and $\mathcal S_n$ have no effect on obtaining the limiting distribution of the test statistic, see \citet[Theorem 3]{Newey2009a}.}   As an application, we may  study the endogeneity of $Y_{2i}$ by testing $H_0: \psi_0 = 0$.  The endogeneity  could be tested  by setting $\varsigma = (-\iota_{d_e}', \iota_{d_e}')'$  for the $d_e$-dimensional vector of ones $\iota_{d_e}$, if  $\mu_{m,n} \Lambda_n ^{-1} \to \Lambda^{-1}$ and $\Lambda ^{-1}\iota_{d_e} \neq 0 $.  
\end{remark} 

\subsection{General case: a functional instrumental variable \label{sec:general}} 
  
  \subsubsection{Details on the instrumental variable \label{sec:inst}}
  Our definition of $Z_i$ in the general case is similar to that in \cite{Carrasco2012} and \cite{Carrasco2015a,Carrasco2015}. \label{r1mp2b}\violet{Specifically, we let $\mathcal H$ denote the Hilbert space of square integrable functions defined on a compact interval $\mathfrak C \in\mathbb R$ with inner product $\langle h_1, h_2\rangle _{\mathcal H}   = \int_{\mathfrak C} h_1 (t)h_2 (t) dt$  for $h_1,h_2 \in \mathcal H$ and its induced norm $\left\Vert h_1 \right\Vert_{\mathcal H} = \langle h_1, h_1\rangle_{\mathcal H}^{1/2}$. We let $Z_i= \{Z (x_i, t), t\in \mathfrak C\}$  be a random variable that takes values in  $\mathcal H$.  That is, 	if $(\Omega,\mathbb F, \mathbb P)$ denotes the underlying probability space,   $Z_i  $ is a measurable function from $ \Omega$ to $\mathcal H$ equipped with the Borel $\sigma$-field (see \citealt[Section 2.3]{HK2012} for details on random elements in $\mathcal H$).   The index $t  $ is the deterministic argument of the function-valued instrumental variable and is suppressed for notational simplicity.\footnote{\label{ft: wl2}\violet{More generally, we may consider the Hilbert space of square integrable functions with respect to $\tau$ where $\tau $ is some positive, continuous and uniformly bounded measure on $\mathfrak C$, as in \cite{Carrasco2012} and \cite{Carrasco2015a, Carrasco2015}. The results given in this paper can be applied without modification to this case, because any separable Hilbert space is isomorphic to $\mathcal H$ (\citealt[Corollary 5.5]{Conway1994}).  }}} % and let $\tau$ denote a positive measure on $\mathfrak C $ whose support is identical to $\mathfrak C$.  
  	%Then, we let $\mathcal{H} $ be the space of square integrable functions defined on $\mathfrak C$ equipped with the inner product  $\langle h_1, h_2\rangle    = \int _{\mathfrak C} h_1 (t)h_2 \left(t\right)dt$. and its induced norm  $\left\Vert h_1 \right\Vert_{\mathcal H} = \langle h_1, h_1\rangle^{1/2}$  for $h_1,h_2 \in \mathcal H$. Then,  $\mathcal H$ is a separable Hilbert space. 

  	 Example~\ref{example:miguel} and Appendix~\ref{sec:add_emp} of the Online Supplement provide examples of  function-valued $Z_i$.
  
  %This definition of $Z_i$ is a natural generalization of a random variable taking values in $\mathbb R^{n}$. Such a generalization is useful for dealing with various types of variables that can be used in practice. For example, in this setup, $Z_i$ can be a vector of many different elements such as those in Examples~\ref{example2} and~\ref{example}. As another example, if we let $Z_i= (1, x_i ,x_i^2,\ldots, x_i^K)'$, our first-stage estimator can be understood as a nonparametric estimator of $\mathbb E  [Y_{2i} | x_i ] $ in some cases; for example, given that $\mathbb E  [Y_{2i} | x_i ] = f (x_i )$ for a real analytic function $f$ defined on $(-1,1)$, we can approximate $\mathbb E[Y_{2i} |x_i]$ by using the instrument $Z_i$, see \citet[Section 2.3]{Carrasco2012} and \cite{Antoine2014}.  Lastly, as a more general case, $Z_i$ can be a functional variable such as that in the following example.%\footnote{For example, \cite{Miguel2004} use the average growth of yearly precipitation as an instrument to predict the agricultural productivity. However, such an information may be better predicted if the yearly curve of rainfall growth rate profiles is used as an instrument.} 

  \begin{example}\normalfont\label{example:miguel}
  	\cite{Miguel2004} studied the effect of economic growth on the occurrence of civil war in sub-Saharan Africa. The authors suggested addressing the potential endogeneity of economic growth using  annual rainfall growth because (i) rainfall variations are exogenous and (ii) a large volume of the economy in this area relies on agriculture. However, agricultural production may be better predicted if the instrument is not simply given by the annual {\it{average}}, but by a {\it{function}} of rainfall variations  over the year, considering the importance of daily rainfall variations  in agricultural productivity. To detail this idea, let $x_{is} (t) = 1-\dot x_{is} (t)/ \dot x_{is-1} (t)$ where $\dot x_{is}(t)$ is the daily rainfall in country $i$ on the $t$th day of year $s$. For each $i$ and $s$, we can obtain the rainfall growth curve, denoted $Z_{is}$,  by smoothing $\{   x_{is} (t) \}_{t=1} ^{365}$.\footnote{\violet{For example, one may consider $Z_{is}$ as a continuous function such that $Z_{is}(t/365) = x_{is}(t )$ for $t \in[0,365]$. Because $x_{is}(t)$ is observable only on integer points of $t$, the entire curve on year $s$ needs to be obtained by smoothing the realizations $\{   x_{is} (t) \}_{t=1} ^{365}$. A similar example can be found in \cite{HK2012}.\label{r2p7}}} %Then, $\Delta Z_{is} = Z_{is}  - Z_{is-1}$ measures the change of rainfall in country $i$ in year $s$. We note that the instrumental variable used in \cite{Miguel2004} can be obtained by projecting $\Delta Z_{is}$ onto the space spanned by the constant function, denoted $ \iota_0$, i.e., $\langle  \Delta Z_{is}, \iota_0\rangle$.
  	Then, one may consider the following model.\begin{equation*}
  		\text{conflict}_{is} = 1\{   \beta_0 +  \text{growth}_{is}\beta +  W_{is}'\gamma \geq u_{is} \},\quad\text{growth}_{is} = \Pi_{n} Z_{is} + W_{is} ' \pi_2 + V_{is}.
  	\end{equation*}
  	The variable $\text{conflict}_{is} $ is equal to 1 if country $i$ experiences a civil conflict in year $s$, and 0 otherwise. $\text{growth}_{is}$ is economic growth in country $i$ in year $s$. $W_{is}$ is a set of exogenous variables. In this case,  the first-stage parameter $\Pi_{n}$ will be represented by a linear map from $\mathcal H$  to $\mathbb R$. 
  \end{example}

   The covariance of $Z_i$, denoted $\mathcal K$, and its sample counterpart, denoted $\mathcal K_n$, are given by \[  \mathcal K  = \mathbb E  [  Z_i \otimes Z_i  ] \quad\text{and}\quad \mathcal K_n =n^{-1} \sum_{i=1} ^n Z_i \otimes Z_i, \]
  where $\otimes$ signifies the tensor product in $\mathcal H$ satisfying that $h_1 \otimes h_2 (\cdot) = \langle h_1, \cdot \rangle_{\mathcal H} h_2$ for any $ h_1,h_2 \in \mathcal H$. The tensor product  is a natural generalization of the outer product in Euclidean spaces, and thus, $\mathcal K$ (resp.\ $\mathcal K_n$) can also be understood as a generalization of the covariance (resp.\ sample covariance) matrix in Section~\ref{sec:finite}. \label{r1p11}\violet{The integral representation of $\mathcal K$ would be useful for understanding the covariance operator and the tensor product. For any $h \in \mathcal H$, $\mathcal Kh$  is an element in $\mathcal H$ and satisfies that \begin{equation*}
  		(\mathcal K h)(t) = \mathbb E [  \langle Z_i, h\rangle_{\mathcal H} Z_i (t)  ] = \int_{\mathfrak{C}} \mathbb E[Z_i(t)Z_i (\tilde t) ] h(\tilde t) d\tilde t = \int_{\mathfrak{C}} \mathtt{k}(t, \tilde t) h(\tilde t) d\tilde t,
  \end{equation*} where $\mathtt{k}(t, \tilde t)  =\mathbb E[ Z_i(t)Z_i (\tilde t)]$  is called the covariance kernel.}  We note that $\mathcal K$ (resp.\ $\mathcal K_n$) allows the spectral decomposition with respect to its eigenvalues and eigenfunctions, denoted $\{ \kappa_j , \varphi_j \}_{j \geq 1 }$ (resp.\ $\{ \hat \kappa_j , \hat \varphi_j \}_{j \geq 1 }$).  Assume that the eigenvalues are in descending order. Given the consistency of $\mathcal K_n$,  $\widehat\kappa_j  $ and $\widehat{\varphi}_j  $  are consistent estimators of $\kappa_j$ and $\varphi_j$ (\citealt[Lemmas 4.2, 4.3]{Bosq2000}).
  
\subsubsection{Estimator \label{sec:est1}}
We consider the following moment condition, which is similar to \eqref{mod2:eq1:finite}.
\begin{equation}
\violet{\mathbb E \left[  Z_i\otimes Y_{2i} \right]= \Pi _{n} \mathbb E \left[Z_i \otimes Z_i\right] = \Pi_n \mathcal K}  . \label{mod2:eq1}
\end{equation}
 \violet{That is, for any $h\in \mathcal H$, $ \mathbb E[\langle Z_i , h\rangle_{\mathcal H} Y_{2i}] = \Pi_n \mathbb E[\langle Z_i, h \rangle_{\mathcal H} Z_i]  = \Pi_n \mathcal K h$.}\label{r2p8} %\violet{The terms appearing in \eqref{mod2:eq1} are  finite rank operators from $\mathcal H$ to $\mathbb R^d_e$.} 
As  discussed by \cite{Carrasco2007}, the unique solution to \eqref{mod2:eq1} exists only when $  \mathcal K $ is bijective and its inverse is continuous. These conditions are not likely to be satisfied when $Z_i$ is an $\mathcal H$-valued random variable. Thus, we use a regularized inverse of $\mathcal K$, denoted $\mathcal K_\alpha ^{-1}$, to obtain a solution to \eqref{mod2:eq1}. The regularized inverse considered here satisfies $\underset{\alpha \to 0}{\lim} \mathcal K_{\alpha} ^{-1} \mathcal K h = h $ for any $h \in \mathcal H $ (see, e.g., \citealp{Kress1999}) and has the following representation:\begin{equation}\mathcal K_\alpha ^{-1} = \sum_{j=1} ^\infty \kappa_j ^{-1} q(\kappa_j , \alpha) \varphi_j \otimes \varphi_j . \label{kinv} \end{equation} 
The conditions on $q(\kappa_j,\alpha)$ will be detailed in Assumption~\ref{asskinv}. This representation encompasses  various types of regularization schemes, such as (i) Tikhonov ($q(\kappa_j , \alpha) = \kappa_j ^2 /(\kappa_j ^2 +\alpha)$), (ii) ridge ($ q(\kappa_j, \alpha) = \kappa_j /(\kappa_j +\alpha)$),\footnote{\label{r2p9}\violet{Tikhonov regularization is considered as an extension of ridge regularization to infinite dimensions. Specifically, as detailed in \citet[pp.\ 5694--5696]{Carrasco2007}, Tikhonov regularization is related to the method of moments estimation, whereas  ridge regularization estimates parameters that minimize the residual sum of squares. That is, the two strategies are associated with different estimation methods, although they are asymptotically equivalent as $\alpha \to 0$. Hence, different conditions are required to use them in the analysis.}} and (iii) spectral cut-off ($q(\kappa_j , \alpha) = 1\{  \kappa_j ^2 \geq \alpha  \}$).   The sample counterpart of $\mathcal K_\alpha ^{-1}$, denoted $ \mathcal K_{n\alpha} ^{-1}$, is similarly defined by replacing $\{ \kappa_j , \varphi_j \}_{j \geq 1}$ in \eqref{kinv} with $\{ \hat\kappa_j , \hat\varphi_j \}_{j \geq 1}$, i.e., $ \mathcal K_{n\alpha} ^{-1} = \sum_{j=1} ^n \widehat \kappa_j ^{-1} q(\widehat \kappa_j , \alpha) \widehat \varphi_j \otimes \widehat \varphi_j$. Then, for $\alpha>0$,  the solution to \eqref{mod2:eq1} (denoted $\Pi_\alpha$) and its sample counterpart (denoted $\widehat \Pi_{\alpha}$) are given as follows. \begin{equation}
\Pi_{\alpha} = \mathbb E  [  Z_i \otimes Y_{2i}      ] \mathcal K_\alpha  ^{-1} =   \Pi_{n}\mathcal K  \mathcal K_{\alpha} ^{-1}  \quad\text{and}\quad\widehat \Pi_{\alpha} =  (n^{-1} \sum_{i=1} ^n Z_i \otimes Y_{2i} ) \  {\mathcal K}_{n \alpha}  ^{-1}   .  \label{eq:first-est}
\end{equation} 
     Although  $\Pi_\alpha $ and $\widehat{\Pi}_\alpha$ in \eqref{eq:first-est} act on a Hilbert space of infinite dimension, the predicted value of $Y_{2i} $ and the first-stage residual, which are computed using $\widehat\Pi_\alpha$ in \eqref{eq:first-est}, are still $d_e$-dimensional. Therefore, we can obtain our estimators by solving \eqref{model:5eq} as described  in Section~\ref{sec:finite}. 
    \subsubsection{Asymptotic Properties \label{sec:asy.pro}}
% In the following,  $\mathcal A^\ast$ denotes the adjoint of a linear operator $\mathcal A: \mathcal H \to \mathbb R^{d_e}$. 
\begin{assumption}\label{ass1} \begin{enumerate*}[(i)]
		\item \label{ass1:prime} \label{r1pmp2c}\violet{Assumptions~\ref{ass1:finite},~\ref{ass3:finite} and~\ref{ass4} hold  for a bounded linear operator $\Pi_0 : \mathcal H  \hto \mathbb R^{d_e}$ and  with $ \Vert Z_i \Vert $ being replaced by $\Vert Z_i \Vert_{\mathcal H}$.} \item \label{ass7} For $\{\pi_{\ell , 0} \hspace{-0.2em}=\hspace{-0.2em} \Pi_{0} ^\ast e_{\ell}\}_{\ell=1} ^{d_e}$, there is $\rho \geq 1$ such that $ \sum_{j=1} ^\infty  \langle \varphi_j   , \pi_{\ell,0} \rangle_{\mathcal H}   ^2 /\kappa_{j} ^{2\rho } < \infty$, where $\Pi_0 ^\ast$ is the adjoint of $\Pi_0$.
	\end{enumerate*}
\end{assumption} 
%\end{assumption}
\begin{assumption}\label{asskinv}
	% \begin{enumerate*}[(i)]%\item\label{asskinv0} $\mathbb E [\Vert Z_i \Vert^2] < \infty $;
	  $\mathcal K_{\alpha} ^{-1}$ allows the representation in \eqref{kinv}. For $\kappa \geq 0$ and $\alpha >0$, the function $q  (\kappa , \alpha  ) $ in \eqref{kinv} satisfies  $ q (\kappa,\alpha)\in [0,1]$, $ \alpha q  (\kappa , \alpha  ) \leq   c \kappa $ and $\sup_{\kappa}\kappa^{2\tilde\rho} (q(\kappa , \alpha)-1) \leq$ $ c\alpha^{\min \{\tilde \rho,1\} } $ for some $\tilde \rho >0$ and $c>0$. Moreover, $\underset{\alpha \to 0}{\lim} q(\kappa , \alpha) = 1$ if $\kappa > 0$.
	% \item abel{asskinv2}$\left\Vert \mathcal K_\alpha ^{-1} \right\Vert_{op} \leq c \alpha^{-1}$ for a positive constant $c$.
	%\end{enumerate*}
\end{assumption}

Assumption~\ref{ass1}.\ref{ass1:prime} is equivalent to the former conditions, except that the first-stage parameter and the norm of $Z_i$ are adapted to the general case. \violet{The vector $f_i$ can be viewed as the (infeasible) optimal instrument that summarizes the information from infinite-dimensional instrument $Z_i$.} \label{r1p8}\label{r4p3a}\violet{The $d_e \times d_e$ matrix $\Lambda_n$    determines the identification strength of $Z_i$ in the direction to each $\pi_{\ell,0}$. As an example, suppose that \begin{equation*}
		d_e = 2,\quad	 f_i =  \begin{bmatrix}
			 	 \langle Z_i, \pi_{1,0}\rangle _{\mathcal H} \\\langle Z_i, \pi_{2,0}\rangle _{\mathcal H} 
			 \end{bmatrix}, \quad \Lambda_n = \begin{bmatrix}
			 1 &0 \\ \ddot{\lambda} &1 
		 \end{bmatrix}\begin{bmatrix}
		 \sqrt{n} &0\\0& \mu_{2,n}
	 \end{bmatrix} , \quad\langle \pi_{1,0} ,\pi_{2,0}\rangle_{\mathcal H} = 0 .
		\end{equation*}Let $[ a]_{i}$ denote the $i$th row of vector $a$. Then, the first stage can be written by \begin{equation*}
	  [Y_{2i}]_1 = \langle Z_i, \pi_{1,0}\rangle_{\mathcal H} +[ V_i]_{1}, \quad\text{ and }\quad  [Y_{2i}]_2 = \ddot{\lambda}\langle Z_i,  \pi_{1,0}\rangle_{\mathcal H} + \frac{\mu_{2,n}}{\sqrt{n}} \langle Z_i, \pi_{2,0} \rangle_{\mathcal H} + [V_i]_2 .
	\end{equation*}}\violet{Thus, the reduced-form coefficient of the component of $Z_i$ with respect to $\pi_{1,0}$, $\langle Z_i, \pi_{1,0}\rangle_{\mathcal H}$, is strongly identified. In contrast, $[\beta]_2$ is strongly identified if $\mu_{2,n}=\sqrt{n}$, while weakly identified otherwise. Thus, $\Lambda_n$ allows for reduced-form coefficients to be associated with different identification powers. For further discussion when $\mathcal H = \mathbb R^{d_z}$, refer to \citet[pp.\ 48--49]{Chao_et_al}.}   

 Assumption~\ref{ass1}.\ref{ass7} is needed to identify $\Pi_0$ without the injectivity of $ \mathcal K$; see \cite{Carrasco2007} and \cite{Benatia2017}. As briefly mentioned in Section~\ref{sec:finite}, this assumption restricts the relative decaying rate of $\mathcal K$'s eigenvalues with respect to its Fourier coefficients $\{\langle\pi_{\ell,0},\varphi_j\rangle_{\mathcal H}\}_{j \geq 1} $. It is particularly essential for proving that  $n^{-1}\sum_{i=1} ^n \Vert \widehat g_{i,\alpha} - g_i \Vert ^2 =o_p(1)$, which plays an important role in obtaining our main results. Assumption~\ref{asskinv} describes the conditions on $\mathcal K_{\alpha} ^{-1}$. Many popular regularization methods, such as Tikhonov and spectral cut-off, satisfy the conditions (see \citealt[Section 3.3]{Carrasco2007}). Therefore, Assumption~\ref{asskinv} is not restrictive in practice. 
 % These conditions will be fulfilled if Tikhonov or spectral cut-off regularization, which are commonly used in practice, are employed (see, e.g., \citealt[Section 3.3]{Carrasco2007}). Thus,  from a practical perspective, Assumption~\ref{asskinv} is not restrictive.

The following theorem states the consistency of $\widehat{\theta}_j$ in the general case. 
 \begin{theorem}\label{thm1:consistency}\normalfont
	Suppose that $\alpha \hto 0$ and $\mu_{m,n} ^2 \alpha\hto \infty$ as $n \hto \infty$, and Assumptions \ref{ass1} to \ref{asskinv} are satisfied. Then, $
	\widehat \theta_{ j} - \theta_{0}  \pto 0
	$ as $n \hto \infty$ for $j \in\{ M,N \}$.
\end{theorem}
 \label{r1p13}\violet{The condition  on $\mu_{m,n}$ and $\alpha $ in Theorem~\ref{thm1:consistency} requires a slower convergence rate of $\alpha$ compared to that in Theorem~\ref{thm1:consistency:finite}. This stronger condition is  required mainly due to the additional conditions on the regularized inverse in Assumption \ref{asskinv}.} In fact, some regularization methods, such as Tikhonov, spectral cut-off, and Landweber Fridman, satisfy  $\alpha^{1/2} q(\kappa,\alpha)\leq c \kappa$ for some $c >0$, and for them, the condition on $\mu_{m,n}^2 \alpha$ in Theorem~\ref{thm1:consistency} can be replaced by its weaker counterpart in Theorem~\ref{thm1:consistency:finite}.

 Next, we discuss the asymptotic distributions of the proposed estimators. As discussed in Section~\ref{sec:finite}, the asymptotic distributions cannot be straightforwardly obtained using  the approaches  in \cite{rivers1988limited}, \cite{Carrasco2012}, and \cite{Carrasco2015a,Carrasco2015}; this is not only because of some distinctive properties of our model which are described in the previous section, but also because of the fact that $\widehat{\gamma}_{i,\alpha}$ and $\widehat{V}_{i,\alpha}$ now depend on an infinite-dimensional parameter estimator, $\widehat \Pi_\alpha$. This necessitates  an approach to asymptotic analysis that differs from those in previous articles. Our setting is similar to the one studied by \cite{Chen2003}, in which the criterion function contains an infinite-dimensional parameter estimate. \cite{Chen2003} discuss the primitive conditions for the asymptotic normality of an estimator obtained in such a setting. Here, we make use of their approach and study the asymptotic distributions of our estimators.  
 
 To this end, we employ additional conditions. We first introduce a condition on the integrabiliy of the covering number of $\mathcal H$ with respect to $ \Vert\cdot\Vert_{\mathcal H}  $ (denoted $\mathfrak N (\epsilon,\mathcal H,  \Vert \cdot\Vert_{\mathcal H}  ) $), i.e., $\mathfrak N (\epsilon,\mathcal H,  \Vert \cdot  \Vert_{\mathcal H}  ) = \inf \{  n \in \mathbb N : \mathcal H \subset \bigcup_{k=1} ^n \mathcal B_k  (\epsilon ) \}$ where $\mathcal B_k  (\epsilon ) = \{ h \in \mathcal H :  \Vert h - h_k  \Vert_{\mathcal H}  < \epsilon , h_k \in \mathcal H \}$. 
\begin{assumption}\label{ass8}
	 $ \mathcal H$ satisfies $\int_0 ^\infty \sqrt{\log \mathfrak N  (\epsilon  , \mathcal H , \| \cdot \|_{\mathcal H}   )} d\epsilon <\infty$.
\end{assumption}
%  Assumption \ref{ass8} is similar to Condition 3.3 in \cite{Chen2003}.
  In \citet[Condition 3.3]{Chen2003}, the covering number of an infinite-dimensional parameter space should satisfy the integrability condition similar to Assumption~\ref{ass8}. In our setting, the parameter space is given by the space of bounded linear operators from $\mathcal H$ to $\mathbb R^{d_e}$. Thus, it is not straightforward to verify \citepos{Chen2003} condition. In contrast,  Assumption~\ref{ass8} pertains only to the domain of our infinite-dimensional parameter, making it easier to verify. In the current paper, the restriction on the domain of $\Pi_0$ is sufficient, because its rank is finite (see Lemma~\ref{lem4} in the Online Supplement). This assumption is likely to hold in many practical choices of $\mathcal H$; for example, it holds if $Z_i$ is  finite-dimensional (as discussed in Section~\ref{sec:finite}), or \label{r1p14}\violet{is smooth enough and takes values in the space of functions with many finite derivatives}, see, e.g., \cite{Vaart1996}.

\begin{assumption}\label{ass10} 
\begin{enumerate*}[(i)]
\item\label{ass9} Assumption~\ref{ass9:finite}  holds;	\item  \label{ass10-1} For $\kappa \geq 0$ and $\alpha >0$,   $q(\kappa, \alpha)$ in \eqref{kinv} satisfies $ \alpha^{1/2} q  (\kappa , \alpha  ) \leq  c\kappa $ for some constant $c>0$;
	%\item \label{ass10-3} $  \lambda^{2\rho}( q  (\lambda , \alpha  ) -1) \leq O  (\alpha^{\rho})$ for $\rho \in [0,1]$; 
	\item \label{ass10-2} Assumption \ref{ass1}.\ref{ass7} holds with $\rho \geq 3/2$.
\end{enumerate*}	
\end{assumption}  Assumption~\ref{ass10} is needed to address the bias discussed in Remark~\ref{rem1}.   Assumption~\ref{ass10}.\ref{ass10-1} implies that not all regularization methods are applicable in our setting. For example,  ridge regularization may fail to satisfy the condition since its conservative upper bound of $ q(\kappa,\alpha)$ is  $\kappa  O(  \alpha ^{-1})$. On the other hand,   other aforementioned  methods, such as Tikhonov and spectral cut-off, satisfy the condition. Hence, in spite of its popularity,   ridge regularization is less preferred  in this paper.  

 Assumption~\ref{ass10}.\ref{ass10-2}  can be relaxed if $\ran \mathcal V_j ^\ast$, the range of the adjoint operator of $ \mathcal V_j $,  belongs to $ \ran \mathcal K$, the range of $\mathcal K$, where $\mathcal V_j $ is  defined by $  \mathbb E  [ \dot m_{2j} \left( \theta_0, g_i\right) Z_i \otimes  g_{0i}  ]$.   Appendix~\ref{sec:rem} in the Online Supplement shows that the closure of $\ran \mathcal V_j ^\ast$, denoted $ \cl (\ran \mathcal V_{j}^\ast)$, is always a subspace of the closure of $\ran \mathcal K$, denoted $\cl (\ran \mathcal K)$. 
 If $Z_i$ is finite-dimensional or can be represented by a finite number of basis functions,   $\cl(\ran \mathcal V_{j}^\ast) = \ran \mathcal V_j ^\ast $ (resp.\ $ \cl (\ran \mathcal K)=\ran \mathcal K$). In this case, Assumption~\ref{ass10}.\ref{ass10-2} will be satisfied.  Hence, the assumption is not restrictive as it appears. 
%In fact, Assumption \ref{ass10}-\ref{ass10-2} can be relaxed if Assumption \ref{ass7} holds with $\rho \geq 2$, and furthermore, Remark \ref{rem4} ensures that the first part of Assumption \ref{ass10}-\ref{ass10-2} may not be restrictive from a pragmatic point of view. To be specific, in many empirical examples, $Z_i$ is finite-dimensional or may be approximated by a finite number of basis functions, and Remark \ref{rem4} ensures that the first part of Assumption \ref{ass10}-\ref{ass10-2} holds in such cases.

  In this section, $\mathcal C_j$  is  a linear operator from $\mathcal H$ to $\mathbb R^{d_e}$ satisfying $  \mathcal K^{1/2}\mathcal C_j ^\ast = \mathcal V_{j} ^\ast $, and the $(\ell, k)$th element of $\mathcal J_{2,j}$ is given by $\sigma_{\psi_0} ^2 \langle \mathcal C_j ^{\ast} e_{\ell}, \mathcal C_{j}^{\ast}e_k\rangle_{\mathcal H}$. The operator $\mathcal C_j$ is unique and bounded, even when $\mathcal H$ is a Hilbert space of infinite dimension (\citealt[Theorem 1]{baker1973joint}). 
\begin{theorem}\label{prop4}\normalfont
	Suppose that $\alpha \hto 0$, $\mu_{m,n} ^2 \alpha\hto\infty$ and $\sqrt{n} \alpha\hto0$ as $n \hto\infty$, and  Assumptions \ref{ass1} to \ref{ass10} are satisfied.  Then, $ \mathcal W_{0j} ^{-1/2} \sqrt{n} \mathcal S_n ^\prime \left(\widehat \theta_j - \theta_{0}  \right) \dto \mathcal N \left( 0, \mathcal I_{2d_e}  \right)$ for $j \in \{M,N\}$, where $  \mathcal W_{0j} =  (\mathcal S_n ^{-1}\Gamma_{1,0j}\mathcal S_{n}^{\prime -1} ) ^{-1} \mathcal J_j (\mathcal S_n ^{-1}\Gamma_{1,0j} \mathcal S_n ^{\prime -1} ) ^{-1}$.  
\end{theorem}  
 
 The asymptotic covariance $\mathcal W_{0j}$ is nonsingular irrespective of the dimension of $Z_i$. Thus, we can conduct inference on $\theta_0$ and $\varsigma'\theta_0$ as discussed in Remarks~\ref{rem:tmp1} and \ref{rem:tmp2}. 
 
\subsection{Estimation of Asymptotic Variances and Average Structural Functions}\label{sec:asy.var}
\violet{ Remark~\ref{rem:tmp1} tells us that the Wald testing procedure enables us to conduct inference on $\theta_0$ without prior knowledge of $\mathcal S_n$, as long as $\widetilde{W}_{j}$, a consistent estimator of $\mathcal{W}_{0j}$ after proper normalization, is available. $\widetilde{W}_j$ is an essential input for the testing procedure. In this section, we discuss its example.}\label{r1p9}%A potential candidate would be a bootstrap estimator
	% method may be used for this purpose, but it may be  computationally demanding, especially when the dimension of $Z_i$ is very large. } consistent estimation of $\mathcal W_{0j}$ in Theorems~\ref{prop4:finite} and~\ref{prop4} is essential to conduct inference on $\theta_0$. As noted by \cite{Chen2003}, a bootstrap method may be used for this purpose, but it may be  computationally demanding, especially when the dimension of $Z_i$ is very large. \violet{In this section, we provide a simple and practical method for computing a consistent estimator of $\mathcal W_{0j}$. }

\violet{Let $\widehat \Gamma_{1,j} = n^{-1}\sum_{i=1} ^n \dot m_{2j} ( \widehat \theta_j,\widehat g_{i,\alpha}  ) \widehat g_{i,\alpha} \widehat g_{i,\alpha} '$, $\widehat {\mathcal J}_{1,j} = n^{-1}\sum_{i=1} ^n  m_{1j} ^2  (\widehat \theta_j, \widehat g_{i,\alpha} ) \widehat g_{i,\alpha} \widehat g_{i,\alpha} '$, %Lemmas \ref{lem8} and \ref{lem9} in Appendix~\ref{sec:lem} in the Online Supplement show that $ \Vert \mathcal S_n ^{-1} (\widehat \Gamma_{1,j} -\Gamma_{1,0j} )\mathcal S_n ^{\prime-1} \Vert_{\HS}  $ and $  \Vert \mathcal S_n ^{-1} \widehat{\mathcal J}_{1,j} \mathcal S_n ^{\prime -1} - \mathcal J_{1,j}  \Vert_{\HS}   $ are $o_p(1)$. 
and $\widehat {\mathcal J}_{2,j} = \widehat {\sigma}_{j}^2 \widehat {\mathcal V}_{jn}  \mathcal K_{n\alpha} ^{-1}\mathcal K_n \mathcal K_{n\alpha} ^{-1}  \widehat {\mathcal V}_{jn}  ^\ast$ where $ \widehat {\sigma}_{j}^2 = n^{-1}\sum_{i=1} ^n  (\widehat \psi_j ' \widehat V_{i,\alpha} )^2 $ and $ \widehat {\mathcal V}_{jn} =n^{-1} \sum_{i=1} ^n \dot m_{2j}  ( \widehat \theta_j , \widehat g_{i,\alpha } ) Z_i \otimes \widehat g_{i,\alpha} $. %Lemma \ref{lem10} shows that $\mathcal S_n^{-1} \widehat{\mathcal J}_{2,j} \mathcal S_n^{\prime-1}$ is a consistent estimator of $\mathcal J_{2,j}$. Thus,  $\mathcal S_n ^{-1}\widehat {\mathcal J}_{j}   \mathcal S_n^{\prime-1}$   converges to $\mathcal J_j$ in probability where $  \widehat {\mathcal J}_{j}  = \widehat {\mathcal J}_{1,j} + \widehat {\mathcal J}_{2,j} $. 
The   theorem below summarizes Lemmas~\ref{lem8}--\ref{lem10} in the Online Supplement. } 
\begin{theorem}\normalfont\label{prop6}
	Let $\widehat {\mathcal W}_{j } =  \widehat \Gamma_{1,j} ^{-1} (\widehat {\mathcal J}_{1,j} + \widehat {\mathcal J}_{2,j}) \widehat \Gamma_{1,j} ^{-1}   $ and suppose that the conditions in Theorem \ref{prop4} hold. Then, for each $j$, $
	 \Vert \mathcal S_n ^\prime\widehat {\mathcal W}_{j}\mathcal S_n- \mathcal W_{0j}  \Vert_{\HS} = o_p(1)
	$, \violet{where $\Vert  A \Vert_{\HS}= (\tr(A'A))^{1/2}$ for a matrix $A$.} \label{r1mpmp1}
\end{theorem} 
%An important discussion related to inference on $\widehat \theta_j$ is related to the exogeneity test of $Y_{2i}.$ The rejection of the null hypothesis that $\psi = 0$ provides statistical evidence for the endogeneity of $Y_2$ that results in bias of the naive probit estimator. This idea is also suggested in \cite{rivers1988limited} and we can find a detailed discussion on relevant test statistics therein.
\begin{remark}\normalfont \label{rem:asf}
	We note that interest often lies in how the outcome probability changes according to a change in $Y_2$. In practice, this is often measured by two estimates, the average structural function (ASF) and the average partial effect (APE), each of which is given by \[{\text{ASF}} (y_{2}) := \frac{1}{n} \sum_{i=1} ^n\Phi  ( y_{2} ' \widehat \beta + \widehat V_{i,\alpha} ' \widehat \psi   )\quad\text{and}\quad\text{APE} (y_2) := \frac{1}{n} \sum_{i=1} ^n \phi (y_{2} ' \widehat \beta + \widehat V_{i,\alpha} ' \widehat \psi  ) \widehat{\beta} ,\] see \cite{Wooldridge2010}, \cite{Blundell_Powell2004} and \cite{Rothe2009}. The ASF (resp.\ APE) is a consistent estimator of $\mathbb E_{V}[\mathbb P(y=1 | Y_2 = y_2, V)]$ (resp.\ $\mathbb E_{V}[  (\partial\mathbb P(y = 1 | Y_2, V)/\partial Y_2 )_{Y_2 = y_2}  ]$).  
\end{remark} 
\subsection{ Mean Squared Error and the Choice of $\alpha$  }\label{sec: mse}
 A practical challenge in implementing our estimation procedure is the selection of the regularization parameter. A theoretically  grounded approach would be to choose $\alpha$ in such a way as to minimize the conditional mean squared error (MSE), defined by $\mathbb E[ \Vert \hat{\theta}_j -\theta_0 \Vert^2 |x ]$, as in  \cite{donald2001choosing} and \cite{Carrasco2012}.  However,  our analysis of the MSE is more complicated than theirs because of  the nonlinearity and  the use of the estimated control function  whose asymptotic properties rely on a possibly infinite-dimensional random element. In particular, we need more conditions to control for the linearization error. Thus, to reduce the complexity, we consider the conditional MSE of an alternative estimator $\overline\theta$, defined as the solution to a linear problem. In this section, the subscript $j$ indicating estimators is suppressed to simplify the explanation. \violet{In  the simple case considered in Section~\ref{sec:finite}, $\overline\theta$ is given by the solution to the following: \begin{equation*}
\mathcal{L}(\overline{\theta})=\frac{\partial \mathcal{Q}_{n}(\theta_0 ,g_i)  }{ \partial \theta} + \frac{\partial^2 \mathcal{Q}_{n} (\theta_0 ,g_i)}{ \partial \theta\partial \theta'} (\overline\theta - \theta_0) + \frac{\partial ^2\mathcal{Q}_{n}(\theta_0 ,g_i)}{ \partial \theta\partial \pi'}  (\hat{\pi}_\alpha - \pi_n) = 0,
\end{equation*} where $\pi$,  $\hat{\pi}_\alpha$, and $\pi_n$ denote  the vectorizations of $\Pi'$, $\widehat\Pi_{\alpha}'$ and $\Pi_{n}'$, respectively.}
 \label{r4p5}\violet{We defer its formal definition for the general case to \eqref{thetaoverline: def} in Appendix~\ref{app: mse} of the Online Supplement, as it involves additional notations.}  In our setting, the asymptotic distribution of $ \overline\theta$ is equivalent to that of $ \widehat{\theta} $. Hence, studying $\overline \theta$ would be enough for the purpose of discussing a way to choose $\alpha$.  

 The proposition below summarizes an upper bound of $\overline\theta$'s conditional MSE; we can provide more details  if  $\mu_{jn}$'s are known. While the following result may look conservative, it is still useful in practice due to two reasons:  (i) the unknown nature of $\{\mu_{jn}\}_{j=1} ^{2d_e}$ and (ii) the fact that the rate below represents the outcome when there is at least one strong instrument.
\begin{proposition}\label{prop: mse}\normalfont \violet{Suppose that the conditions in Theorem~\ref{prop4} hold.} Then, $\mathbb E[ \Vert \mathcal S_n ' (\overline\theta - \theta_0)\Vert ^2 |x]= O_p(n^{-1}\alpha^{-1/2} + \alpha^2).$
\end{proposition} 
The MSE of $\overline \theta$ turns out to be minimized where the MSE of  $\widehat{\Pi}_\alpha$ is minimized. Due to this, the above convergence rate  is similar to that in the functional linear model studied by \citet[Proposition~2]{Benatia2017}, who suggested choosing $\alpha$ to minimize the conditional MSE. In our setting, the optimal $\alpha$ satisfies  $\alpha \sim n^{-2/5}$. However, this rate may not be preferred, because it does not ensure the asymptotic normality of our estimators. Our asymptotic normality results require a fast convergence of $\widehat{\Pi}_n$'s bias, even at the cost of a larger variance; if we choose $\alpha$ as in \cite{Benatia2017}, the limiting distribution of  $\hat{\theta}_j$ will depend on the regularization bias of the first-stage estimator, which may not be desirable in practice. Hence, a slower convergence of $\overline \theta$ (and that of $\hat \theta_j$), resulting from not choosing the optimal rate of $\alpha$, should be understood as the cost of implementing  inference with a possibly infinite-dimensional instrumental variable.    In spite of its limitation, the above criterion still provides us with a practical guideline to select the regularization parameter. We defer its detailed discussion to Appendix~\ref{stepwise} in the Online Supplement.

\section{Simulation Study\label{sec:sim}} 
In this section, we study the finite sample performance of our estimators via Monte Carlo simulations when  $Z_i \in \mathbb R^{K}$. The simulation results for the general case is separately discussed in Appendix~\ref{sec: func} of the Online Supplement. 
\subsection{Experiment 1: Gaussian Instrument \label{sec:sim1}}
We consider the following data generating process (DGP).
\begin{equation}
\begin{aligned}
y_i & = 1 \{  y_{2i}  \beta_1 + z_{1i} \beta_{2} \geq u_i   \},\\
y_{2i}  &=  Z_i '\pi + v_i ,
\end{aligned} \quad \qquad
\begin{aligned}
\begin{bmatrix}
u_i \\v_i
\end{bmatrix} \sim_{\text{iid}}\mathcal N \left(\begin{bmatrix}
0\\0
\end{bmatrix}, \begin{bmatrix}
\sigma_1 ^2  & -\rho \sigma_1 \sigma_2 \\ -\rho\sigma_1 \sigma_2 & \sigma_2 ^2
\end{bmatrix}\right),
\end{aligned}  \label{eq:sim}
\end{equation}
where $\beta_1 = 1$, $\beta_2 = -1$, $\rho =  0.6$, and $\sigma_1 ^2 = 1 /  (1-\rho ^2 )$. Thus, $\eta_i = u_i +\rho \sigma_2 ^{-1}\sigma_1 v_i \sim_\text{iid}\mathcal N  (0,1 )$.  The variable $z_{1i}$ is the first element of $Z_{i}\in \mathbb R^K$, and $Z_i \sim_{\text{iid}} \mathcal N (0, \Sigma_Z)$ where   $\Sigma_Z =  [ \sigma_z ^2   \rho_z ^{ |i-j |}   ]$, $\sigma_z ^2 = 0.5$ and $\rho_z = 0.7$. %Under this design, as either $\Sigma_Z $ gets closer to the identity matrix or $K$ increases, our estimators may perform relatively poorly, because, in that case, Assumption \ref{asskinv} nearly fails. 
We let $\sigma_2^2 = 1-\pi ' \Sigma_Z \pi$ so that the unconditional variance of $y_{2i}$ becomes 1.

We set $\pi =  c_\ast (\iota_{\lfloor s K \rfloor} ', 0_{K-\lfloor s K \rfloor})'$, where $\lfloor \cdot \rfloor$ is the floor function. The parameter $s$ determines the sparsity of $\pi$. We consider sparse ($s=0.2$) and dense ($s=0.8$) cases.    We set $K$ to $ 50$.  The values of $c_\ast $ are chosen to have specific values of the concentration parameter,  $\mu^2 = n \pi ' \Sigma_Z \pi/ (1-\pi ' \Sigma_Z \pi  )$, as in \cite{Belloni}. We consider two values of $\mu^2$: 30 and 60. Under the design, the infeasible first-stage F test statistic, given by $\mu^2 / \lfloor s K \rfloor$, ranges from 0.375 to 6. \label{r2p1} \violet{The concentration parameter may not be a perfect measure of instrument strength for binary response models. However, considering its relevance in indicating the weakness of instruments in linear models,   standard  approaches to estimating endogenous binary response models, such as that of \cite{rivers1988limited}, may not produce reliable estimation results when $\mu^2$ is small. A similar scenario with small $\mu^2$ has been considered by \cite{Magnusson2010} and \cite{Dufour2018} who study weak instruments in  endogenous binary response models. Later in the section, we experiment with different values of $\mu^2$, and the results confirm the distortion of an existing estimator, similar to the TSLS, when $\mu^2$ is small. Hence, this simulation setting would be enough to study finite sample performance of our estimators when the existing estimation approach does not work well.}  %Later in this section, we will show that \citepos{rivers1988limited} estimator tends to have larger bias as the concentration parameter decreases, using different values of $\mu^2$, $s$ and $K$. 

%Lastly, we obtain regularization parameter $\alpha$ based on 10-fold cross-validation minimizing first stage mean absolute errors (MAE). Since MAE is relatively robust for an extreme value compared to root mean squared errors (RMSE), it is preferred when the first stage signal is weak.

We consider two RCMLEs computed with different regularization methods: the TRCMLE (Tikhonov) and SCRCMLE (spectral cut-off).\footnote{The results from the RNLSE are similar to those from the RCMLE. Hence, they are omitted.}  Their regularization parameters are chosen as described in Appendix~\ref{stepwise} in the Online Supplement with $\alpha$ being selected from 25 equally spaced points between $\mathtt{c}_a n^{-0.6}0.001$ and $\mathtt{c}_a n^{-0.6}$.\footnote{Specifically, we compute  \eqref{eq: mse: tmp} with Generalized cross-validation in \citet[]{Carrasco2012} for the RNLSE.} The constant $\mathtt{c}_a$ is set to $\overline{\mathtt{c}}_a \max\{ 0.1, 1/\delta   \}$ where $\delta  $ is the first-stage F test. The constant $\overline{\mathtt{c}}_a$ is given by $\tr(\Sigma_Z ' \Sigma_Z)^{1/2} $  for the TRCMLE and its square for the SCRCMLE. This is designed to have a larger regularization parameter when a smaller concentration parameter is given.  We compare the performance of our estimators with  four alternatives: \citepos{rivers1988limited} 2SCMLE, the  Inf.2SCMLE (to be detailed shortly), the naive probit estimator, and the Tikhonov-regularized TSLS estimator (TTSLS).  Similar to the TSLS, the 2SCMLE is expected to be biased as $K$ increases. To separate any effect of this bias from that induced by weak instruments, we compute the 2SCMLE using  instruments with non-zero first-stage coefficients only. This modified estimator is called the Inf.2SCMLE. If there is no distortion from weak instruments, the infeasible 2SCMLE is likely to perform well when $\mu^2 = 30$ and $s=0.3$. In contrast, if the estimator is not affected by the many instruments bias, it will perform well when $\mu^2 = 60$ and $s=0.8$.  %This estimator is expected to perform well compared to the 2SCMLE especially in the sparse design ($s=0.2$). 
The  TTSLS  is similar to the estimators proposed by \cite{Carrasco2012} and \cite{Carrasco2015a, Carrasco2015}, except for the use of the first-stage residual as an additional regressor in the second stage.\footnote{The variance is computed with the heteroscedasticity-robust covariance estimator $\text{HC}_1$ in \cite{mackinnon1985}. We use the regularization parameter chosen for the TRCMLE.}  This estimator is considered because of the popularity of linear probability models in empirical studies. As its true parameter is different from that of all the other estimators, we evaluate its performance at $\mathbb E  [\phi (g_i ' \theta )  ] \beta_1$ so that a reasonable comparison can be made (see \citealt[Chapter 14]{cameron2005microeconometrics}). For each estimator, we compute the median bias (Med.Bias) and median absolute deviation (MAD). We also test $H_0: \beta_1 = 1$ at 5\% significance level as explained in Remark~\ref{rem:tmp2}. The rejection rate is reported in   columns labeled ``RP'' in Table~\ref{tab1}.%The regularization parameters $\alpha$ for the aforementioned estimators are chosen to the value that optimizes a cross-validation criterion in the first stage. Specifically, we use the value that minimizes the first-stage mean squared error of ridge regression.\footnote{The parameter space of $\alpha$ is set to $ [ \Vert \Sigma_Z\Vert/ ( n^{4/5} \delta),  \Vert \Sigma_Z\Vert/ ( n^{3/5} \delta) ] $ where $\delta  $ denotes the first-stage F test divided by 10. This is designed to make the lower bound of $\alpha$ be larger when a smaller concentration parameter is given.} This naive selection criterion is chosen because it is readily available in practice and is easy to compute. Although the following simulation results show that our estimators  computed with this choice of $\alpha$  work comparably well, this is of course  not the best strategy for choosing the regularization parameter and estimation results can be improved if a better model selection methodology is given. For example, as  in Section~\ref{sec:emp},  $\alpha$ can be chosen to the value that maximizes the predicted conditional likelihood in the second stage.  
\begin{table}[h!]
	%	\rule[6pt]{1\textwidth}{1.25pt} 
		\caption{\small{Simulation results (Gaussian instruments)}}\label{tab1}
	\vspace{-.5em}{\footnotesize{
			\begin{tabular*}{1\textwidth}{@{\extracolsep{\fill}}lllrccrcc}
			\midrule &&&\multicolumn{3}{c}{$\mu^2=30$}&\multicolumn{3}{c}{$\mu^2=60$}\\ \cmidrule{4-6}\cmidrule{7-9}
			$n$&$s$&&Med.Bias&MAD&RP&Med.Bias&MAD&RP\\ \midrule
	\multirow{12}{*}{$200$}&\multirow{6}{*}{$0.2$}&TRCMLE & 0.006 & 0.287 & 0.045 & 0.042 & 0.221 & 0.058 \\ 
&&	SCRCMLE & -0.047 & 0.243 & 0.040 & -0.039 & 0.189 & 0.050 \\ 
&&	Inf.2SCMLE & -0.221 & 0.196 & 0.134 & -0.128 & 0.166 & 0.096 \\ 
&&	2SCMLE & -0.559 & 0.121 & 0.840 & -0.461 & 0.113 & 0.729 \\ 
&&	Probit & -0.729 & 0.068 & 0.788 & -0.712 & 0.068 & 0.839 \\ 
&&	TTSLS & 0.078 & 0.273 & 0.110 & 0.073 & 0.204 & 0.136 \\ 
\cmidrule {2-9}  	
&\multirow{6}{*}{$0.8$}& TRCMLE & 0.033 & 0.258 & 0.053 & 0.032 & 0.195 & 0.062 \\ 
&&SCRCMLE & -0.002 & 0.226 & 0.048 & -0.015 & 0.170 & 0.047 \\ 
&&Inf.2SCMLE & -0.468 & 0.133 & 0.651 & -0.358 & 0.123 & 0.513 \\ 
&&2SCMLE & -0.516 & 0.124 & 0.776 & -0.413 & 0.117 & 0.649 \\ 
&&Probit & -0.711 & 0.070 & 0.824 & -0.683 & 0.069 & 0.892 \\ 
&&TTSLS & 0.074 & 0.244 & 0.110 & 0.047 & 0.176 & 0.131 \\ 
\midrule\multirow{12}{*}{$400$}&\multirow{6}{*}{$0.2$}&TRCMLE & -0.035 & 0.235 & 0.045 & 0.011 & 0.187 & 0.042 \\ 
&&SCRCMLE & -0.059 & 0.209 & 0.040 & -0.053 & 0.168 & 0.049 \\ 
&&Inf.2SCMLE & -0.209 & 0.182 & 0.130 & -0.128 & 0.152 & 0.095 \\ 
&&2SCMLE & -0.526 & 0.117 & 0.845 & -0.419 & 0.108 & 0.721 \\ 
&&Probit & -0.738 & 0.046 & 0.968 & -0.730 & 0.046 & 0.975 \\ 
&&TTSLS & 0.029 & 0.233 & 0.072 & 0.048 & 0.177 & 0.091 \\ 
\cmidrule {2-9}  	
&\multirow{6}{*}{$0.8$}&TRCMLE & -0.024 & 0.221 & 0.046 & 0.010 & 0.164 & 0.046 \\ 
&&SCRCMLE & -0.020 & 0.203 & 0.034 & -0.017 & 0.151 & 0.050 \\ 
&&Inf.2SCMLE & -0.450 & 0.119 & 0.662 & -0.325 & 0.108 & 0.496 \\ 
&&2SCMLE & -0.489 & 0.113 & 0.794 & -0.371 & 0.104 & 0.640 \\ 
&&Probit & -0.729 & 0.048 & 0.977 & -0.711 & 0.049 & 0.989 \\ 
&&TTSLS & 0.025 & 0.216 & 0.078 & 0.036 & 0.157 & 0.094 \\ \midrule
		\end{tabular*}
	}}
 \vspace{-1.5em}	\flushleft{\scriptsize{ Notes: The simulation results based on 2,000 replications are reported. Each cell reports the median bias (Med.Bias), median absolute deviation (MAD), and rejection probability at 5\% significance level (RP).} }\vspace{-0.8em}
\end{table}

In Table~\ref{tab1}, our estimators exhibit better performance compared to the 2SCMLE and Inf.2SCMLE. Specifically, the 2SCMLE is severely biased and does not have reasonable rejection rates. In contrast, our estimators have smaller median bias and correct size in most cases considered in the table. %The  also has small median bias and correct size compared to the 2SCMLE, although its performance is dominated by that of the TRCMLE. We also note that the performance of the SCRCMLE is related to the DGP under which every instrument is equally important; the SCRCMLE can produce better estimation results when the DGP is driven by a few factors, see Section~\ref{sec:factor}.
 Another interesting observation is  that the SCRCMLE outperforms the TRCMLE in the dense design, although the difference diminishes with a larger sample size. The better performance of the SCRCMLE in the small sample may be related to the relative stability of the regularized inverse when the signal is dense. Similar results can be found in   \cite{seong2021}.  

The Inf.2SCMLE does not perform well when the signal is either dense ($s=0.8$) or weak ($\mu^2= 30$) in  Table~\ref{tab1}. The distortion in the dense design may be related to the many instruments bias;  the estimator uses 40 instruments when the signal is dense. Meanwhile, the distortion when $\mu^2= 30$  casts doubt on the reliability of  \citepos{rivers1988limited} 2SCMLE when the concentration parameter is small. This is consistent with the ASF estimates reported in Figure~\ref{sim:fig}. In the figure, the ASFs of the 2SCMLE shift toward those of the probit estimator as $\mu^2$ decreases.

In Table~\ref{tab1},  the TTSLS and the TRCMLE produce similar estimation results. However, the TTSLS always exhibits rejection rates above the nominal level. Furthermore, as reported in Figure~\ref{sim:fig},  the TTSLS does not provide good ASF estimates, especially at the boundaries of the support of $y_2$, whereas the TRCMLE yields estimates close to the true values even at the boundaries.

%In Table~\ref{tab1}, when the first-stage signal is dense, both the median bias and median absolute deviation of our estimators decrease, which means that our estimators do not depend on how sparse the first-stage signal is. It should be pointed out that our selection criterion for the regularization parameter is not optimally designed for the case where instruments are weak. Our estimators would exhibit better performance, if we were to apply a selection criterion for $\alpha$ that is robust to weak instruments. Nevertheless, in the table, our estimators obtained with the regularization parameter chosen from the cross-validation procedure report small bias and excellent size control.

Next, we examine the performance of the estimators under the same DGP but with different values of $K$, $\mu^2$ and $s$. In the first experiment, we change $K$ from 10 to 100 while keeping $s$ and $\mu^2$ fixed at $ 0.5$ and $ 50$, respectively. Secondly, we set $K=50$ and $\mu^2 = 50$, and change $s$ from 0.1 to 1. Lastly, we fix $K= 50$ and $s = 0.5$, and change $\mu^2$ from 25 to 250. The sample size is  200. We focus on three estimators: the TRCMLE, 2SCMLE, and probit estimator.

The estimation results are reported in Figure~\ref{fig}. The TRCMLE appears to have  the smallest median bias and correct size in all the considered cases. On the other hand, as $K$ increases, the 2SCMLE and the probit estimator tend to perform similarly. This may not be surprising, since in the linear model, it is known that the TSLS  shifts  toward the OLS estimator as $K$ increases. \violet{A similar convergence is observed as $\mu^2$ decreases, which suggests that \citepos{rivers1988limited} estimator may not be reliable when the concentration parameter is small.}
\subsection{Factor Model\label{sec:factor}}
Here, we consider the case where only  $\widetilde Z_i \in \mathbb R^{\tilde{K}}$ is observed rather than the   true instrument $Z_i \in\mathbb R^K$. Let $M$ be a $ \widetilde K \times K$ matrix consisting of  elements that are randomly drawn from $\text{Unif}[-1,1]$. This matrix is fixed across simulations, and  let $\widetilde Z_i =M  Z_i + \widetilde V_i$ where $\widetilde V_i \sim_{\text{iid}}\mathcal N  (0, \widetilde\sigma ^2 \mathcal I_{\widetilde K}  )$ and $\widetilde\sigma  = 0.3$. The measurement error $\widetilde V_i$ is independent of $u_i$ and $V_i$. The parameters are set as follows: $K$=5, $\widetilde{K}=100$, $s=1$, $\rho_z = 0$, and $\sigma_z = 1$. We consider two different values of $\mu^2$: 30 and 60.

\begin{table}[h!] %\rule[6pt]{1\textwidth}{1.25pt} 
	\caption{\small{Simulation results (Factor model, $\widetilde{K}=100$, $K = 5$)\label{tab:fac}}}
	\vspace{-.5em}	{\footnotesize
		\begin{tabular*}{1\textwidth}{@{\extracolsep{\fill}}lllrccrcc}
		\midrule &&&\multicolumn{3}{c}{$\mu^2=30$}&\multicolumn{3}{c}{$\mu^2=60$}\\ \cmidrule{4-6}\cmidrule{7-9}
 	$n$& &&Med.Bias&MAD&RP&Med.Bias&MAD&RP\\ \midrule
	\multirow{6}{*}{$200$}& &TRCMLE & -0.085 & 0.235 & 0.086 & -0.055 & 0.190 & 0.073 \\ 
	&&SCRCMLE & -0.015 & 0.241 & 0.038 & -0.012 & 0.187 & 0.039 \\  
	&&Inf.2SCMLE & -0.062 & 0.228 & 0.059 & -0.027 & 0.184 & 0.045 \\ 
	&&2SCMLE & -0.637 & 0.099 & 0.968 & -0.576 & 0.097 & 0.945 \\ 
	&&Probit & -0.715 & 0.071 & 0.792 & -0.693 & 0.073 & 0.841 \\ 
	&&TTSLS & -0.093 & 0.220 & 0.121 & -0.064 & 0.169 & 0.122 \\  \midrule
	\multirow{6}{*}{$400$}
	&&TRCMLE & -0.080 & 0.216 & 0.062 & -0.036 & 0.165 & 0.060 \\ 
	&&SCRCMLE & -0.044 & 0.220 & 0.042 & -0.030 & 0.164 & 0.051 \\ 
 	&&Inf.2SCMLE & -0.080 & 0.212 & 0.058 & -0.037 & 0.164 & 0.055 \\ 
	&&2SCMLE & -0.625 & 0.092 & 0.989 & -0.544 & 0.090 & 0.973 \\ 
	&&Probit & -0.738 & 0.049 & 0.953 & -0.725 & 0.050 & 0.971 \\ 
	&&TTSLS & -0.074 & 0.216 & 0.102 & -0.042 & 0.159 & 0.097 \\ \midrule
	\end{tabular*}
}	 \vspace{-1.5em}	\flushleft{\scriptsize{ Notes: The simulation results based on 2,000 replications are reported. Each cell reports the median bias (Med.Bias), median absolute deviation (MAD), and rejection probability at 5\% significance level (RP).} } \vspace{-0.8em}
\end{table}

 %In Appendix~\ref{add.table} in the Online Supplement, we report the simulation results computed with $30$ auxiliary instruments.   

The Inf.2SCMLE is computed with $Z_i$, while the others are computed using $\widetilde Z_i$. The simulation results are reported in Table~\ref{tab:fac} (the estimated ASFs are similar to those in Figure~\ref{sim:fig} and reported in the Online Supplement). The results closely  align with those in Section~\ref{sec:sim1}; our estimators exhibit considerably smaller median bias and reasonable rejection rates compared to the others. The SCRCMLE tends to have a larger bias but a smaller MAD with a larger sample size, suggesting a potential bias-variance tradeoff. The Inf.2SCMLE, computed with $Z_i$, performs considerably well, whereas the 2SCMLE exhibits a substantially large median bias and poor size control.

\section{Empirical Example: \cite{Bastian2018}\label{sec:emp}}
Recent studies on the EITC have shown that it has a considerable impact on   its recipients (see \citealp{meyer2001welfare}; \citealp{Eissa2004}) and their children (see \citealp{Dahl2012}). %For example, as discussed by \cite{meyer2001welfare} and \cite{Eissa2004}, EITC exposure increases maternal labor force participation of low-income families. As another example, \cite{Dahl2012} show that the recipients' children have more resources to improve their academic achievement. 
 In particular, \cite{Bastian2018} provided empirical evidence that increasing  EITC exposure positively impacts family earnings, consequently leading to better long-term academic achievement for children. %To be specific, \cite{Bastian2018} conclude that an increase in family income during an individual's adolescent years has a positive effect on the likelihood of achieving a better academic and employment outcome in the individual's early adulthood. 
 We revisit their work and complement it as follows: First, we reexamine the weakness of their instruments. The first-stage F statistics reported in \citet[Table 5]{Bastian2018} range between 3.2 and 5.1, suggesting that their instruments are potentially weak. However, this measure may not be perfect considering the nonlinear nature of the binary response model. Thus, we use the test proposed by \cite{frazier2020weak} to study their weakness. Second, we utilize additional instruments to increase estimation efficiency and  account for heterogeneous effects of EITC exposure on family income depending on the state of residence.  Lastly, we employ a parametric binary response model instead of the linear model; as detailed in Example~\ref{example2}, this approach is expected to produce  better estimation results from a theoretical perspective.  \label{r2p1ab}
 
% We complement \citepos{Bastian2018} results as follows. First, we employ the binary response model and use our estimators. Second, we utilize additional instruments to increase estimation efficiency and account for heterogeneous effects of EITC exposure on family income depending on the state of residence. Lastly, we examine the weakness of the suggested instruments in the context of the binary reponse model by using the test proposed by \cite{frazier2020weak}.

We use \citepos{Bastian2018} data from the 1968-2013 waves of the Panel Study of Income Dynamics.    Our model  is similar to theirs and is given as follows:
\begin{equation} 
 y_i =  {1} \{  \beta_0 + I_{i}'\beta_1  + W_{1i} ' \beta_2 + W_{2i} '\beta_3 \geq u_i  \} ,\quad\quad I_{i} = \pi_{0} + \widetilde Z_i ' \pi_1 + W_{1i} ' \pi_{2} + W_{2i} ' \pi_{3} + V_{i} ,\label{model:eitc}
\end{equation}
where $I_i = ( I_{i,\text{(0-5)}} , I_{i,\text{(6-12)}} ,  I_{i,\text{(13-18)}})'$ and each $I_{i,(a)}$ is the family income of individual $i$ at  age interval $a$. The outcome variable $y_i$ is 1 if individual $i$ is a college graduate and 0 otherwise.\footnote{College graduation is assessed when individuals reach the age of 26.} $W_{1i}$ is an 11-dimensional vector of personal characteristics: age, age square, the number of siblings at age 18, and indicators for black, Hispanic, female, ever-married parents, and whether the individual's mother and father completed high school or are at least some college-educated. $W_{2i}$, which is measured at age 18, is state-by-year  economic indicators: per capita GDP, the unemployment rate, the top marginal income tax rate, the minimum wage, maximum welfare benefits, spending on higher education, and tax revenue. All the continuous variables are standardized before analysis.

We consider two models with different instruments. In Model 1, we follow \cite{Bastian2018} and let $\widetilde Z_i$ be $(\text{EITC}_{i,(\text{0-5})}, \text{EITC}_{i,(\text{6-12})}, \text{EITC}_{i,(\text{13-18})})'$, where $\text{EITC}_{i,(a)}$ is the standardized measure of EITC exposure of individual $i$ at age interval $a$.\footnote{The measures of EITC exposure and family income are in thousands of 2013 dollars and discounted at a 3\% annual rate from age 18. A detailed description can be found in \cite{Bastian2018}.} As EITC exposure is measured in childhood and adolescent years, this could affect children's long-term educational attainment only through $I_i$; see \cite{Bastian2018} for more discussion on its validity. In Model 2, we use 144 instruments, consisting of the measures of EITC exposure and their interactions with 47 state-of-residence dummies (observed in the sample). These additional instruments are expected to help us address the potential bias associated with the heterogeneous effects of EITC transfers on family income caused by, such as, the state-specific cost of living.   The sample size is 2,654.

We first apply \citepos{frazier2020weak} distorted J test to study if our instruments are weak in the sense of \cite{Staiger1997}. This is needed to ensure the conditions in Theorem~\ref{prop4:finite}. However, their test is not designed for the case where (i) the (sample) covariance of instruments is nearly singular and (ii) multiple endogenous variables are present. Hence, we apply the test only to Model~1 separately for each age interval $a$, after replacing $I_i$ in \eqref{model:eitc} with $I_{i,(a)}$. The computed distorted J tests are 37.15, 21.81, and 42.76 for the age intervals 0-5, 6-12, and 13-18, respectively. The critical value at 5\% significance level is around $  9.48 $.\footnote{To conduct the test, $\delta_n$, $a_i$ and $b_i$ in \cite{frazier2020weak} are set to $\widehat{\psi}/\log(\log(n))$, $( I_{i,(a)}, W_{1i}', W_{2i} ', \widetilde Z_i' , 0_{k}')'$, and $(0_{k}',1, W_{1i}', W_{2i}', \widetilde Z_i')'$ where $k = \dim(W_{1i})+\dim(W_{2i})+\dim(\widetilde Z_{i})+1=23$.} Therefore, the   tests suggest that our model is not weakly identified in the sense of \cite{Staiger1997}. 

% We now present the estimation results.
 We consider the TRCMLE and the SCRCMLE. Their regularization parameters are chosen as in Section~\ref{sec:sim}. Specifically, $\alpha$ is chosen from 25 equally spaced points between $ \overline{\texttt{c}}_a n^{-0.6} 10^{-4}$ and $\overline{\texttt{c}}_a n^{-0.6}  10^{-1}$.   We also consider \citepos{rivers1988limited} 2SCMLE  and the probit estimator. In this section, we further consider another estimator similar to the post-Lasso estimator  proposed by \cite{Belloni}. This estimator is included  due to its popularity in empirical research. To compute it, we first select instruments with non-zero first-stage coefficients using the Lasso procedure. Then, we refit the first stage with the selected instruments and compute the first-stage fitted values and residuals. Lastly, a linear model is estimated using the fitted values and residuals as regressors.  The results are reported in  the column labeled ``Lasso'' in Table~\ref{eitc:tab2}.

 \begin{table}[h!]
 %	\rule[6pt]{1\linewidth}{1.25pt}  	
 	\caption{\small{Effects of family earnings on college completion}}
 	\label{eitc:tab2}
 \vspace{-.5em}	{\footnotesize
 	\begin{tabular*}{\textwidth}{@{\extracolsep{\fill}}lllllll} \midrule 
 &Variables&TRCMLE&SCRCMLE&2SCMLE&Lasso&Probit\\\midrule
 %\multicolumn{6}{l}{Model 1 \small{($\dim(\widetilde Z_i)=3$)}}\\
 \multirow{7}{*}{\shortstack[l]{Model 1\\\footnotesize{($\dim(\widetilde Z_i)=3$)}}}& \multirow{2}{*}{$I_{i,\text{(0-5)}}$}& 1.5237$^{***}$ & 1.5032$^{***}$ & 1.5032$^{***}$ & 0.2720$^{***}$ & 0.1309$^{***}$ \\ 
 & &(0.5502) & (0.5429) & (0.5429) & (0.0646) & (0.0484) \\  
  &\multirow{2}{*}{$I_{i,\text{(6-12)}}$  } & -2.9299$^{***}$ & -2.8910$^{***}$ & -2.8910$^{***}$ & -0.4324$^{***}$ & 0.2022$^{***}$ \\ 
  &&(1.1011) & (1.0809) & (1.0809) & (0.1217) & (0.0699) \\ 
 &\multirow{2}{*}{$I_{i,\text{(13-18)}}$}& 2.2176$^{**}$ & 2.2188$^{**}$ & 2.2188$^{**}$ & 0.3332$^{***}$ & 0.0530$^{}$ \\ 
&& (0.9601) & (0.9477) & (0.9477) & (0.0968) & (0.0502) \\  \cmidrule{2-7} 
& {Exogeneity Test}&10.7426$^{***}$ & 10.6553$^{***}$ & 10.6553$^{***}$ & 24.9366$^{***}$ &   \\ \midrule
 %\multicolumn{6}{l}{Model 2 \small{($\dim(\widetilde Z_i)=144$)}}\\
 \multirow{7}{*}{\shortstack[l]{Model 2\\\footnotesize{($\dim(\widetilde Z_i)=144$)}}} & \multirow{2}{*}{$I_{i,\text{(0-5)}}$}&  1.5519$^{***}$ & 1.5330$^{***}$ & 0.4107$^{**}$ & 0.1750$^{***}$ & 0.1309$^{***}$ \\ 
&& (0.5346) & (0.5343) & (0.1867) & (0.0435) & (0.0484) \\ 
&\multirow{2}{*}{$I_{i,\text{(6-12)}}$}&   -2.7763$^{***}$ & -2.7998$^{***}$ & -0.3722$^{}$ & -0.2894$^{***}$ & 0.2022$^{***}$ \\ 
&&(0.9543) & (1.0032) & (0.2738) & (0.0805) & (0.0699) \\  
&\multirow{2}{*}{$I_{i,\text{(13-18)}}$}&  2.1080$^{***}$ & 2.1979$^{***}$ & 0.5260$^{**}$ & 0.2807$^{***}$ & 0.0530$^{}$ \\ 
&&(0.8114) & (0.8531) & (0.2138) & (0.0700) & (0.0502) \\ \cmidrule{2-7}
 &{Exogeneity Test}&14.4770$^{***}$ & 12.9709$^{***}$ & 10.0550$^{***}$ & 29.7836$^{***}$ &  \\ \bottomrule
 	\end{tabular*} 
 } \vspace{-1.5em}	\flushleft{\scriptsize{Notes: Standard errors of coefficient estimates are reported in parentheses. Significance $^{***}$ $p<0.01$, $^{**}$ $p<0.05$,  $^{*}$ $p<0.1$. \textcolor{redby}{$\alpha$ is chosen to 0.00404 (TRCMLE) and $4.038\times 10^{-6}$ (SCRCMLE) in Model 1 and 0.00411 (TRCMLE) and $3.429\times 10^{-3}$ (SCRCMLE) in Model 2.} } }\vspace{-0.8em}
 \end{table}

  In Table~\ref{eitc:tab2}, our and \citepos{rivers1988limited} approaches yield similar estimates in Model~1. However, substantial discrepancies arise in Model 2. While our estimators exhibit similar coefficient estimates in both models,  \citepos{rivers1988limited} estimates tend to converge toward the probit estimates in Model 2.  This is consistent with the findings in Section~\ref{sec:sim}.  Moreover,  all the instrumental variable estimators report smaller standard errors in Model 2, suggesting that the large set of instruments in Model 2 contains additional information useful for predicting family income.
  
  Similar observations can be found in Figure~\ref{fig:eitc}, which reports the estimated probability of college completion depending on  family income for each age interval. For clarity, let $\mu_{{a}}$ (resp.\ $ \sigma_{{a}}$) denote the sample mean (resp.\ standard deviation) of $I_{i(a)}$. Figure~\ref{fig:eitc} shows how the probability changes as  $I_{i(a)}$ varies from $\mu_{{a}}-\sigma_{a}$ to  $\mu_{{a}}+\sigma_{{a}}$, while other variables remain fixed at their medians. There is no significant change in the probabilities computed with the TRCMLE in both models. In contrast, the probability computed with the 2SCMLE substantially  shifts toward that of the probit estimator in Model 2.  Similar results are observed in the APE estimates  in the Online Supplement.

Next, we focus on the Lasso estimates. As mentioned in Section~\ref{sec:sim}, parameters in linear probability models are different from those in nonlinear binary response models. Hence, to make a reasonable comparison of the coefficient estimates, we compute the APEs at the sample mean of $Y_{2i}$ and report them in Table~\ref{eitc:tab3} in Appendix~\ref{add.table} of the Online Supplement. However, even in the table, Lasso estimates tend to be biased toward probit estimates in Model 2. This bias may be attributed to factors discussed in Section~\ref{sec:lasso} and to a slower decreasing rate of the eigenvalues of $\mathcal K_n$.

  Table~\ref{eitc:tab2} reports the exogeneity testing results conducted as detailed in Remark~\ref{rem:tmp2}. The $p$-values are computed from the  $\chi^2(3)$ distribution. The results indicate the presence of endogeneity at the 1\% significance level regardless of the number of instruments and estimation methods. This explains the difference in coefficient estimates   between the probit estimator and the others.  

%In Table~\ref{eitc:tab2}, we first note that the estimation results of the TRNLSE and SCRNLSE are robust to the number of instruments, while the R\&V estimate is substantially biased toward the Probit estimate in the second model. These are consistent with our simulation results in Section \ref{sec:sim}. Although, in Table~\ref{eitc:tab3}, it seems that our estimators obtained using the Spectral cut-off is also biased toward the Probit estimator as the dimension of instruments increases, this bias is smaller than that of the R\&V estimator. Lastly, all the instrumental variable estimators report smaller standard errors in the second model, which would be related to improvements in estimation efficiency, caused by the use of many instruments.

The estimation results in Table~\ref{eitc:tab2} are similar to those in \cite{Bastian2018};  increasing family income between ages 13 and 18 positively influences early adulthood college completion. While the authors found a similar effect, their estimates, reported in column 3 of their Table 6, are not significant. This discrepancy may arise from different models and estimation methods, such as the absence of the large number of fixed effects in our approach.  Table~\ref{eitc:tab2} also reports a positive (negative) impact of an increase in family income between birth and age 5 (ages 6 and 12) on college completion. \cite{Bastian2018} reported similar mixed coefficient signs and noted that these may be attributed to cross-age correlations in EITC exposure.

%Figure~\ref{fig:eitc} reports the estimated probability of college completion depending on each $I_{i,(a)}$. To be specific, we estimate the probability of college completion by changing $I_{i,(a)}$ from its minimum to 3rd quartile, while holding family earnings at the other age intervals to their medians. Figure~\ref{fig:eitc} clearly shows that the R\&V estimator is biased toward the Probit estimator as the number of instruments increases. In addition, according to the figure, an individual, whose family income is substantially high in the individual's adolescent years, has a higher probability of college completion as of age 26.  

\section{Conclusion \label{sec:con}}

This paper proposes an estimation procedure to resolve the issue of many (nearly) weak instruments in endogenous binary response models. The proposed estimators are easy to compute and have good asymptotic properties.  Monte Carlo studies suggest that our estimators outperform the existing estimators  when there are many (nearly) weak instruments. We  apply the estimators to study the effect of family income in childhood and adolescent years on college completion. %This example shows that our estimators  are robust to many instruments. %The method in this paper could be extended to semi-parametric models such as that proposed by \cite{Klein_Spady1993} or \cite{Rothe2009}. We also expect that the theory developed in this paper can be generalized to other nonlinear models when the  moment conditions  of interest depend on a possibly infinite-dimensional random variable. 

\onehalfspacing

\begin{figure}[h!]
	\begin{center} 
		\caption{\small{Estimated average structural functions (Gaussian instruments) \label{sim:fig}}}
		\vspace{-0.5em}	\includegraphics[width=\textwidth, height=.25\textheight]{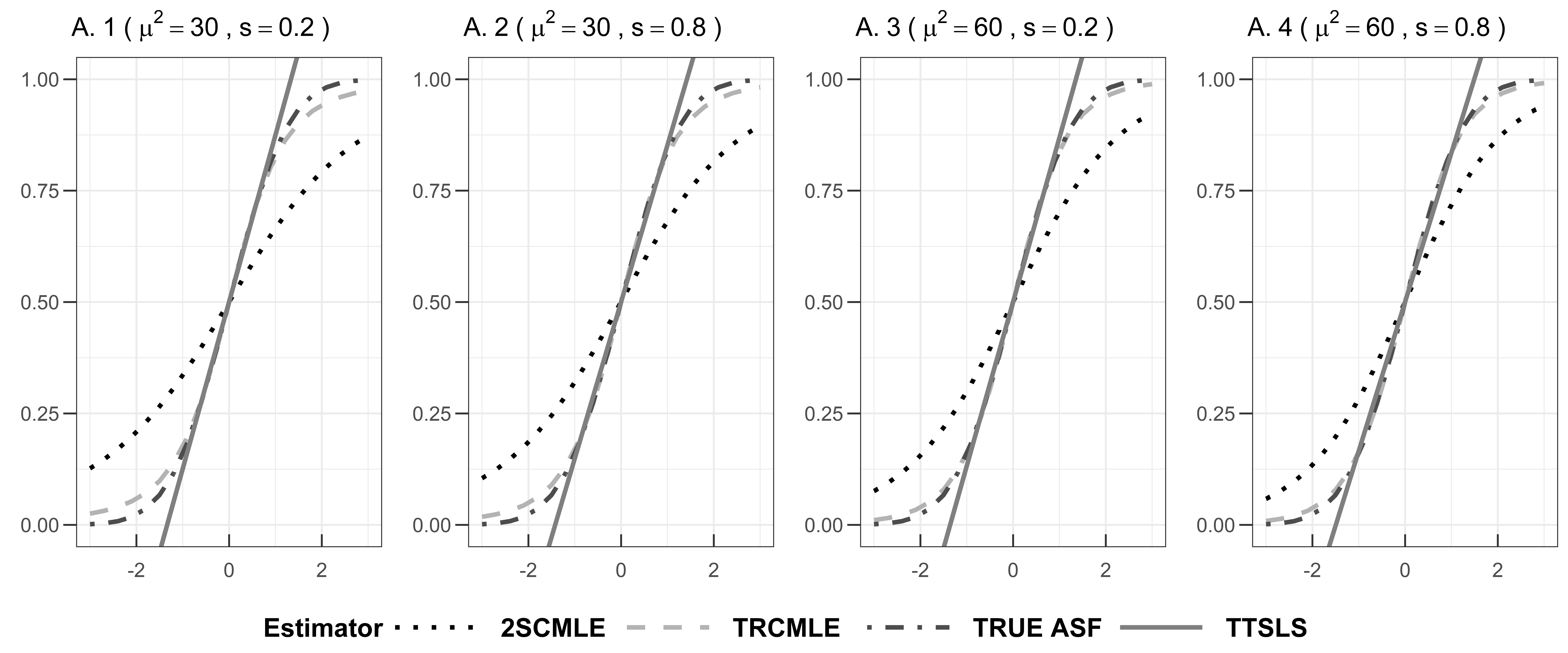}
	\end{center}
	\vspace{-1em}	\footnotesize{Notes: The simulation results based on 2,000 replications are reported. $K = 50$ and $n=200$.} \vspace{-2em}
\end{figure}

\begin{figure}[h!]
	\caption{\small{Simulation results for different parameter values}}
	\label{fig} 
	\vspace{-0.5em}\begin{subfigure}{\textwidth}
		{\subcaption{\footnotesize{Median Bias}}}
		\vspace{-0.5em}		\includegraphics[width=\linewidth, height=.15\textheight]{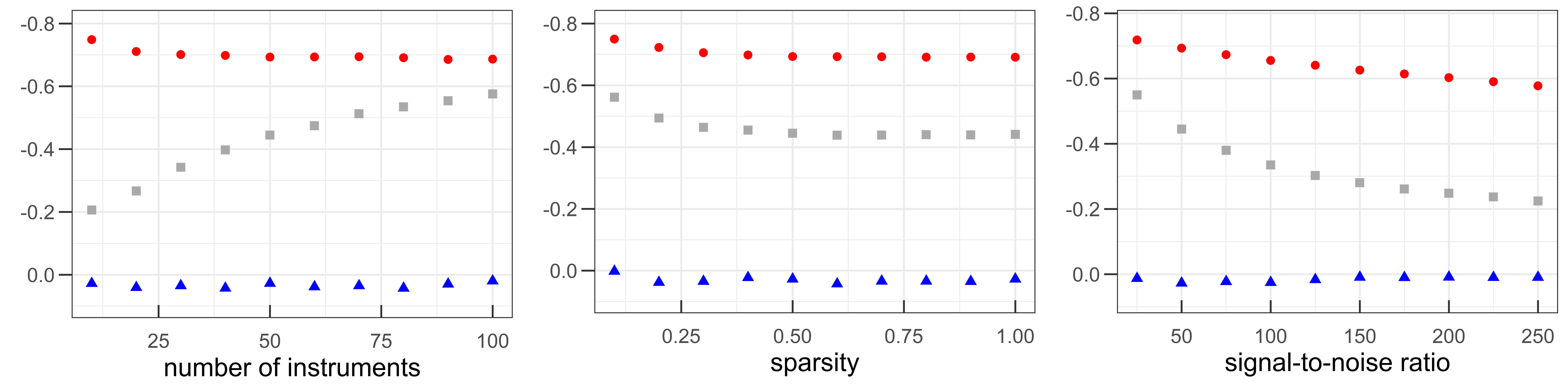} 
	\end{subfigure} \\
	\begin{subfigure}{\textwidth}
		\subcaption{\footnotesize{Median Absolute Deviation}}
		\vspace{-0.5em}	\includegraphics[width=\linewidth, height=.15\textheight]{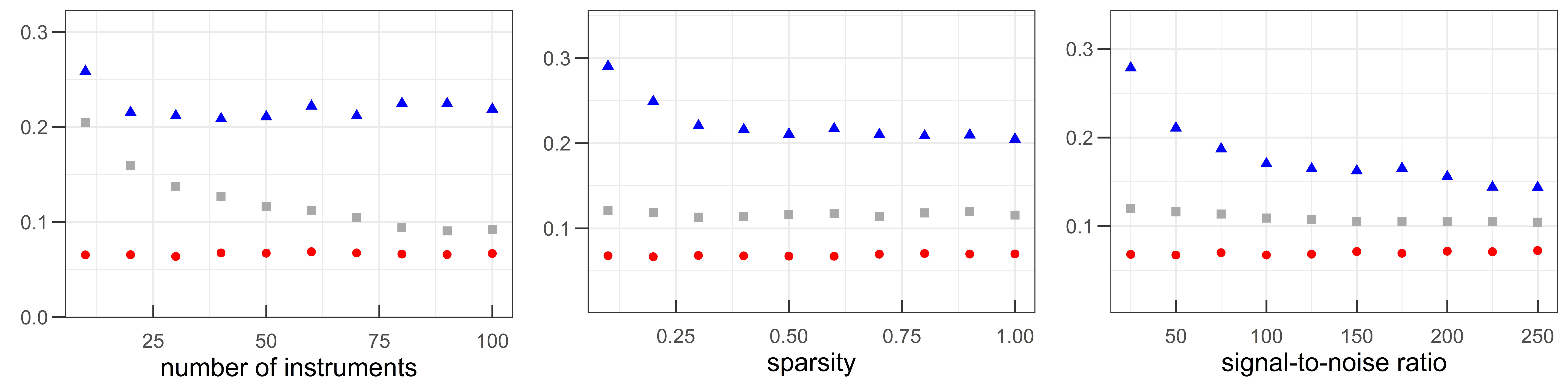} 
	\end{subfigure} \\
	\begin{subfigure}{\textwidth}
		\subcaption{\footnotesize{Rejection Probability at 5\% significance level}}
		\vspace{-0.5em}	\includegraphics[width=\linewidth, height=.15\textheight]{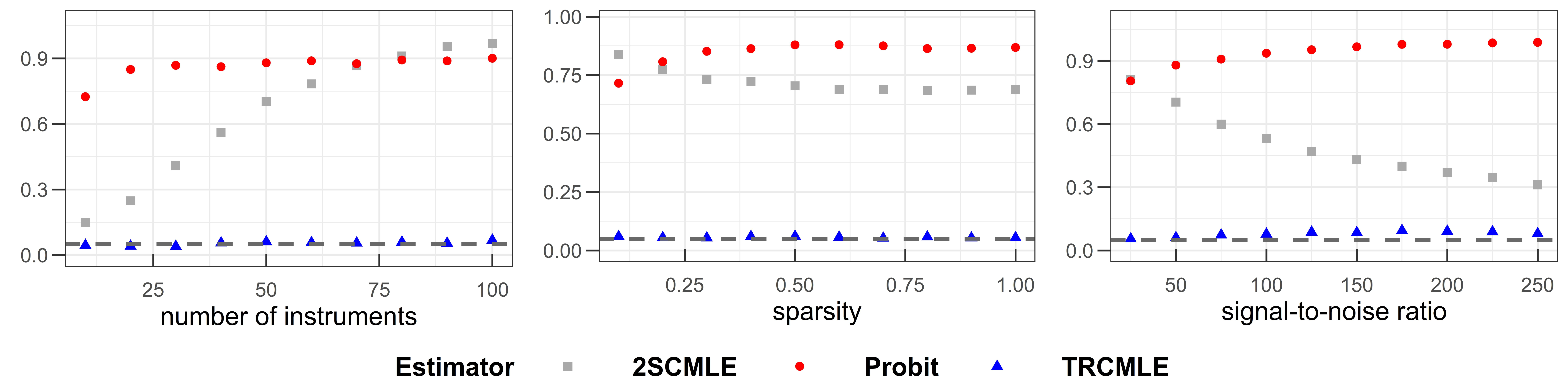} 
	\end{subfigure} 
	\footnotesize{Notes: The simulation results based on 2,000 replications are reported. The default values of $K$, $s$ and $\mu^2$ are $50$, $ \lfloor 0.5\rfloor$, and $50$. The dashed lines indicate the nominal size  0.05.} \vspace{-2em}
\end{figure}

\begin{figure}
	\caption{\small {Estimated average structural functions}}
	\vspace{-.5em}	\includegraphics[width=\textwidth , height = .25\textheight]{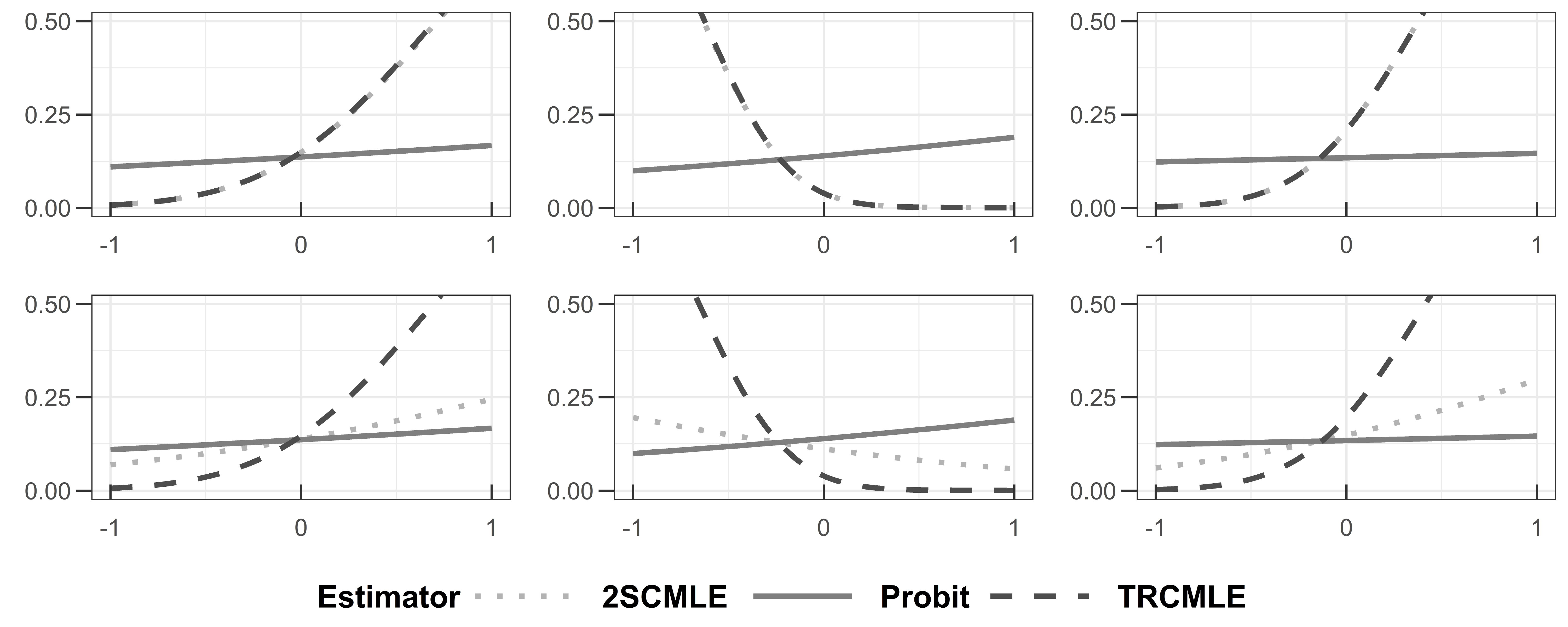}\vspace{-1em}\label{fig:eitc}
	\vspace{-1em}	\flushleft{\footnotesize{Notes: Each figure reports the estimated average structural function (ASF) that reports the probability of college completion according to family earnings in the age interval 0-5 (left), 6-12 (middle), and 13-18 (right).}} \vspace{-2em}
\end{figure}
\clearpage
\newpage
	\begin{center}
	\Large{The Online Supplement to\\ ``Binary response model with many weak instruments''}
	\end{center} 

\appendix
\makeatletter %remove leading `A \quad` from section headers
\def\@seccntformat#1{\@ifundefined{#1@cntformat}
	{\csname the#1\endcsname\quad}
	{\csname #1@cntformat\endcsname}}

\newcommand\section@cntformat{}
\newcommand\subsection@cntformat{}
\makeatother 

\vspace{1em}

\section{Appendix \ref*{stepwise}: Computation of estimators and the choice of $\alpha$ \label{stepwise}}
\violet{In this section, we discuss how to compute the proposed estimators and the regularization parameter. We first focus on the computation of the estimators with a prespecified $\alpha$. The following computation method is applicable regardless of the dimension of $Z_i$ by using different inner products according to $\mathcal H$; for example, when $\mathcal H = \mathbb R^{d_z}$, the inner project $\langle h_1, h_2\rangle_{\mathcal H}$ will be replaced by $h_1 ' h_2$.}
\begin{center}\begin{minipage}{.95\textwidth}
		\begin{enumerate}[\textbf{Step }1]
			\item \label{stp1}	\violet{Compute the $n \times n$ matrix $\mathbf{M}$ whose $(i,j)$th element is given by $\langle Z_i, Z_j\rangle_{\mathcal H}/n$. }
			\item \label{stp2} \violet{Obtain the eigenvalues and eigenvectors, $\{\widehat{\kappa}_j, \widehat{\varpi}_j\}_{j=1} ^n$, of the matrix $\mathbf{M}$. }
			%	\item Choose the parameter space of $\alpha$, $\mathtt{C}_\alpha$, satisfying the conditions in Theorems~\ref{prop4:finite} and \ref{prop4}. For example, one may set $\mathtt{C}_\alpha = n^{-0.6} \times \mathtt{C}$ for some set of small numbers $\mathtt{C}$.
			\item \violet{For a prespecified $\alpha$, compute the $n \times n$  matrix $
				\mathcal P_{n\alpha} = \sum_{j=1} ^n {q} (\widehat\kappa_j, \alpha) \widehat{\varpi}_j\widehat{\varpi}_j'$. Then, compute $\widehat{g}_{i,\alpha}$ as the $i$th row of $ (\mathcal P_{n\alpha} \mathbf{Y}_2 , (\mathbf{I}_{n} - \mathcal P_{n\alpha}) \mathbf{Y}_2 )$, where $\mathbf{I}_n$ and $\mathbf{Y}_2$ denote the identity matrix of dimension $n$ and the $n \times d_e$ matrix whose $i$th row is given by $Y_{2i}$. } \label{stp3}
			\item \label{stp4}\violet{ Compute the estimator by using $\widehat{g}_{i,\alpha}$ in \textcolor{black}{\ref{stp3}} and the objective function in \eqref{model:5eq}.  }
		\end{enumerate}
\end{minipage}\end{center}

\violet{The last step involves the standard nonlinear optimization procedure provided in most statistical software programs. The choice of $\alpha$ may be accompanied by the estimation of $\widehat{\theta}_j$ to minimize its MSE studied in Section~\ref{sec: mse}. In  Sections~\ref{sec:sim} and \ref{sec:emp}, we used the following  to compute $\alpha$.}
\begin{center}\begin{minipage}{.95\textwidth}
		\begin{enumerate}[\textbf{Step' }1]
			\item 	\violet{Set  $\mathtt{C}_\alpha$, the parameter space of $\alpha$, with taking into account the conditions in Theorem~\ref{prop4}. For example, one may set $\mathtt{C}_\alpha = n^{-0.6} \times \mathtt{C}$ with $\mathtt{C}$ being a set of small positive numbers. }
			\item \violet{For each $\alpha \in \mathtt{C}_\alpha$, compute the bound \begin{equation}
					\mathtt{c}_{g}
					\left(  \frac{  h'  {\Upsilon}' (\mathcal P_{n\alpha} - \mathcal I)^2  {\Upsilon}h }{n} + \sigma_{v_{h}} ^2\frac{\text{tr}(\mathcal P_{n\alpha} ^2) }{n} \right)  , \label{eq: mse: tmp}
				\end{equation}for some  $h \in \mathbb R^{d_e}$.  The scalars $\sigma_{v_h} ^2$ and $\mathtt{c}_{g} $ denote $h'\Sigma_Vh$ and  $ \lambda_{\max}\left(\mathbb E[  \dot{m}_{2j} ^2 (\theta_0, g_i) g_{i} g_{i}'|x] \right)$.  The $i$th row of $n\times d_e$ matrix $ {\Upsilon}$ is equal to ${\Pi}_n Z_i$. The term in the parenthesis of \eqref{eq: mse: tmp} is similar to the conditional MSE in \citet[Proposition 2]{Carrasco2012} whose estimators are available therein. The constant  $\mathtt{c}_g$ is estimable from its  sample counterpart.}\label{stp'3}
			\item \violet{Choose $\alpha$ that minimizes the bound computed in \ref{stp'3}.}
		\end{enumerate} 
\end{minipage}\end{center} 

\violet{It would be worth noting that \eqref{eq: mse: tmp} is related to the upper bound of the conditional MSE, particularly that of \eqref{eq: ahat}. To see this in detail, a few additional notations are required. Thus, we defer its detailed discussion  until after \eqref{eq: ahat} in Appendix~\ref{app: mse}. }

\section{Appendix \ref*{sec: func}: Additional simulation study with function-valued instruments }\label{sec: func}
We consider a scenario where the functional form of $\mathbb E[Y_{2i}|x_i]$ is unknown, and to estimate it, the continuum of moments is used as $Z_i$. Specifically, we replace the first stage in \eqref{eq:sim} by\begin{equation*}
	Y_{2i} = \pi z_{1i} + \pi f(z_{2i}) + v_i,
\end{equation*}
where $z_{1i}\sim_{\text{iid}}\mathcal N(0,1)$, $z_{2i}\sim_{\text{iid}}\text{Unif}[0,1]$ and  $f(\cdot)$ is set to the probability density function of the beta distribution with shape parameters 2 and 5. The constant $\pi$ is computed as  before,  replacing $\Sigma_Z$ by its estimates calculated using 1,000 samples of $z_{1i}$ and $ f(z_{2i})$. 

Our estimators and the TTSLS are computed using  $\exp( tz_{2i})$\footnote{Since $z_{2i}$ takes values on a compact interval, this function is bounded and square integrable.} as $Z_i $ where $t$ takes values in 100 equally spaced points between -5 and 5. {We consider the weighted inner product discussed in a footnote on page~\pageref{ft: wl2}. The weight is set to the standard normal density function.} The Inf.2SCMLE is calculated with the realized values of $f(z_{2i})$, and the 2SCMLE, which requires the linear first stage, is not considered in this example. The concentration parameter is computed using the true   $f(z_{2i})$ and we consider two values: 60 and 180. Since we do not use the true values of $f(z_{2i})$, the  signal from $Z_i$ may be weaker than what the concentration parameter suggests.

\begin{table}[h!] %\rule[6pt]{1\textwidth}{1.25pt} 
	\caption{\small{Simulation results (Functional instrument of the continuum of moments)\label{tab:fun}}}
	\vspace{-.5em}	 	{\footnotesize	
		\begin{tabular*}{1\textwidth}{@{\extracolsep{\fill}}lllrccrcc}
			\midrule &&&\multicolumn{3}{c}{$\mu^2=60$}&\multicolumn{3}{c}{$\mu^2=180$}\\ \cmidrule{4-6}\cmidrule{7-9}
			$n$& &&Med.Bias&MAD&RP&Med.Bias&MAD&RP\\ \midrule
			\multirow{6}{*}{$200$}& &TRCMLE & 0.049 & 0.344 & 0.045 & -0.014 & 0.230 & 0.078 \\ 
			&&SCRCMLE & 0.001 & 0.287 & 0.046 & -0.032 & 0.204 & 0.058 \\ 
			&&Inf.2SCMLE & 0.035 & 0.239 & 0.037 & 0.028 & 0.181 & 0.043 \\ 
			&&Probit & -0.753 & 0.074 & 0.633 & -0.748 & 0.079 & 0.605 \\ 
			&&TTSLS & 0.033 & 0.341 & 0.137 & -0.047 & 0.212 & 0.155 \\ 
			\midrule 
			\multirow{6}{*}{$400$}& &TRCMLE & 0.057 & 0.341 & 0.032 & -0.018 & 0.185 & 0.068 \\ 
			&&SCRCMLE & -0.020 & 0.265 & 0.038 & -0.030 & 0.168 & 0.054 \\ 
			&&Inf.2SCMLE & -0.001 & 0.218 & 0.040 & 0.004 & 0.143 & 0.047 \\ 
			&&Probit & -0.754 & 0.051 & 0.918 & -0.759 & 0.052 & 0.896 \\ 
			&&TTSLS & 0.066 & 0.344 & 0.120 & -0.036 & 0.175 & 0.124 \\ \bottomrule
	\end{tabular*}} \vspace{-1.5em}	\flushleft{\scriptsize{Notes: The simulation results based on 2,000 replications are reported. Each cell reports the median bias (Med.Bias), median absolute deviation (MAD), and rejection probability at 5\% significance level (RP).}} \vspace{-0.8em}
\end{table}

Table~\ref{tab:fun} summarizes  simulation results. For all the cases considered in the table, the SCRCMLE performs as good as the Inf.2SCMLE that is computed using the realized values of $f(z_{2i})$. In contrast, the two Tikhonov regularized estimators report  a relatively large bias and MAD especially when $\mu^2 = 60$. This may be related to the fact that (i) the eigenvalues of the sample covariance of $Z_i$ shrink at a relatively fast rate in this example and (ii) Tikhonov regularized estimators have a relatively slow convergence rate in those cases (see, e.g., \citealt[Remark 3.2]{Benatia2017} or \citealt[p.\ 389]{Carrasco2012}). However, in general, the regularized estimators perform as good as the Inf.2SCMLE. 

\section{Appendix \ref*{sec:add_emp}: Additional empirical illustration \label{sec:add_emp}}
\cite{Miguel2004} studied the effect of economic growth on the occurrence of civil war in sub-Saharan Africa by instrumenting the economic growth to the annual rainfall growth rate. As briefly mentioned in Example~\ref{example:miguel}, the model may provide better insight if the annual curve of daily rainfall growth is employed as an instrument, rather than its average. We leave this question for future study because of the current data availability. Instead, in this section, we use a function  of the annual rainfall growth as an instrument to demonstrate how the function-valued instrumental variable can be used in practice.

The model is given by:\begin{equation}
	\text{conflict}_{is} = 1\{    \text{growth}_{is-1}\beta +  W_{is}'\gamma \geq u_{is} \},\quad\quad\text{growth}_{is-1} = \Pi_{n} Z_{is-1} + W_{is} ' \pi_2 + V_{is}. \label{eq:add_emp}
\end{equation}
The dependent variable indicates if there was a civil conflict resulting in at least 25 deaths, and  $\text{growth}_{is-1}$ measures annual economic growth in country $i$ at time $s-1$. \cite{Miguel2004} found that the lagged economic growth plays an important role in the occurence of civil conflicts associated with the 25-death threshold. The vector $W_{is}$ contains the explanatory variables: the intercept, the log of GDP per capita in 1979, a lagged measure of democracy, ethnolinguistic fractionalization, religious  fractionalization and the indicator on if country $i$ exports oil. We refer  readers to \cite{Miguel2004} for details  on these variables and the data.

To compute our estimators, the instrument $Z_{is-1}$ is set to $ \exp(tx_{is-1})$ where $x_{is-1}$ is the lagged rainfall growth that was used as the instrument in \cite{Miguel2004}.  The weight function is given by the standard normal density and the function-valued random variable is measured on 100 equally spaced points between -5 and 5. All continuous exogenous variables are standardized, and $Z_{is}$ is centered before analysis. The regularization parameters are obtained as in Section~\ref{sec:emp}. We compare the TRCMLE and the SCRCMLE, computed with the  function-valued instrument, and compare them with  the naive probit estimator and   \citepos{rivers1988limited} estimator calculated using $x_{is-1}$ as an instrument.

\begin{table}[h!]
	%	\rule[6pt]{1\linewidth}{1.25pt}  	
	\caption{Estimation results for Appendix~\ref{sec:add_emp}}
	\label{mig:tab}{\footnotesize
		\begin{tabular*}{\textwidth}{@{\extracolsep{\fill}}lrrrr} \midrule 
			Variable&TRCMLE&SCRCMLE&2SCMLE&Probit\\\midrule 
			\multirow{1}{*}{$\text{growth}_{is-1}$} &-0.5136$^{}$ & -0.4930$^{}$ & -0.6056$^{}$ &   -0.0335$^{}$ \\ 
			&(0.4046) & (0.3967) & (0.4609) &  (0.0522) \\ \midrule
			{Exogeneity Testing}&1.4317$^{***}$ & 1.3647$^{***}$ & 1.5605$^{***}$ & \\ \bottomrule
	\end{tabular*} } \vspace{-1em}
	\flushleft{\scriptsize{Notes: Each cell reports point estimates of $\beta$ in \eqref{eq:add_emp}. The standard errors are reported in parentheses. The sample size is 743. The regularization parameters of the TRCMLE and SCRCMLE are respectively $1.38\times 10^{-6}$ and $1.54\times 10^{-5}$.}}
\end{table}  

Table~\ref{mig:tab} summarizes estimation results. The exogeneity testing results suggest the presence of endogeneity, which explains the significant difference between the estimates obtained of the probit estimates from the others. All instrumental variable estimators suggest that economic growth has a negative association with the occurrence of civil conflicts, although none of them are statistically significant. Similar results are obtained even when the dependent variable is specified to civil conflicts associated with deaths greater than 1,000. The insignificant results may be attributed to failure in addressing country-specific fixed effects that are not included in our model. The fixed effects are not included to mitigate the potential curse of dimensionality in the second stage. Despite the insignificance, the table provides a valuable insight on our estimators. For example, our estimators, computed using function-valued instruments, have smaller variance than the 2SCMLE.

\section{Appendix \ref*{sec:rem}: Additional remark\label{sec:rem}}
\begin{remark}\label{rem4} \normalfont
	Let $\ker \mathcal K$ and $(\ker \mathcal K)^{\perp}$ denote the kernel of $\mathcal K$ and its orthogonal complement space. Then, because $(\ker \mathcal K)^{\perp} = \cl (\ran \mathcal K)$, for any $h \in \ker \mathcal K$, \begin{equation*}\langle \mathcal V_{j} h,\mathcal V_{j} h\rangle =  \Vert \mathbb E  [ \dot m_{2j}  ( \theta_0 , g_i )   g_{0i} \langle Z_i,h\rangle_{\mathcal H} ]  \Vert ^2   \leq   \mathbb E  [   \Vert  \dot{m}_{2j}   (\theta_0 , g_i ) g_{0i}  \Vert ^2  ] \mathbb E [  \langle Z_i,h\rangle_{\mathcal H}^2   ]   = 0 ,   \end{equation*} 
	where the inequality follows from the Cauchy-Schwarz inequality. The last equality holds because $h \in \ker \mathcal K$. Thus, $\ker \mathcal K \subset \ker \mathcal V_{j}$, which implies $ \left(\ker \mathcal V_{j} \right)^{\perp} = \cl (\ran \mathcal V_{j}^\ast) \subset \cl ( \ran\mathcal K) = \left(\ker \mathcal K\right)^{\perp} $. 
\end{remark}  
\section{Appendix \ref*{add.table}: Additional Tables and Figures \label{add.table}}
\begin{table}[!htbp] %\rule[6pt]{1\textwidth}{1.25pt} 
	\caption{\small{Simulation results (Factor model with $\widetilde{K} = 30$)\label{tab:fac:add}}}
	\vspace{-.5em}	{\small
		\begin{tabular*}{1\textwidth}{@{\extracolsep{\fill}}lllrccrcc}
			\midrule &&&\multicolumn{3}{c}{$\mu^2=30$}&\multicolumn{3}{c}{$\mu^2=60$}\\ \cmidrule{4-6}\cmidrule{7-9}
			$n$& &&Med.Bias&MAD&RP&Med.Bias&MAD&RP\\ \midrule
			\multirow{6}{*}{$200$}& &TRCMLE & -0.056 & 0.236 & 0.052 & -0.025 & 0.189 & 0.050 \\ 
			&&SCRCMLE & -0.070 & 0.231 & 0.056 & -0.047 & 0.182 & 0.058 \\  
			&&Inf.2SCMLE & -0.073 & 0.227 & 0.058 & -0.034 & 0.181 & 0.049 \\ 
			&&2SCMLE & -0.431 & 0.154 & 0.476 & -0.319 & 0.142 & 0.351 \\ 
			&&Probit & -0.726 & 0.069 & 0.781 & -0.703 & 0.071 & 0.840 \\ 
			&&TTSLS & -0.050 & 0.226 & 0.080 & -0.028 & 0.165 & 0.086 \\  \midrule
			\multirow{6}{*}{$400$}
			&&TRCMLE & -0.073 & 0.231 & 0.056 & -0.028 & 0.170 & 0.052 \\ 
			&&SCRCMLE & -0.089 & 0.219 & 0.059 & -0.045 & 0.171 & 0.055 \\ 
			&&Inf.2SCMLE & -0.094 & 0.221 & 0.064 & -0.046 & 0.169 & 0.059 \\ 
			&&2SCMLE & -0.434 & 0.142 & 0.517 & -0.310 & 0.127 & 0.376 \\ 
			&&Probit & -0.736 & 0.050 & 0.953 & -0.724 & 0.050 & 0.975 \\ 
			&&TTSLS & -0.076 & 0.222 & 0.080 & -0.031 & 0.164 & 0.084 \\ \midrule
		\end{tabular*}
	} 
	{\scriptsize Notes: The simulation results based on 2,000 replications are reported. Each cell reports the median bias (Med.Bias), median absolute deviation (MAD), and rejection probability at 5\% significance level (RP).} 
\end{table}

\begin{figure}[!htbp]
	\begin{center} 
		\caption{\small{Estimated average structural functions (Factor model in Section~\ref{sec:sim}) \label{sim:fig:factor}}}
		\vspace{-0.5em}	\includegraphics[width=\textwidth, height=.25\textheight]{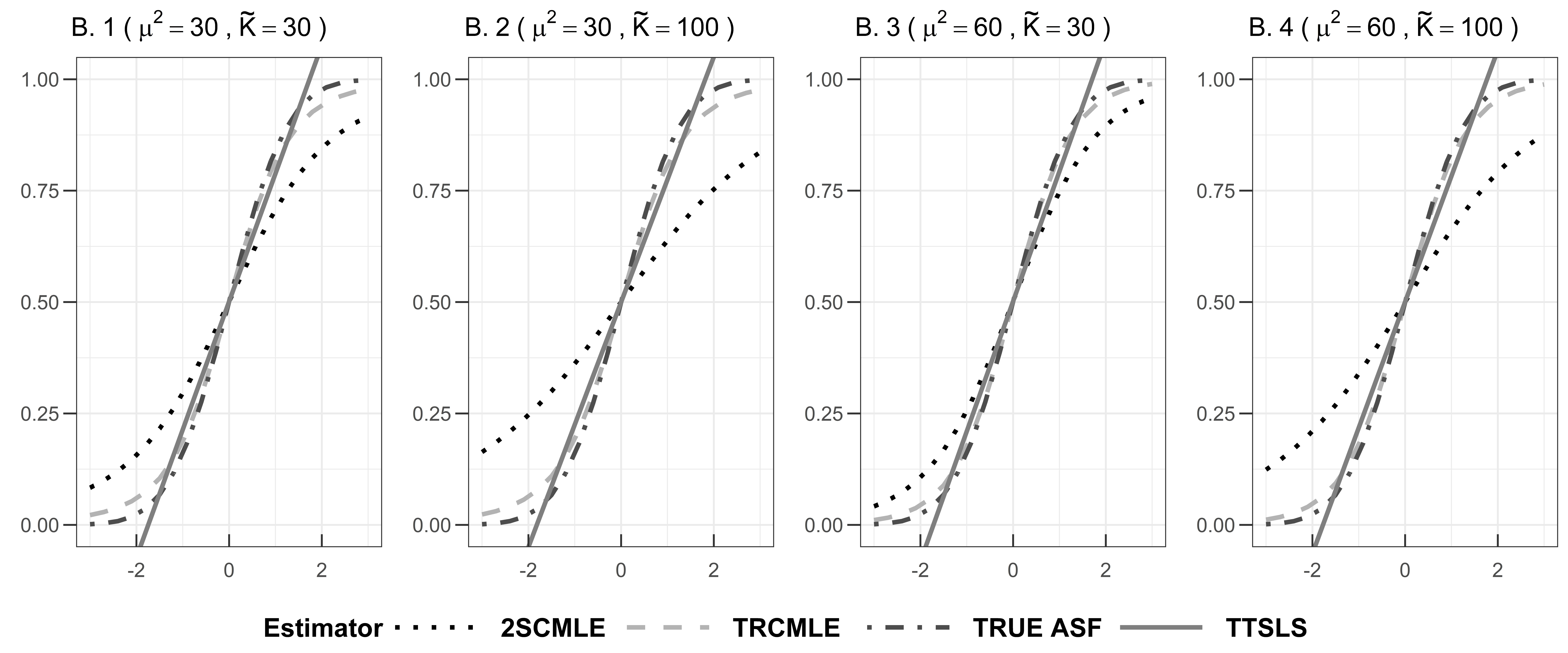}
	\end{center}
	\vspace{-1em}	\footnotesize{Notes: The estimated average structural functions based on 2,000 replications are reported. The sample size $n$ is $200$.}
\end{figure}

\begin{table}[!htbp]
	%	\rule[6pt]{1\linewidth}{1.25pt}  	
	\caption{Average partial effect estimates}
	\label{eitc:tab3}
	\begin{tabular*}{\textwidth}{@{\extracolsep{\fill}}lrrrrr} \midrule 
		Variables&TRCMLE&SCRCMLE&2SCMLE&Lasso&Probit\\\midrule
		\multicolumn{6}{l}{Model 1 \small{($\dim(Z_i)=3$)}}\\		
		\multirow{1}{*}{$I_{i,\text{(0-5)}}$} &0.3322 & 0.3277 & 0.3277 & 0.2720 & 0.0285 \\ 
		\multirow{1}{*}{$I_{i,\text{(6-12)}}$}&-0.6388 & -0.6303 & -0.6303 & -0.4324 & 0.0441 \\   
		\multirow{1}{*}{$I_{i,\text{(13-18)}}$}&	0.4835 & 0.4837 & 0.4837 & 0.3332 & 0.0116 \\ \midrule
		\multicolumn{6}{l}{Model 2 \small{($\dim(Z_i)=144$)}}\\
		\multirow{1}{*}{$I_{i,\text{(0-5)}}$} &	0.3372 & 0.3328 & 0.0901 & 0.1750 & 0.0285 \\ 
		\multirow{1}{*}{$I_{i,\text{(6-12)}}$  } &	-0.6032 & -0.6078 & -0.0816 & -0.2894 & 0.0441 \\ 
		\multirow{1}{*}{$I_{i,\text{(13-18)}}$}&	0.4580 & 0.4772 & 0.1153 & 0.2807 & 0.0116 \\ 
	\end{tabular*} 
	\rule[6pt]{1\linewidth}{1.25pt}\vspace{-1em}
	\flushleft{\footnotesize{Notes: Each cell reports the estimated average partial effect at the sample mean of family income. The sample size is 2,654.}}
\end{table} 

\begin{figure}[!htbp]
	\caption{Estimated average partial effect of family income on the probability of college completion}
	\includegraphics[width=\textwidth ]{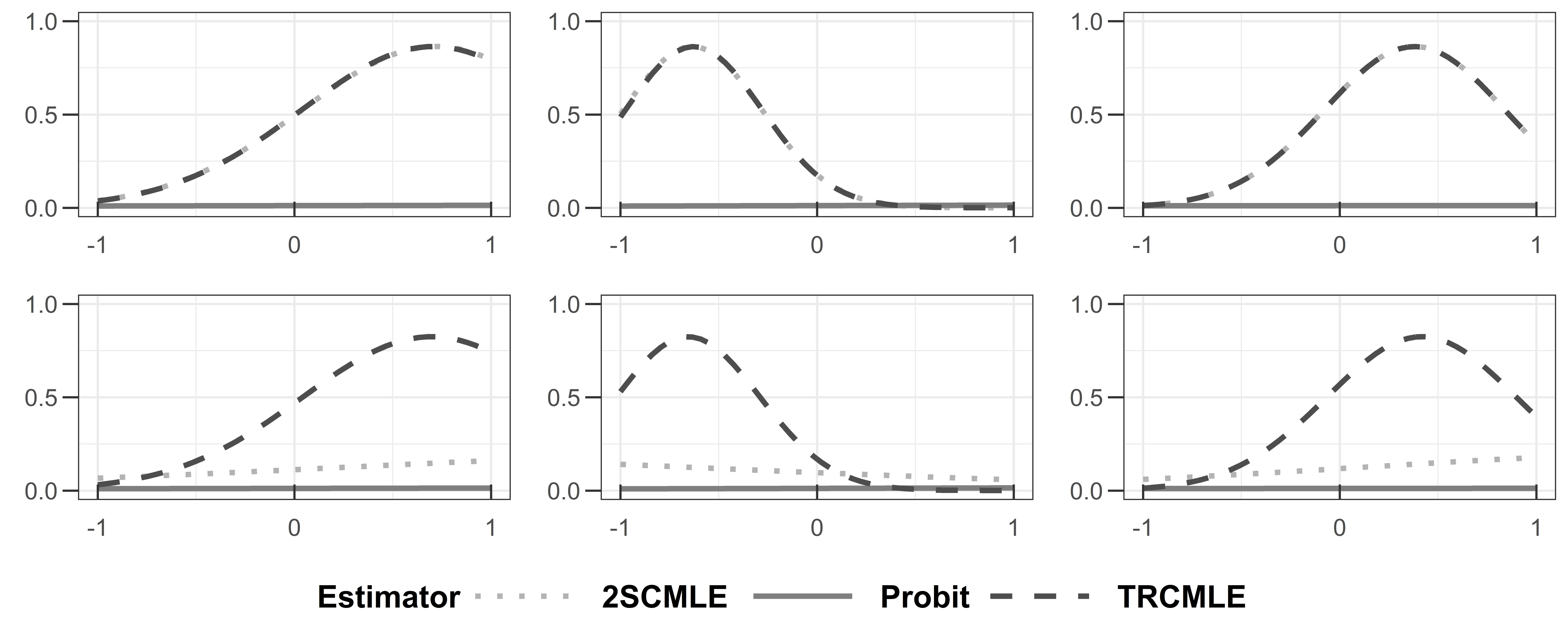}\vspace{-1em}\label{fig:eitc:ape}
	\vspace{-1em}	\flushleft{\footnotesize{Notes: Each figure reports the estimated average partial effect of family income on the probability of college completion according to the standardized family income in the age interval 0-5 (left), 6-12 (middle), and 13-18 (right).}}
\end{figure}

\section{Appendix \ref*{sec:pf}: Proofs\label{sec:pf} of the results in Sections~\ref{sec:general} and \ref{sec:asy.var}}
\subsection{Appendix~\ref*{subsec:pf}: Proofs}\label{subsec:pf}
Since Theorems~\ref{thm1:consistency:finite} and~\ref{prop4:finite} are special cases of Theorems~\ref{thm1:consistency} and~\ref{prop4}, we prove only the latter two. We first introduce a few notations. In the following,  $\mathcal M_{M} (\theta, g_i)= \partial \mathcal{Q}_M (\theta,g_i) / \partial \theta$ and $\mathcal M_{N} (\theta, g_i)= -\partial \mathcal{Q}_N (\theta,g_i) / \partial \theta$ so that the optimization problem for the RNLSE is formulated as a minimization problem. Whenever it is convenient, we will use $\sum_{i}$ to denote $\sum_{i=1} ^n$. For any $a \in \mathbb R$, $\widetilde \phi (a)$ denotes the inverse Mill's ratio given by $\widetilde \phi (a ) = \phi (a ) / \Phi  (a)$. For a linear operator $\mathcal A :\mathcal H \to \mathcal H$, $ \Vert \mathcal A  \Vert _{\op}$ and $\Vert \mathcal A \Vert _{\HS}$ respectively denote its operator norm and Hilbert-Schmidt norm, i.e.,   $ \Vert \mathcal A  \Vert _{\op} = {\sup}_{ \Vert h  \Vert_{\mathcal H} \leq 1}  \Vert \mathcal A h  \Vert_{\mathcal H}  $ and  $ \Vert \mathcal A \Vert _{\HS} =  (\sum_{j=1} ^\infty \Vert \mathcal A \zeta_j \Vert_{\mathcal H} ^2)^{1/2} =  (\sum_{j=1} ^\infty \langle\mathcal A ^{\ast }\mathcal A \zeta_{j}, \zeta_{j}\rangle_{\mathcal H})^{1/2} = (\tr(\mathcal A^{\ast} \mathcal A))^{1/2}$ where $\{  \zeta_j\}_{j \geq 1}$ is a set of orthonormal basis of $\mathcal H$. \label{r2mmm}\violet{Similarly, if $\mathcal A$ is a matrix, it reduces to $(\tr(A'A))^{1/2}$. In addition, for finite-dimensional vectors $a_1$ and $a_2 $, we sometimes use $\langle a_1, a_2\rangle $ to denote $a_1'a_2$ and $\Vert a_1\Vert = (a_1'a_1)^{1/2}$.}\label{r1mp2c} For any $h \in \mathcal H$, $\Vert \mathcal A h \Vert_{\mathcal H}  \leq \Vert \mathcal A \Vert_{\op} \Vert h\Vert_{\mathcal H}$ and $\Vert \mathcal A \Vert_{\op} \leq \Vert \mathcal A \Vert_{\HS}$.  Lastly,  let $\mathcal K_n ^{1/2} = \sum_{j=1} ^\infty \widehat \kappa_j ^{1/2} \widehat \varphi_j \otimes \widehat \varphi_j$ and  $\mathcal K^{1/2} = \sum_{j=1} ^\infty   \kappa_j ^{1/2}   \varphi_j \otimes   \varphi_j$, and let $c$ denote a generic positive constant.
\begin{proof}[Proof of Theorem \ref{thm1:consistency}]
	%	In the proof, $\mathcal K_n = \frac{1}{n} \sum_{i=1} ^n Z_i \otimes Z_i$ and $\mathcal K_{n\alpha} ^{-1} = \left(\mathcal K_n  + \alpha \mathcal I\right)^{-1}$, which are estimators of $\mathcal K$ and $\mathcal K_{\alpha} ^{-1}$ respectively. We first focus on the identification of $\theta_j$ for $j \in \{M,N\}$ on $\Theta$. By the strict monotonicity of $\Phi\left(\cdot\right)$, $\theta_{j}$ is identified iff $g_i ' \theta_{0} \neq g_i ' \theta$ for any $\theta \in \Theta\backslash\{ \theta_{0} \}$. %Let $\mathcal A_n = I_2\otimes_k S_n/\sqrt{n}$ and $\mathcal A_n^{-1} = I_2 \otimes_k \sqrt{n}S_n^{-1}$.  
	%For notational simplicity, let $\theta_{\mathcal A}$ denote $\mathcal A_n \theta/\sqrt{\mu_m} $, then 
	%\begin{equation}
	%\begin{aligned}
	%\frac{n}{\mu_m^2 }\left| g_i ' \theta - g_i '\theta_{0j}\right|^2 & = \left( \theta_{\mathcal A} - \theta_{\mathcal A %0j} \right)' \begin{bmatrix}
		%\mathbb E \left[f_i f_i ' \right]&0\\0&n S_n ^{\prime -1} \Sigma_{V} S_n ^{-1} 
		%\end{bmatrix}\left( \theta_{\mathcal A} - \theta_{\mathcal A 0j} \right) \\
		%& \geq c \left( \theta_{\mathcal A} - \theta_{\mathcal A 0j} \right)'\left( \theta_{\mathcal A} - \theta_{\mathcal A %0j} \right) \label{iden:eq}
		%\end{aligned}
		%\end{equation}
		%for a constant $c$ and $j \in \{ N,M \}$ by Assumption \ref{ass1} and \ref{ass3}. This proves the identification of %$\theta_{0j}$s.
		
		Since the compactness of $\Theta$ and the identification condition of $\theta_0$ are given in Assumption~\ref{ass4}, we focus on the proof of the uniform convergence between $\mathcal Q_{jn}(\cdot)$ and $\mathcal Q_j(\cdot)$. To this end, let $\widetilde {\mathcal Q}_{jn}(\theta) \equiv \widetilde {\mathcal Q}_{jn} (\theta,g_i ) = n^{-1} \sum_{i}  m_{j} (\theta , g_i )$ that is an auxiliary objective function associated with $\mathcal Q_j(\cdot)$. Then, because of Assumptions~\ref{ass3:finite} and \ref{ass4} and the uniform convergence results in \citet[]{andrews1987consistency}, we have  $\sup_\theta |\widetilde {\mathcal Q}_{jn}(\theta)-\mathcal Q_j(\theta) |= o_p(1)$ for $j \in \{M,N\}$. Therefore, to prove the consistency of $\widehat{\theta}_j$, it remains to show  the uniform convergence between $\widetilde {\mathcal Q}_{jn} (\cdot)$ and ${\mathcal Q}_{jn}(\cdot)$.
		
		For $j=N$, due to  the boundedness  of $\Phi (\cdot)$ and $\phi(\cdot) $, we have ${\sup}_{\theta \in \Theta }| \widetilde{\mathcal Q}_{Nn} (\theta, g_i ) - \mathcal Q_{Nn} (\theta,\widehat g_{i,\alpha} )  | \leq cn^{-1}\sum_{ i}  \Vert \mathcal S_n ^{-1} (\widehat g_{i,\alpha} - g_i ) \Vert$. Similarly, because of the Lipschitz continuity of $\widetilde \phi(\cdot)$, the  mean-value theorem, and  the  triangular inequality,  ${\sup}_{\theta \in \Theta}  |   \widetilde{\mathcal Q}_{Mn}  (\theta, g_i ) - \mathcal Q_{Mn}  (\theta,\widehat g_{i,\alpha} )   |	 \leq c (n^{-1}\sum_{ i}  \Vert \mathcal S_n ^{-1}  (\widehat g_{i,\alpha} - g_i  ) \Vert ^2 )^{1/2}  $. Hence, it suffices to show that $n^{-1}\sum_{ i}  \Vert \mathcal S_n ^{-1}  (\widehat g_{i,\alpha} - g_i  ) \Vert ^2=o_p(1)$ to prove the uniform convergence between $\widetilde {\mathcal Q}_{jn} (\theta,g_i ) $ and $ {\mathcal Q}_{jn} (\theta,g_i )$. This can be shown using   the following two inequalities.
		
		First, by using the spectral decomposition of $\mathcal K_{n\alpha} ^{-1}$, we have \begin{align}
			&\frac{1}{n}\sum_{i=1} ^n\Vert \Pi_0 (\mathcal K_n \mathcal K_{n\alpha} ^{-1} - \mathcal I) Z_i \Vert ^2 = \frac{1}{n} \sum_{i=1} ^n \Vert \sum_{j=1} ^\infty (q(\widehat \kappa_j , \alpha)-1) \langle Z_i ,\widehat \varphi_j \rangle_{\mathcal H} \Pi_0\widehat \varphi_j   \Vert ^2 \nonumber\\
			&= \sum_{j,k=1} ^\infty (q(\widehat \kappa_j,\alpha)-1)(q(\widehat \kappa_k,\alpha)-1) \langle \widehat \varphi_k, \mathcal K_n \widehat \varphi_j\rangle_{\mathcal H}  \langle \Pi_0\widehat \varphi_j , \Pi_0 \widehat \varphi_k\rangle\nonumber\\
			&=\sum_{\ell=1} ^{d_e} \sum_{j=1} ^\infty  \widehat \kappa_j (q(\widehat \kappa_j,\alpha)-1)^2 \langle \widehat \varphi_j,\pi_{\ell,0}\rangle_{\mathcal H} ^2 \leq \underset{j}{\sup} \widehat \kappa_j^{2\rho+1}(1-q(\widehat \kappa_j,\alpha))^2 \sum_{\ell=1} ^{d_e} \sum_{j=1} ^\infty  \frac{\langle \widehat \varphi_j,\pi_{\ell,0}\rangle_{\mathcal H} ^2  }{\widehat \kappa_j ^{2\rho}},	\label{p3:eq1} 
		\end{align}
		where the equalities follow from the definition of $\mathcal K_n$, $\mathcal K_{n\alpha}^{-1}$ and the orthogonality of $\{ \widehat \varphi_j \}_{j \geq 1}$ across $j$. The last term in \eqref{p3:eq1} is bounded above by $O_p( \alpha^ {\min\{(2\rho+1)/2,2\}})$ because of Assumptions \ref{ass1}.\ref{ass7} and \ref{asskinv}.\footnote{Consider the examples of $q(\kappa, \alpha)$ mentioned in the paper. \cite{Carrasco2012} shows that $ {\sup}_j \widehat \kappa_j^{(2\rho+1)/2}(1-q(\widehat \kappa_j ,\alpha))$ is bounded above by $O_p   (\alpha^{(2\rho+1)/4}  )$ for spectral cut-off and $O_p (\alpha^{\min\{(2\rho+1)/4,1\}})$ for Tikhonov regularization. Therefore, $ \sup_j \widehat\kappa_j^{2\rho+1} (q(\widehat\kappa_j, \alpha) -1 )^2 \leq (\sup_j \widehat\kappa_j^{(2\rho+1)/2} (q(\widehat\kappa_j,\alpha)-1))^2 \leq O_p (\alpha^{\min \{  (2\rho+1)/2 , 2 \}} ).$ For ridge regularization, we have ${\sup_j} {\widehat \kappa_j}^{2\rho}(1-q(\widehat \kappa_j,\alpha))^2 \leq \sup_j \alpha^2  \widehat \kappa_j^{2\rho}({\widehat \kappa_j}^2 + \alpha^2)^{-1}  \leq O_p(\alpha^{\min\{2,2\rho\}}) , \label{kinv.tmp2}$ by using arguments similar in \citet[p.397]{Carrasco2012}. Hence,  \eqref{p3:eq1} is bounded by $O_p  (\alpha^{\tilde\rho}  )$ whose convergence rate depends on the regularization scheme.}
		
		Next, let  $v_{i\ell}$ be the $\ell$th row of $V_i$. The inequality below follows from the definitions of $\Vert \cdot \Vert_{\HS}$ and $\mathcal K_n$ and the well-known properties of $\Vert \cdot \Vert_\HS$ and $\Vert \cdot \Vert_{\op}$ (see, e.g., \citealp[p.36]{Bosq2000}).  
		\begin{align} 
			\frac{1}{n } \sum_{j=1} ^n \left\Vert \frac{1}{\sqrt{n}} \sum_{i=1} ^n \left(\mathcal K_{n\alpha} ^{-1}Z_i \otimes \Lambda_n ^{-1} V_i \right) Z_j \right\Vert ^2 &  = \left\Vert \Lambda_n ^{-1}\left(\frac{1}{\sqrt n} \sum_{i=1} ^n  Z_i \otimes V_i \right) \mathcal K_{n\alpha} ^{-1} \mathcal K_n ^{1/2 } \right\Vert _{\HS} ^2\nonumber\\
			&  \leq \Vert \Lambda_n ^{-1} \Vert_{\op} ^2 \left\Vert \mathcal K_{n\alpha} ^{-1}\mathcal K_{n} ^{1/2}\right\Vert_{\op} ^2\sum_{\ell = 1}^{d_e} \left\Vert \frac{1}{\sqrt n} \sum_{i=1} ^n    Z_i   v_{i\ell}\right\Vert_{\mathcal H} ^2  \nonumber
		\end{align}
		Due to Assumptions~\ref{ass1:finite} and \ref{asskinv}, $  \Vert \Lambda_n ^{-1} \Vert_{\op} $ and $ \Vert \mathcal K_{n\alpha} ^{-1}\mathcal K_{n} ^{1/2} \Vert_{\op}$ are respectively bounded above by $\mu_m^{-1}$ and $\alpha^{-1/2}$.\footnote{We here note that \eqref{p3:eq4} is bounded above by $O_p(1/(\mu_{m,n}^2 \sqrt{\alpha}))$ under Assumption \ref{ass10} for a later use.} In addition, the central limit theorem applied to $n^{-1/2}\sum_i Z_i v_{i\ell}$ (see \citealp[Theorem 2.7]{Bosq2000}  or \citealp[Theorem 2.2]{PR1994}) tells us  $\Vert n^{-1/2}\sum_i Z_i v_{i\ell}\Vert_{\mathcal H} =O_p(1)$ for each $\ell$. Hence, we have
		%The last equality is obtained from the fact that ${\sup}_j \hspace{0.1em} \alpha  q^2(\widehat \kappa_j , \alpha)/\widehat \kappa_j \leq  {\sup}_j \hspace{0.1em} \alpha  q(\widehat \kappa_j , \alpha)/\widehat \kappa_j = O_p (1)$ due to Assumption~\ref{asskinv}.
		\begin{equation}
			\frac{1}{n } \sum_{j=1} ^n \left\Vert \frac{1}{\sqrt{n}} \sum_{i=1} ^n \left(\mathcal K_{n\alpha} ^{-1}Z_i \otimes \Lambda_n ^{-1} V_i \right) Z_j \right\Vert ^2= O_p\left(\frac{1}{\mu_{m,n} ^2 \alpha} \right).\label{p3:eq4} 
		\end{equation} 
		From \eqref{p3:eq1} and \eqref{p3:eq4}, we find that \begin{equation*}
			\frac{1}{n}\sum_{i=1} ^n  \Vert \mathcal S_n ^{-1}  (\widehat g_{i,\alpha} - g_i ) \Vert ^2  \leq \frac{4}{n} \sum_{i=1} ^n  \Vert \Pi_0 \mathcal K_n \mathcal K_{n\alpha} ^{-1} Z_i - \Pi_0 Z_i \Vert ^2 + \frac{4}{n}\sum_{j=1} ^n  \Vert \frac{1}{\sqrt{n}}\sum_{i=1} ^n \langle \mathcal K_{n\alpha} ^{-1}Z_i,Z_j\rangle_{\mathcal H} \Lambda_n ^{-1}V_i  \Vert ^2 =o_p(1)
			%\leqc \,\underset{j}{\sup} \widehat \kappa_j^{2\rho+1}(1-q(\widehat \kappa_j,\alpha))^2 \sum_{\ell=1} ^{d_e} \sum_{j=1} ^\infty  \frac{\langle \widehat \varphi_j,\pi_{\ell,0}\rangle  ^2  }{\widehat \kappa_j ^{2\rho}}   + \frac{c}{\mu_{m,n}^2 \alpha } \underset{j}{\sup} \frac{\alpha  q(\widehat \kappa_j + \alpha)}{\widehat \kappa_j}  \leq 
			%	=O_p (\alpha^{\frac{(2\rho+1)}{2}}  )+ O_p ( 1/(\mu_{m,n}^2 \alpha) ) ,
		\end{equation*}
		from which it can be deduced that $\underset{\theta \in \Theta}{\sup}  | \mathcal Q_{jn}  (\theta , \widehat g_{i,\alpha} ) - \mathcal Q_{j}  (\theta, g_i )  | = o_p(1)$. 	
	\end{proof}
	
	\begin{proof}[Proof of Theorem \ref{prop4}] Theorem \ref{prop4} is obtained from Lemmas \ref{lem4}-\ref{lem7} as in the proof of Theorem 2 in \cite{Chen2003} applied to $\mathcal S_n ^{-1}\partial \mathcal Q  \left(\theta ,\mathcal G_i \left(\Pi\right)\right)/\partial \theta$.  
	\end{proof}
	\begin{proof}[Proof of Theorem \ref{prop6}]
		Theorem \ref{prop6} is a consequence of Lemmas \ref{lem8}-\ref{lem10} and Slutsky's theorem. 
	\end{proof} 
	
	\subsection{Appendix \ref*{sec:lem}: Lemmas \label{sec:lem}}	 
	We use $\mathcal H^r$ to denote the Cartesian product of $r$ copies of $\mathcal H$ equipped with the inner product $\langle h_1,h_2\rangle_{\mathcal H^{r}}  = \sum_{j=1} ^r \langle h_{1j},h_{2j} \rangle_{\mathcal H} $ and its induced norm $\Vert h_1\Vert_{\mathcal H^{r}} = (  \sum_{j=1} ^r \langle h_{1j} , h_{1j}\rangle_{\mathcal H} ^2 )^{1/2} $ for $ h_1  = (h_{11},\ldots ,h_{1r})' $ and $h_2= (h_{21},\ldots ,h_{2r})' $. In addition, let $\mathcal T_i $ denote an operator from $\mathcal H ^{ d_e}$ to $\mathbb R^{d_e}$ satisfying \[ \mathcal T_i h =  \left( \langle h_{1}, Z_i\rangle_{\mathcal H} ,\ldots,  \langle h_{d_e}, Z_i\rangle_{\mathcal H}  \right)', \] for  $ h  = (h_{1},\ldots ,h_{d_e})' \in \mathcal H^{d_e} $ Then, for $\Pi:\mathcal H \to \mathbb R^{d_e}$ and $\pi = (\Pi^\ast e_1 , \ldots , \Pi ^\ast e_{d_e})' \in \mathcal H^{d_e}$, we have $\Pi Z_i = \mathcal T_i \pi$. We let $\mathcal G_i (\pi)= ((\mathcal T_i \pi) ' , (Y_{2i} -\mathcal T_i \pi) ' )' =    ((\Pi Z_i)' , (Y_{2i} - \Pi Z_i)') 
	' $ so that $g_i$ and $\widehat g_{i,\alpha}$ can be understood as the function $\mathcal G_i (\cdot)$ evaluated at $\pi_n =  (\Pi_{n} ^\ast e_1 , \ldots , \Pi_{n} ^\ast e_{d_\ell} )'  $ and $\hat \pi_\alpha =  (\widehat \Pi_{\alpha} ^\ast e_1 , \ldots , \widehat \Pi_{\alpha} ^\ast e_{d_\ell} )'  $. In the following, $m_{1j}(\theta, g_i)$ is $\phi(g_i'\theta)(y_i - \Phi(g_i'\theta))/(\Phi(g_i'\theta)(1-\Phi(g_i'\theta))) $ for the RCMLE and $-\phi(g_i'\theta)(y_i - \Phi(g_i'\theta))$ for the RNLSE. In addition, $ \dot m_{2j} (\theta , g_i )  $ are defined in the main text. $\ddot{m}_{2N} (\theta, g_i) = -   (y_i - \Phi (g_i ' \theta ) )\phi '  (g_i ' \theta )$     and $\ddot m_{2M} (\theta , g_i ) = \phi '  (g_i ' \theta  ) (y_i - \Phi (g_i'\theta ))/(\Phi (g_i ' \theta ) (1-\Phi (g_i ' \theta ) )) $

	For both $j$, $\widehat \theta_j$ satisfies that $\mathcal M_j (\widehat\theta_j, \widehat{g}_{i,\alpha})  = 0$   (\citealt[Condition 2.1]{Chen2003}). Lemmas \ref{lem6} and \ref{lem6prime} verify the results implied by Conditions (2.2)-(2.4) in \cite{Chen2003} under our setup.   We focus on the proof for $j=N$; under Assumption \ref{ass9:finite}, the proofs for the RCMLE can be obtained in a similar manner. Thus, we omit the details.

	\begin{lemma}\label{lem6}\normalfont
		Suppose that the conditions in Theorem \ref{prop4} hold, and let $\Theta_{\delta} = \{   \theta \in \Theta :  \Vert \theta - \theta_{0}   \Vert \leq \delta     \}$. Then the following holds for $j \in \{M,N\}$. \begin{enumerate}[\text{(\ref*{lem6}.}1\text{)}]
			\item \label{lem6-1}The  derivative $\Gamma _{1,j}   (\theta , g_{i} )$ in $\theta$ of $\mathcal M_{j}  (\theta , g_{i}  )$ exists for $\theta \in \Theta_{\delta}$, and is continuous at $\theta = \theta _{0}$. Let $\Gamma _{1,0j}   \equiv \Gamma _{1,j}  (\theta_{0}  , g_{i}  )$ and then $\mathcal S_n ^{-1}\Gamma _{1,0j}  \mathcal S_n ^{\prime -1} $ is nonsingular for $n$ large enough. 
			\item  \label{lem6-3}  For $\theta \in \Theta_{\delta}$, the Fr\'echet derivative $ \Gamma_{2,j}  (\theta , \mathcal G_i  (\pi_{n}  ) ): \mathcal H ^{d_e} \to  \mathbb R^{2d_e}  $ of $\mathcal M_j (\theta, \mathcal G_i (\pi))$ in $\pi =  (\pi_{1} , \ldots , \pi_{d_e} )' \in \mathcal H^{ d_{e}}   $ at $\pi_{n} $ exists. $\Gamma_{2,0j} \equiv \Gamma_{2,j} (\theta_0 , \mathcal G_i (\pi_n))$. For all $   \theta  \in \Theta_{\delta_n}  $ with a positive sequence $ \delta_n= o(1)$, \begin{enumerate}[\text{(\ref*{lem6}.2.}1\text{)}] \item\label{lem4.2.1} $ \Vert\mathcal S_n ^{-1} ( \Gamma _{2,j}  (\theta ,g_{i} )- \Gamma_{2,j}  (\theta_{0}, g_{i} )  )  ( \widehat \pi_\alpha - \pi_{n}  )   \Vert  \leq o_p(1)\Vert \mathcal S_n ' (\theta - \theta_0)\Vert    $,\item\label{lem4.2.2} $\Vert \mathcal S_n^{-1} ( \mathcal M_j (\theta , \mathcal G_i (\hat\pi_\alpha)) - \mathcal M_j(\theta , \mathcal G_i (\pi_n) ) - \Gamma_{2,j} (\theta, \mathcal G_i (\pi_n)) (\hat\pi_\alpha - \pi_n) ) \Vert \leq  O( \mu_{m,n}^{-1}\sqrt{n} ) \Vert ( \hat\Pi_\alpha -\Pi_n )\mathcal K^{1/2} \Vert_{\HS}^2  $.
			\end{enumerate}
		\end{enumerate}
	\end{lemma}
	\begin{proof}[Proof of Lemma \ref{lem6}]
		$\Gamma _{1,j}  (\theta , g_{i}  ) = \mathbb E  [ \dot m_{2j}  (\theta , g_{i } ) g_{i} g_{i} '   ] +  \mathbb E  [\ddot{m}_{2j}  (\theta , g_{i } )g_{i} g_{i}' ]$.    By the law of iterated expectations,  $\mathbb E  [\ddot{m}_{2j}  (\theta_0 , g_{i } )g_{i} g_{i}' ]$  is equal to 0. Then, \ref{lem6-1} follows from the continuity of $\dot m_{2j}(\theta, g_i)$ and $\ddot m_{2j}(\theta, g_i)$ with respect to $\theta$, Assumptions~\ref{ass3:finite} and~\ref{ass9:finite}, and the boundedness of $\phi (\cdot )$ and $\phi'(\cdot)$. 
		
		To prove \ref{lem4.2.1}, note that the Fr\'echet derivative of $\mathcal T_i h$ in $h$ is $ \mathcal T_i $ for any $h \in \mathcal H^{ d_e} $. Hence, by the chain rule, the Fr\'echet derivative of $ \mathcal M_{j}  (\theta , \mathcal G_i (\pi) )$ at $\pi_{n}$ exists, and we have
		\begin{equation}
			\begin{aligned}
				\mathcal S_n ^{-1} \Gamma_{2,N}  ( \theta, g_{i}   )  (\pi - \pi_{n} ) = &\mathbb E  [\dot{m}_{2N} (\theta, g_i )  g_{0i}\psi ' \mathcal T_i   ](\pi - \pi_{n} )+\mathbb E  [\ddot{m}_{2N} (\theta, g_i )  g_{0i}  \psi'\mathcal T_i ] (\pi - \pi_{n} ) \\
				& - \mathbb E  [m_{1N}  (\theta , g_i ) \widetilde{\mathcal S}_n ^{-1}   \mathcal T_i]  (\pi-\pi_n),
			\end{aligned}\label{eq:lem3-1}
		\end{equation} where $\widetilde{\mathcal S}_n ^{-1} = \mathcal S_n ^{-1} (\mathcal I_{d_e} , - \mathcal I_{d_e})'$. We only prove that the first term in the RHS of \eqref{eq:lem3-1} satisfies  \ref{lem4.2.1}; the proofs for the remaining terms are similar. We have  % satisfies \ref{lem4.2.1};  Specifically,   
		\begin{align}
			& \Vert \mathbb E  [\dot{m}_{2N} (\theta, g_i )  g_{0i}\psi ' \mathcal T_i   ](\widehat \pi_\alpha - \pi_{n} ) - \mathbb E  [\dot{m}_{2N} (\theta_0, g_i )  g_{0i}\psi_0 ' \mathcal T_i   ](\widehat\pi_\alpha - \pi_{n} )\Vert \nonumber\\
			%	& =\Vert \mathbb E [\phi^2 (g_{i} ' \theta )\mathcal S_n ^{-1} g_{i} \psi ' \mathcal T_i] (\pi - \pi_{n} ) - \mathbb E [\phi^2 (g_{i} ' \theta_0 ) \mathcal S_n ^{-1}g_{i} \psi_{0} ' \mathcal T_i ] (\pi - \pi_{n} ) \Vert  \nonumber\\
			&\leq  \Vert \mathbb E [( \phi ^2 (g_{i} ' \theta ) - \phi ^2 (g_{i} ' \theta_{0}  )   ) g_{i0} \psi ' \mathcal T_i ] (\widehat\pi_\alpha - \pi_{n} )  \Vert  +  \Vert \mathbb E [\phi ^2 (g_{i} ' \theta_{0}  ) g_{i0} (\psi - \psi_{0} ) ' \mathcal T_i ](\widehat\pi_\alpha- \pi_{n} )  \Vert \nonumber\\ 
			&\leq   O(1)  \mathbb E [  \Vert g_{i0} \Vert ^2 \Vert Z_i \Vert_{\mathcal H}      ] \Vert \mathcal S_n ' (\theta - \theta_0)\Vert  \Vert\widehat\pi_\alpha - \pi_n\Vert_{\mathcal H^{d_e}}  + \phi^2 (0) \mathbb E [  \Vert g_{i0} \Vert\Vert Z_i\Vert_{\mathcal H}  ]  \Vert \psi - \psi_0\Vert \Vert \widehat\pi_\alpha - \pi_n\Vert_{\mathcal H^{d_e}}  \nonumber\\ 
			&= O(1) \Vert \mathcal S_n ' (\theta - \theta_0)\Vert \Vert \widehat{\pi}_\alpha - \pi_n\Vert_{\mathcal H^{d_e}} = o_p(1) \times \Vert \mathcal S_n ' (\theta - \theta_0)\Vert  , \label{eqlem6tmp}
		\end{align}
		where the  inequalities  follow from the triangular inequality and the Lipschitz continuity   and the boundednesses of $\phi(\cdot)$. The last line is deduced from Assumption \ref{ass9:finite} and because $\Vert \widehat\pi_\alpha - \pi_n\Vert_{\mathcal H^{d_e}}  \leq \Vert \mathcal S_n ^{-1}( \widehat\Pi_\alpha - \Pi_n) \Vert_{\HS} \leq   \Vert \Pi_0(\mathcal K_{n\alpha} ^{-1}\mathcal K_n - \mathcal I)\Vert_{\HS} + O_p((\alpha ^{1/2}\mu_{m,n} ) ^{-1})\Vert n^{-1/2}\sum_i Z_i \otimes V_i\Vert_{\HS} = o_p(1) $.
		
		We then focus on \ref{lem4.2.2}. Given that $\mathcal M_j (\theta, \mathcal G_i (\pi))$ is Fr\'echet differentiable in $\pi$, it is G\^ateaux differentiable and, by the continuity of G\^ateaux derivative and the mean-value theorem, we have   \begin{align} 
			\Vert    \mathcal S_n^{-1}& ( \mathcal M_j (\theta , \mathcal G_i ( \pi)) - \mathcal M_j(\theta , \mathcal G_i (\pi_n) ) - \Gamma_{2,j} (\theta, \mathcal G_i (\pi_n)) ) (  \pi - \pi_n)    \Vert  \nonumber\\&\leq  \sup_{0\leq t \leq 1} \Vert   \mathcal S_n ^{-1} (\Gamma_{2,j} (\theta, \mathcal G_i ( \pi_n +t(  \pi - \pi_n))) -\Gamma_{2,j} (\theta, \mathcal G_i (\pi_n)))(  \pi - \pi_n) \Vert ,\label{eq:tmp111}
		\end{align} 
		see \citet[Theorem 3.2.7]{drabek2007methods}. Using the representation of $\Gamma_{2,j}(\cdot,\cdot)$ in \eqref{eq:lem3-1},   \eqref{eq:tmp111} can be decomposed into terms each of which is associated with each term in the RHS of \eqref{eq:lem3-1}. Then, by using similar arguments in proving \eqref{eqlem6tmp} and the fact that $\theta'(\mathcal G_i (\pi_n + t (  \pi - \pi_n) )- \mathcal G_i (\pi_n))  = t  \theta ' (\mathcal I_{d_e} , -\mathcal I_{d_e}) '  \mathcal T_i (\pi - \pi_n)  $, it can be shown that \begin{equation*}
			\eqref{eq:tmp111} \leq  O(1)(\mathbb E[\Vert \mathcal T_i (\pi - \pi_n)\Vert ^2 \Vert g_{0i}\Vert  ] + \Vert \mathcal S_n ^{-1}\Vert_{\op} \mathbb E[  \Vert \mathcal T_i (\pi - \pi_n)\Vert ^2 ]    ).
		\end{equation*}
		Then, the desired result is obtained because $\mathbb E[ \Vert \mathcal T_i (\pi - \pi_n) \Vert ^2 ] = \sum_{\ell=1} ^{2d_e} \mathbb E[  \langle Z_i, \pi_\ell - \pi_{n\ell} \rangle_{\mathcal H} ^2  ] = \sum_{\ell=1} ^{2d_e}   \langle \mathcal K (\Pi - \Pi_n) ^\ast e_\ell ,  (\Pi - \Pi_n) ^\ast e_\ell \rangle_{\mathcal H}  = \Vert \mathcal K^{1/2} (\Pi - \Pi_n )^\ast \Vert_{\HS}^2   $, $\Vert \mathcal S_n^{-1}\Vert_{\op} = O(\mu_{m,n}^{-1}\sqrt{n})$, and $\mathbb E[\Vert g_{0i}\Vert |x ] \leq c$ a.s.n.
			\end{proof}
			Next, we show that the term appearing in \ref{lem4.2.2} is bounded above by $o_p(n^{-1/2})$. This and Lemma \ref{lem4.2.2} imply the result obtained from Conditions (2.3).(i) and (2.4) in \cite{Chen2003}.
			\begin{lemma}\label{lem6prime}\normalfont
				Suppose that the conditions in Theorem \ref{prop4} hold. Then, $ 	\Vert \mathcal K^{1/2} (\mathcal K_{n\alpha} ^{-1} \mathcal K_n - \mathcal I) \pi_{\ell,0} \Vert_{\mathcal H}  ^2 = o_p (n^{-1})$   for all $\ell=1,\ldots,d_e$ and  $  \mu_{m,n}^{-1}\sqrt{n}  \Vert ( \widehat \Pi_\alpha - \Pi_n )\mathcal K^{1/2} \Vert_{\HS}^2  = o_p(n^{-1/2} ) $.
			\end{lemma}
			\begin{proof}For each $\ell$, we have that
				\begin{equation}
					\Vert \mathcal K^{1/2} (\mathcal K_{n\alpha} ^{-1} \mathcal K_n -\mathcal I) \pi_{\ell,0} \Vert_{\mathcal H}  ^2 \leq 2 \Vert \mathcal K^{1/2}  (\mathcal K_{\alpha} ^{-1} \mathcal K - \mathcal I) \pi_{\ell,0} \Vert_{\mathcal H} ^2  + 2 \Vert \mathcal K^{1/2} (\mathcal K_{n\alpha} ^{-1} \mathcal K_n - \mathcal K_{\alpha} ^{-1} \mathcal K)\pi_{\ell,0}\Vert_{\mathcal H}  ^2 . \label{lem.eq1}
				\end{equation}
				Using similar arguments in obtaining \eqref{p3:eq1}, we bound the first summand as follows.
				\begin{equation}
					\Vert \mathcal K^{1/2} (\mathcal K_{\alpha} ^{-1} \mathcal K - \mathcal I) \pi_{\ell,0} \Vert_{\mathcal H}  ^2 \leq  (\sup_j \kappa_j ^{\rho+1/2 } (q(\kappa_j,\alpha)-1))^2 \sum_{j=1} ^\infty \kappa_j ^{-2\rho}\langle \pi_{\ell,0} , \varphi_j \rangle_{\mathcal H} ^2  = O ( \alpha ^2 ), \label{lem.eq2}
				\end{equation}
				by Assumptions \ref{ass1}.\ref{ass7}, \ref{asskinv} and \ref{ass10}. Because $\alpha \sqrt{n } =o(1)$, \eqref{lem.eq2} is $ o_p(n^{-1})$. We next focus on the second summand in \eqref{lem.eq1}. To this end,  let $\mathcal K_\alpha$ denote the generalized inverse of $\mathcal K_\alpha^{-1}$ such that $
				\mathcal K_\alpha   = \sum_{j :  q(\kappa_j,\alpha) > 0} \frac{\kappa_j}{q(\kappa_j,\alpha)} \varphi_j \otimes \varphi_j 
				$, and $\mathcal K_{n\alpha}$ is similarly defined with $\{ \widehat{\kappa}_j , \widehat{\varphi}_j \}_{j \geq 1}$. Then, the following holds. \begin{equation}
					\Vert\mathcal K ^{-\rho}\pi_{\ell,0} \Vert_{\mathcal H}  =O(1) ,\quad\Vert \mathcal K_{n\alpha} ^{-1/2} \mathcal K \mathcal K_{n\alpha} ^{-1/2}\Vert_{\op}= O_p(1)\quad\text{and}\quad \Vert \mathcal K^{1/2} \mathcal K_{n\alpha} ^{-1}\mathcal K^{1/2}\Vert_{\op}  = O_p(1) ,\label{eq: help: tmp}
				\end{equation}because $\Vert \mathcal K_{n\alpha} ^{-1/2} \mathcal K \mathcal K_{n\alpha} ^{-1/2}\Vert_{\op} \leq \Vert \mathcal K_{n\alpha} ^{-1/2}\Vert_{\op}\Vert  \mathcal K - \mathcal K_n \Vert_{\op} \Vert  \mathcal K_{n\alpha} ^{-1/2}\Vert_{\op} + \Vert \mathcal K_{n\alpha} ^{-1}\mathcal K_n\Vert_{\op}\leq O_p( (n\alpha)^{-1/2}) + \sup_j q(\kappa_j ,\alpha) \leq o_p(1) + 1$, and the last is obtained from the second bound from that $\Vert \mathcal A^\ast \mathcal A\Vert_{\op} = \Vert \mathcal A\Vert_{\op}^2 = \Vert \mathcal A^{\ast}\Vert_{\op}^2 = \Vert\mathcal A \mathcal A^\ast \Vert_{\op} $ for any linear bounded operator $\mathcal A$. The term inside the norm of the second summand in \eqref{lem.eq1} satisfies that\begin{align}
					&\mathcal K^{1/2} (\mathcal K_{n\alpha} ^{-1}\mathcal K_n - \mathcal K_{\alpha}^{-1} \mathcal K) \pi_{\ell,0} = \mathcal K^{1/2} (\mathcal K_{n\alpha} ^{-1} - \mathcal K_{\alpha} ^{-1}) \mathcal K \pi_{\ell,0} +\mathcal K^{1/2} \mathcal K_{n\alpha} ^{-1} (\mathcal K_n - \mathcal K) \pi_\ell  \nonumber\\
					&=\mathcal K^{1/2} \mathcal K_{n\alpha} ^{-1} (\mathcal K_{\alpha} -\mathcal K+\mathcal K_n- \mathcal K_{n\alpha} +\mathcal K -\mathcal K_n )\mathcal K_{\alpha} ^{-1} \mathcal K \pi_{\ell,0}+ \mathcal K^{1/2} \mathcal K_{n\alpha} ^{-1} (\mathcal K_n - \mathcal K) \pi_{\ell,0}   \nonumber\\
					& = \mathcal K ^{1/2} \mathcal K_{n\alpha} ^{-1} (\mathcal K_{\alpha} - \mathcal K) \mathcal K_{\alpha} ^{-1} \mathcal K  \pi_{\ell,0}+\mathcal K^{1/2} \mathcal K_{n\alpha} ^{-1} (\mathcal K_n - \mathcal K_{n\alpha}) \mathcal K_{\alpha} ^{-1} \mathcal K  \pi_{\ell,0}+\mathcal K^{1/2} \mathcal K_{n\alpha} ^{-1} (\mathcal K_n - \mathcal K)(\mathcal I - \mathcal K_{\alpha} ^{-1} \mathcal K) \pi_{\ell,0} . \nonumber
				\end{align}
				Using similar arguments in \eqref{lem.eq2}, the bounds in \eqref{eq: help: tmp} and the fact $\mathcal K\mathcal K_\alpha^{-1}=\mathcal K_\alpha^{-1}\mathcal K$, we find
				\begin{equation*} 
					\Vert \mathcal K ^{1/2} \mathcal K_{n\alpha} ^{-1} (\mathcal K_{\alpha} - \mathcal K) \mathcal K_{\alpha} ^{-1} \mathcal K  \pi_{\ell,0}\Vert_{\mathcal H}  ^2 \leq  \Vert \mathcal K^{1/2} \mathcal K_{n\alpha} ^{-1}  \mathcal K^{1/2}  (\mathcal I - \mathcal K \mathcal K_{\alpha}^{-1}) \mathcal K^{\rho+1/2} \Vert_{\op} ^2  \Vert \mathcal K^{-\rho}  \pi_{\ell,0}\Vert_{\mathcal H}  ^2  = O_p(\alpha^ 2). \label{lem.eq3}  
				\end{equation*}
				Similarly, we have $ \Vert 	\mathcal K^{1/2} \mathcal K_{n\alpha} ^{-1} (\mathcal K_n - \mathcal K_{n\alpha}) \mathcal K_{\alpha} ^{-1} \mathcal K  \pi_{\ell,0} \Vert_{\mathcal H} ^2  = O_p(\alpha^2) $. Lastly, because of Assumption~\ref{ass1}.\ref{ass7} and the facts that $\Vert \mathcal K_n - \mathcal K\Vert_{\HS} = O_p(n^{-1/2})$ (see, e.g., \citealt[Theorem 4.1]{Bosq2000}) and $\Vert \mathcal K^{1/2}\mathcal K_{n\alpha} ^{-1}\Vert_{\op} ^2 =\Vert \mathcal K_{n\alpha} ^{-1} \mathcal K \mathcal K_{n\alpha} ^{-1}\Vert_{\op}= O_p(1) \Vert \mathcal K_{n\alpha} ^{-1} \Vert_{\op} = O_p(\alpha^{-1/2})$, we have that 
				\begin{equation*} 
					\Vert \mathcal K^{1/2} \mathcal K_{n\alpha} ^{-1} (\mathcal K_n - \mathcal K)(\mathcal I - \mathcal K_{\alpha} ^{-1} \mathcal K) \pi_{\ell,0}\Vert_{\mathcal H}  ^2  \leq \Vert \mathcal K^{1/2}\mathcal K_{n\alpha} ^{-1}\Vert_{\op}^2 \Vert   \mathcal K_n - \mathcal K  \Vert_{\op} ^2 \Vert (\mathcal I - \mathcal K \mathcal K_{\alpha} ^{-1}) \mathcal K^{\rho}\Vert_{\op} ^2  =o_p(n^{-1}).  \label{lem.eq4}
				\end{equation*} 
				%Therefore,	\eqref{lem.eq3}, \eqref{lem.eq4} and \eqref{lem.eq5} are bounded above by $o_p(n^{-1})$ by Assumption~\ref{ass10}. This implies that 
				Because $\alpha^2 n =o(1)$ and the inequality $2ab \leq  a^2 + b^2$ for $a,b \in \mathbb R$, we conclude that $\Vert \mathcal K^{1/2} (\mathcal K_{n\alpha} ^{-1} \mathcal K_n - \mathcal K_{\alpha} ^{-1} \mathcal K ) \pi_{\ell,0}\Vert_{\mathcal H} ^2 = o_p(n^{-1})$. This concludes the proof of the first part of the lemma.
				
				We then consider the second part. By the definition of $\widehat{\Pi}_n$, we find that \begin{equation}  
					\Vert (\widehat{\Pi}_\alpha - \Pi_n )\mathcal K^{1/2}\Vert_{\HS}  ^2   \leq 2 \Vert \Pi_n(\mathcal K_{n\alpha} ^{-1} \mathcal K_n - \mathcal I) \mathcal K^{1/2}  \Vert_{\HS}  ^2 +2   \Vert  \mathcal K^{1/2}  \mathcal K_{n\alpha} ^{-1}\frac{1}{n} \sum_{i=1} ^n V_i \otimes Z_i  \Vert _{\HS} ^2 . \label{teq1}
				\end{equation}
				Because of the symmetricity of $\mathcal K$, $\mathcal K_{n\alpha}^{-1}$ and $\mathcal K_n$, the first term satisfies that $\Vert \Pi_n(\mathcal K_{n\alpha} ^{-1} \mathcal K_n - \mathcal I) \mathcal K^{1/2}  \Vert_{\HS}  ^2 \leq  \Vert  \mathcal K^{1/2}(\mathcal K_{n\alpha} ^{-1} \mathcal K_n - \mathcal I)  \Pi_0 ^\ast  \Vert_{\HS}  ^2 = \sum_{\ell=1} ^{d_e} 	\Vert \mathcal K^{1/2} (\mathcal K_{n\alpha} ^{-1} \mathcal K_n -\mathcal I) \pi_{\ell,0} \Vert_{\mathcal H}  ^2$. Therefore, by the first part of the lemma, it is bounded above by $o_p(n^{-1})$. Furthermore, because of the central limit theorem in  \citet[Theorem 4.1]{Bosq2000} and \citet[Theorem 3.1]{hormann2010},  $\Vert n^{-1}\sum_{i=1} ^n V_i \otimes Z_i\Vert_{\HS} ^2 = O_p(n^{-1})$. Combining this with the bound $\Vert \mathcal K^{1/2}\mathcal K_{n\alpha} ^{-1} \Vert_{\op} ^2 = O_p(\alpha^{-1/2})$, the second term in \eqref{teq1} is bounded above by $O_p(n^{-1}\alpha ^{-1/2})$. Thus, the desired result is obtained due to $\mu_{m,n} ^2 \alpha \to \infty$.
			\end{proof}

			Lemmas \ref{lem4} and \ref{lem5} verify the result implied by Condition (2.5) in \cite{Chen2003} via Jain-Marcus Theorem (\citealt[p.213]{Vaart1996}). 
			\begin{lemma}\label{lem4}\normalfont
				Suppose that Assumption \ref{ass8} holds. Then, $
				\int_{0} ^\infty \sqrt{\log{\mathfrak N (  \epsilon, \mathcal H^{d_e },  \Vert \cdot \Vert_{\mathcal H}  )}} \partial \epsilon <\infty  $.
			\end{lemma}
			\begin{proof}[Proof of Lemma \ref{lem4}]
				Let $\mathfrak N = \mathfrak N (\epsilon_1 , \mathcal H , \Vert \cdot \Vert_{\mathcal H}  )$ and $\{\mathcal B_k (\epsilon_1)\}_{k=1} ^{\mathfrak N}$ be the collection of $\epsilon_1$-balls centered at $\{h_k \}_{k=1} ^{\mathfrak N}$ satisfying $\mathcal H \subset\bigcup_{k=1} ^\mathfrak{N} \mathcal B_k (\epsilon_1)  $. We also let $\mathfrak D$ be the set of $d_e$-tuples over $\{h_k \}_{k=1} ^{\mathfrak N}$. Then, for $\epsilon = d_e^{1/2} \epsilon_1$, we can define the collection of $\epsilon$-balls, denoted $\{\mathcal D_k (\epsilon)\}_{k=1} ^{\mathfrak{N}^{d_e}}$, each of which is centered at an element of $\mathfrak D$. Note that, for any $\pi \in \mathcal H ^{d_e}$ and $\ell=1,\ldots, d_e$, there exists $h_\ell ^\ast \in  \{h_k \}_{k=1} ^{\mathfrak N}$ such that $\Vert e_\ell ' \pi - h_\ell ^\ast \Vert_{\mathcal H}  \leq \epsilon_1$ by the definition of $\{h_k \}_{k=1} ^{\mathfrak N}$. Therefore, for any $\pi \in \mathcal H ^{d_e}$, there always exists some $h^\ast = (h_1 ^\ast ,\ldots , h_{d_e} ^\ast)' $ such that $\Vert \pi - h^\ast \Vert_{\mathcal H^{d_e}} \leq (d_e \max_{1\leq \ell \leq d_e} \Vert e_\ell ' \pi - h_\ell ^\ast\Vert_{\mathcal H}^2)^{1/2}  \leq \epsilon$. Hence, $\mathcal H ^{ d_e} \subset  \bigcup_{k=1} ^{\mathfrak{N}^{d_{e}}} \mathcal D_k (\epsilon)$, and thus, $	\int_{0} ^\infty \sqrt{\log{\mathfrak N (\epsilon, \mathcal H ^{d_{e}} ,  \Vert \cdot \Vert_{\mathcal H}  )}}   \partial \epsilon \leq   \int_{0} ^\infty \sqrt{d_e \log{\mathfrak N  }}   \partial \epsilon_1 <\infty  $ by Assumption \ref{ass8}. 
			\end{proof}
			
			\begin{lemma}\label{lem5}\normalfont
				Suppose that the conditions in Theorem~\ref{prop4} hold. Let $\widetilde m_{1j,\ell}  (\theta , \mathcal G_i (\pi) ) = e_\ell ' \partial m_j  (\theta ,  \mathcal G_i (\pi))/\partial \theta  $. Then, $ \widetilde m_{1j,\ell}  (\theta ,  \mathcal G_i (\pi) )$ is H\"older continuous with respect to $\theta$ and $\pi$ for $j \in \{M , N\}$ and $\ell = 1,\ldots,2d_e$.
			\end{lemma}\begin{proof}[Proof of Lemma \ref{lem5}] We will focus on the proof for the first $d_e$ $\ell$'s, since those for the remainder can be obtained in a similar manner. 	For a linear operator $\Pi_1:\mathcal H \to \mathbb R^{d_e}$ (resp.\ $\Pi_2:\mathcal H \to \mathbb R^{d_e}$),  let $\pi_{1\ell} = \Pi_1 ^\ast e_\ell$ (resp.\ $\pi_{2\ell} = \Pi_2 ^\ast e_\ell$) and, for each $s=1,2$,  let $\Phi_s = \Phi    (\mathcal G_i ' (\pi_s)  \theta_s  )$ and $\phi_s = \phi    (\mathcal G_i ' (\pi_s)    \theta_s  )$.   %By using the notations, $\widetilde m_{1N,\ell} (\theta_1 , \mathcal G_i (\pi_s))$ can be written as $ \phi_s (y_i - \Phi_s) \langle \pi_{s\ell} ,Z_i \rangle $. 
				By the triangular inequality, we have\begin{align} 
					&|\widetilde m_{1N,\ell}  (\theta_1 ,\mathcal G_i  (\pi_1 ) )  -\widetilde m_{1N,\ell} (\theta_2 ,\mathcal G_i  (\pi_2 ) ) |\nonumber\\&	\leq  |      (\Phi_1  - \Phi_2 )  \phi_1    | |\langle \pi_{1\ell}, Z_i\rangle_{\mathcal H}  |       + |   (y_i - \Phi _2 ) (\phi _1  - \phi_2   )   | |  \langle \pi_{1\ell}, Z_i\rangle_{\mathcal H}  |  +  |  (y_i - \Phi_2 )\phi_2   | |\langle \pi_{1\ell}  -   \pi_{2\ell}, Z_i\rangle _{\mathcal H}  | .\nonumber
				\end{align} 
				Because of the Lipschitz continuity of $\Phi(\cdot)$, the triangular inequality, and the Cauchy-Schwarz inequality, the first term in the RHS satisfies 
				\begin{equation}
					|      (\Phi_1 - \Phi_2  ) \phi_1     | |\langle \pi_{1\ell}, Z_i\rangle_{\mathcal H} | \leq  (\phi(0))^2  (\Vert \theta_1 - \theta_2 \Vert \Vert \pi_{1,\ell}\Vert_{\mathcal H}  + \Vert \pi_1 -\pi_2 \Vert_{\mathcal H^{d_e}}\Vert \theta_2\Vert  )\Vert  \pi_{1,\ell}\Vert_{\mathcal H}\Vert Z_i\Vert_{\mathcal H} ^2 .  \label{lem5eq1} \nonumber
				\end{equation}
				Similarly, the second term satisfies
				\begin{equation}
					|   (y_i - \Phi_1  ) (\phi_1     - \phi_2    )   | |\langle \pi_{1\ell} ,Z_i\rangle_{\mathcal H}  |  \leq 2 |\phi'(1)| (\Vert \theta_1 - \theta_2 \Vert \Vert \pi_{1,\ell}\Vert_{\mathcal H}  + \Vert \pi_1 -\pi_2 \Vert_{\mathcal H^{d_e}}\Vert \theta_2\Vert  )\Vert  \pi_{1,\ell}\Vert_{\mathcal H}\Vert Z_i\Vert_{\mathcal H} ^2 .
					%2 \Vert \pi_{1\ell}\Vert  |  \Phi_1(\Phi_1^{-1}   \phi_1  - \Phi_1^{-1} \phi_2)  |\Vert Z_i\Vert  \nonumber\\
					%= c|\Phi_1||\Phi_1 ^{-1}\phi_1 - \Phi_2^{-1}\phi_2 + \Phi_2 ^{-1} \phi_2 - \Phi_1 ^{-1} \phi_2|\Vert Z_i \Vert_{\tau} \nonumber  \\
					% \leq  c_1  (| \widetilde    \phi_1 - \widetilde  \phi_2    |  +     |      \Phi_1 - \Phi_2    | )\Vert \pi_{1\ell}\Vert  \Vert Z_i\Vert   \leq c_2  ( \Vert \pi_1 - \pi_2      \Vert  +  \Vert \theta_1 - \theta_2  \Vert   )   \Vert Z_i  \Vert   ^2  \Vert \Pi_1 \Vert _{\op} .
					\label{lem5eq2} \nonumber
				\end{equation}
				%	for some constants $c_1$ and $c_2$, where the second inequality follows from the following:
				%	\begin{equation}
					%		|\phi_1 - \phi_2| = |\Phi_1 (\Phi_1 ^{-1} \phi_1 -\Phi_2^{-1}\phi_2 - (\Phi_1 ^{-1} - \Phi_2 ^{-1})\phi_2) | \leq |\Phi_1 (\widetilde \phi_1 - \widetilde \phi_2)| + |(\Phi_2 - \Phi_1)\Phi_2^{-1} \phi_2| . \label{eq:phi}
					%	\end{equation}
				%	Because of Assumption \ref{ass9}, the RHS of \eqref{eq:phi} is bounded above by $c_1|\widetilde \phi_1 - \widetilde \phi_2| + c_1|\Phi_2 - \Phi_1|$ for a positive constant $c_1$. Then, \eqref{lem5eq2} is given as a consequence of the mean-value theorem, Cauchy-Schwarz inequality,  Lipschitz continuity of $\widetilde \phi (\cdot )$, and boundedness of $\phi (\cdot)$. 
				
				% Both \eqref{lem5eq1} and \eqref{lem5eq2} are bounded by $cb_\ell  (Z_i )  ( \Vert \theta_1 - \theta_2 \Vert +  \Vert \pi_1 - \pi_2      \Vert_{\mathcal H^{d_e}} )$ for some $c>0$ and with $b_j  (Z_i ) = \Vert Z_i  \Vert_{\mathcal H}^2 $ for $\ell=1,\ldots,d_e$, and thus satisfy  Condition (3.1) in \cite{Chen2003}.
				\noindent The last term is bounded as follows:
				%	\begin{align*}
					%		& ( \mathbb E  [ \underset{ (\theta, \Pi ) :  \Vert \theta_1 - \theta_2  \Vert <\epsilon ,  \Vert \pi_1 - \pi_2      \Vert_{\mathcal H^{d_e}} <\epsilon  }{\sup }  |  (y_i - \Phi_2   )\phi_2   |^2 |\langle \pi_{1\ell } - \pi_{2\ell}, Z_i \rangle |^2  ]  )^{1/2}\nonumber\\
					%		& \leq   ( \mathbb E  [  \underset{ (\theta, \Pi ) :  \Vert \theta_1 - \theta_2  \Vert <\epsilon ,  \Vert \pi_1 - \pi_2      \Vert_{\mathcal H^{d_e}}  <\epsilon  }{\sup }  \phi^2(0) \Vert \pi_1 - \pi_2      \Vert_{\mathcal H^{d_e}}  ^2  \Vert Z_i  \Vert_{\mathcal H}   ^2    ]  )^{1/2}  \leq  \epsilon \phi^2(0) (   \mathbb E  [    \Vert Z_i \Vert_{\mathcal H}  ^2  ]  )^{1/2} , \label{lem5eq3}
					%	\end{align*}
				\begin{equation}
					|  (y_i - \Phi_2 )\phi_2   | |\langle \pi_{1\ell}  -   \pi_{2\ell}, Z_i\rangle _{\mathcal H}  |\leq 2\phi(0)\Vert  \pi_{1\ell}  -   \pi_{2\ell}\Vert_{\mathcal H^{d_e}} \Vert Z_i\Vert_{\mathcal H}. \nonumber
				\end{equation} 
				If we let  $\Vert n^{-1/2}Z_i\Vert_\mathcal H ^2$ be the bound $M_{ni}$ in \citet[p.213]{Vaart1996}, then it  satisfies  $\sum_i    \mathbb E[\Vert n^{-1/2}Z_i\Vert_\mathcal H ^4 ] = O(1)$ under Assumption~\ref{ass9:finite}.  
				%
				%For $j=d_e +1 ,\ldots, 2d_e$, the conditions hold by replacing $ \Vert  Z_i  \Vert ^2  $ in \eqref{lem5eq1} and \eqref{lem5eq2} by $ \Vert Z_i  \Vert  ^2 +  \Vert Z_i  \Vert   \Vert V_i  \Vert$. Moreover, the proof for the RCMLE can be easily obtained from Assumption~\ref{ass9}.
			\end{proof}

			Lemma \ref{lem7} proves Condition (2.6) in \cite{Chen2003}.
			\begin{lemma}\label{lem7}\normalfont
				Suppose that the conditions for Theorem \ref{prop4} hold. Then, for $j\in \{M,N\}$, as $n \hto \infty$\begin{equation*}
					\sqrt{n} \mathcal J_{j} ^{-1/2}\mathcal S_n ^{-1} \left( \mathcal M_{jn} \left(\theta_{0} , g_{i}\right) + \Gamma_{2,0j} \left(\widehat \pi_\alpha - \pi_{n}\right)\right) \dto \mathcal N \left(0, I_{2d_e} \right).
				\end{equation*} 
			\end{lemma}
			\begin{proof}[Proof of Lemma \ref{lem7}]
				We first note that the second term satisfies  \begin{align}
					&\sqrt{n } \mathcal S_n ^{-1} \Gamma_{2,0j}  (\widehat \pi - \pi_{n} )  
					=  \sqrt{n}\mathcal V_{j} (\mathcal K_{n\alpha}^{-1} \mathcal K_n - \mathcal I ) \Pi_n ^\ast \psi_0 + \mathcal V_{j} \mathcal K_{n\alpha}^{-1} \frac{1}{\sqrt{n}}\sum_{i=1} ^n V_i'\psi_0 Z_i , \nonumber\\
					& \quad=  \sqrt{n}\mathcal V_{j} (\mathcal K_{n\alpha}^{-1} \mathcal K_n - \mathcal I ) \Pi_n ^\ast \psi_0 + \frac{1}{\sqrt{n}}\sum_{i=1} ^n V_i'\psi_0 \mathcal V_{j} ( \mathcal K_{n\alpha}^{-1} - \mathcal K_\alpha ^{-1})Z_i + \mathcal V_{j} \mathcal K_{\alpha}^{-1} \frac{1}{\sqrt{n}}\sum_{i=1} ^n V_i'\psi_0 Z_i   . \label{eq:asy:new:eq1}
				\end{align}
				The proof consists of three steps. In the first two, we show the first two terms in \eqref{eq:asy:new:eq1} are $o_p(1)$, and then we derive the limiting distribution of the sum of the last term in \eqref{eq:asy:new:eq1} and $\sqrt{n}\mathcal S_n^{-1}\mathcal M_{jn}(\theta_0, g_i)$.
				
				\textbf{Step 1:}  Because there exists a bounded linear operator $\mathcal C_j$ such that $\mathcal V_j = \mathcal C_j \mathcal K ^{1/2}$ (\citealp[Theorem 1]{baker1973joint}),  Lemma~\ref{lem6prime} tells us that   \begin{equation*} 
					\Vert \mathcal V_j (\mathcal K_{n\alpha}^{-1}\mathcal K_n - \mathcal I) \Pi_n ^\ast \psi_0  \Vert  ^2 \leq \sum_{\ell=1} ^{d_e}\Vert \mathcal C_j \Vert_{\op} ^2 \Vert    \mathcal K^{1/2} (\mathcal K_{n\alpha}^{-1} \mathcal K_n -\mathcal I) \pi_{\ell,0} \Vert_{\mathcal H}  ^2  \Vert \psi_0\Vert ^2 = o_p(n^{-1}) . \label{eq:tmp1}
				\end{equation*}

				\textbf{Step 2:} We consider the second term in  \eqref{eq:asy:new:eq1}. Note that $  \mathcal V_j  (\mathcal K_{n\alpha} ^{-1} - \mathcal K_{\alpha} ^{-1})    $ can be decomposed into the sum of $\widetilde{\mathcal A}_{1}$, $\widetilde{\mathcal A}_2$ and $\widetilde{\mathcal A}_3$ each of which is defined by \begin{equation*} 
					\widetilde{\mathcal A}_1 = \mathcal C_j \mathcal K^{1/2} \mathcal K_{n\alpha} ^{-1} (\mathcal K_{n} - \mathcal K) \mathcal K_{\alpha} ^{-1},\quad \widetilde{\mathcal A}_2 =  \mathcal C_j \mathcal K^{1/2} \mathcal K_{n\alpha} ^{-1} ( \mathcal K - \mathcal K_\alpha ) \mathcal K_{\alpha} ^{-1}  ,\quad\widetilde{\mathcal A}_3 =   \mathcal C_j \mathcal K^{1/2} \mathcal K_{n\alpha} ^{-1} (\mathcal K_{n\alpha} - \mathcal K_n) \mathcal K_{\alpha} ^{-1} ,   \label{eq:tmp2}
				\end{equation*}
				where the subscript $j$'s on $\widetilde{\mathcal A}_{1}, \widetilde{\mathcal A}_2$ and $\widetilde{\mathcal A}_3$  are suppressed for the ease of exposition. We first consider $n^{-1/2}\sum_{i=1} ^n V_i ' \psi_0 \widetilde{\mathcal A}_1 Z_i$. To obtain the  result, we  decompose it   into $ \mathcal C_j \mathcal K^{1/2}\mathcal K_{n\alpha} ^{-1} n^{-1/2}\sum_{i=1} ^n \widetilde{\mathcal K}_{-i}V_i'\psi_0 \mathcal K_\alpha ^{-1}Z_i $ and $\mathcal C_j \mathcal K^{1/2}\mathcal K_{n\alpha} ^{-1} n^{-1/2}\sum_{i=1} ^n \widetilde{\mathcal K}_{i}V_i'\psi_0\mathcal K_\alpha ^{-1} Z_i $ with $\widetilde{\mathcal K}_{-i} = n^{-1}\sum_{\ell \neq i} (Z_\ell \otimes Z_\ell - \mathcal K)$ and $\widetilde{\mathcal K}_{i} = n^{-1}(Z_i \otimes Z_i -\mathcal K )$. Note that $\mathbb E[\widetilde{\mathcal K}_{-i}V_i'\psi_0 \mathcal K_\alpha ^{-1}Z_i ] = 0$ and $ \mathbb E[\Vert\widetilde{\mathcal K}_{-i}V_i'\psi_0 \mathcal K_\alpha ^{-1} Z_i \Vert _{\mathcal H} ^2 ]\leq \sigma_{  \psi_0} ^2 \Vert \mathcal K^{1/2}\mathcal K_{\alpha} ^{-1}\Vert_{\HS} ^2  \mathbb E[\Vert \widetilde{\mathcal K}_{-i}\Vert_{\HS} ^2 ]   \leq c \alpha ^{-1/2} \mathbb E[\Vert \widetilde{\mathcal K}_{-i}\Vert_{\HS} ^2 ]  = O_p(n^{-1}\alpha^{-1/2})$ (\citealt[Theorem 4.1]{Bosq2000}). Then, by Markov's inequality,  $\Vert n^{-1/2}\sum_i \widetilde{\mathcal K}_{-i}V_i'\psi_0 Z_i \Vert_{\mathcal H} = O_p(n^{-1/2}\alpha^{-1/4})$. Moreover, $ \Vert \mathcal K^{1/2}\mathcal K_{n\alpha} ^{-1}\Vert_{\op} ^2 \leq (1+o_p(1)) \Vert \mathcal K_{n\alpha} ^{-1} \Vert_{\op} = O_p(\alpha^{-1/2})$. Thus, we have  $$\Vert \mathcal C_j \mathcal K^{1/2}\mathcal K_{n\alpha} ^{-1} n^{-1/2}\sum_{i=1} ^n \widetilde{\mathcal K}_{-i}V_i'\psi_0 \mathcal K_\alpha ^{-1}Z_i\Vert  \leq\Vert \mathcal K^{1/2}\mathcal K_{n\alpha}^{-1}\Vert_{\op} O_p(n^{-1/2}\alpha^{-1/4})   =  O_p((n\alpha)^{-1/2}), $$ which is $o_p(1)$ because $\mu_m ^2 \alpha \to \infty$. The second term associated with $\widetilde{\mathcal K}_i$ satisfies  $\mathbb E[  \widetilde{\mathcal K}_i V_i'\psi_0 \mathcal K_\alpha^{-1}Z_i] = 0$ and $\mathbb E[ \Vert  \widetilde{\mathcal K}_i V_i'\psi_0 \mathcal K_\alpha^{-1}Z_i\Vert_{\mathcal H} ^2]  \leq O(n^{-2})\mathbb E[\Vert \mathcal K_{\alpha} ^{-1}\Vert_{\op}^2 \Vert Z_i\Vert_{\mathcal H}^4] = O(n^{-2}\alpha^{-1})$. Using similar arguments as above,   $\mathcal C_j \mathcal K^{1/2}\mathcal K_{n\alpha} ^{-1} n^{-1/2}\sum_{i=1} ^n \widetilde{\mathcal K}_{i}V_i'\psi_0\mathcal K_\alpha ^{-1} Z_i  = o_p(1)$, and thus  $n^{-1/2}\sum_{i=1} ^n V_i ' \psi_0 \widetilde{\mathcal A}_1 Z_i = o_p(1)$.  %due to the independence across $i$,\begin{align*}

					We now consider $n^{-1/2}\sum_{i=1} ^n V_i ' \psi_0 \widetilde{\mathcal A}_2 Z_i$. Note that $\Vert (\mathcal K_\alpha - \mathcal K)\mathcal K_\alpha ^{-1}\Vert_{\op} = O_p(\alpha^{1/2})$ (Assumption~\ref{asskinv} and arguments similar to those in \citealt[p.394]{Carrasco2012}) and $\Vert \mathcal K^{1/2}\mathcal K_{n\alpha} ^{-1}\Vert_{\op}   = O_p(\alpha^{-1/4})$. Therefore, $\Vert \widetilde{\mathcal A}_2 \Vert_{\op}  = O_p(\alpha^{1/4})$. Combining this with the central limit theorem applied to $n^{-1/2}\sum_i Z_i V_i'\psi_0$, we have  
					\begin{equation}
						\Vert n^{-1/2}\sum_{i=1} ^n V_i ' \psi_0 \widetilde{\mathcal A}_2 Z_i\Vert \leq \Vert \widetilde{\mathcal A}_2 \Vert_{\op} \Vert n^{-1/2}\sum_i V_i'\psi_0 Z_i\Vert_{\mathcal H} = o_p(1) .\label{eq:temp3}
					\end{equation}
					By using similar arguments, $n^{-1/2}\sum_i V_i'\psi_0 \widetilde{\mathcal A}_3 Z_i = o_p(1)$. Thus, the second term in  \eqref{eq:asy:new:eq1} is $ o_p(1)$.
					%In addition, for any $h \in \ran \mathcal K$, there exists $h_1 \in \mathcal L_\tau ^2 $ such that $h = \mathcal K h_1$ and $ \Vert h_1  \Vert <\infty$. Hence, $ \Vert \mathcal K_{\alpha} ^{-1}h  \Vert \leq  \Vert \mathcal K_{\alpha} ^{-1} \mathcal K  \Vert _{op}  \Vert h_1  \Vert \leq O_p(1)$, and thus,  \begin{equation}
						% \Vert  (\mathcal K_{n\alpha} ^{-1}  - \mathcal K_{\alpha} ^{-1}  ) h  \Vert \leq  \Vert \mathcal K_{n\alpha} ^{-1}  \Vert_{op}  \Vert \mathcal K_n - \mathcal K  \Vert_{op}  \Vert \mathcal K_{\alpha} ^{-1} h  \Vert + O_p(\alpha^{1/2}) \leq O_p (n^{-1/2}\alpha^{-1/2} ) \label{eq:asy:new:eq2}
						%\end{equation}
						%where the second inequality is from $A^{-1} - B^{-1} = B^{-1} (B-A) A^{-1}$ and the last is by Assumption \ref{ass10}-\ref{ass10-1} and Theorem  2.7 in \cite{Bosq2000}. Hence, 
						%\begin{equation}
						% \sqrt{n } \mathcal S_n ^{-1} \Gamma_{2,0j}  (\widehat \pi - \pi_{0} )  =  \sqrt{n}\mathcal V_{j}  (\widehat \Pi^\ast - \Pi_n ^\ast  ) \psi_0  =  \mathcal V_{j} \mathcal K_{\alpha} ^{-1} \frac{1}{\sqrt{n}} \sum_{i=1} ^n  V_i ' \psi_0 Z_i + o_p  (1). \label{eq.new.7}
						%\end{equation} 
						
						\textbf{Step 3:}	Let $\mathcal J_{2,j\alpha } = \sigma_{  \psi_0} ^2 \mathcal V_j \mathcal K_{\alpha} ^{-1} \mathcal K \mathcal K_{\alpha} ^{-1} \mathcal V_j^\ast$ and $\mathcal J_{j\alpha} = \mathcal J_{1,j} + \mathcal J_{2,j\alpha}$. Then, the following holds from the results in Steps 1 and 2, the central limit theorem, the continuous mapping theorem, and the orthogonality between $m_{1j} (\theta_0,g_i )   g_{0i}$ and $V_i '\psi_0\mathcal V_j \mathcal K ^{-1}  Z_i$:
						\begin{align}
							&\sqrt{n} \mathcal J_{j\alpha} ^{-1/2} \mathcal S_n ^{-1}  \left(\mathcal M _{jn}  (\theta_{0}  , g_{i } ) + \Gamma _{2j}  (\theta_{0}  , g   )  (\widehat \pi_\alpha - \pi_{n} ) \right) \nonumber\\
							&= \frac{1}{\sqrt{n}}\mathcal J_{j\alpha} ^{-1/2} \sum_{i=1} ^n (m_{1j} (\theta_0,g_i )  g_{0i}+  V_i ' \psi_0 \mathcal V_j \mathcal K_\alpha ^{-1} Z_i) +o_p(1)  \dto \mathcal N  (0, \mathcal I_{2d_e}  ). \nonumber
						\end{align} 
						
						It remains to show that $\Vert \mathcal J_{j\alpha} - \mathcal J_j\Vert_{\HS} = \Vert \mathcal J_{2,j\alpha }- \mathcal J_{2,j}\Vert_{\HS} = o_p(1) $. It follows from \citet[Theorem 1]{baker1973joint} and the commutative property between $\mathcal K^{1/2}$, $\mathcal K$ and $\mathcal K_{\alpha} ^{-1}$; specifically, we have $
						\Vert \mathcal J_{2,j\alpha }- \mathcal J_{2,j}\Vert_{\HS}   \leq \Vert\mathcal C_j\Vert_{\HS} ^2 \Vert\mathcal K_\alpha ^{-2}   \mathcal K^2 - \mathcal I  \Vert_{\op} = o(1),$ because of Assumption \ref{asskinv} and arguments similar to those in \citet[p.394]{Carrasco2012}. 
					\end{proof} 
					\begin{lemma}\label{lem8} \normalfont
						Suppose that the conditions in Theorem \ref{prop4} hold. Then, $
						\Vert \mathcal S_n ^{-1}  (\widehat \Gamma_{1,j} \hspace{-.1em}-\hspace{-.1em} \Gamma_{1,0j}  )\mathcal S_n ^{\prime -1}  \Vert_{\HS} \hspace{-.2em}=\hspace{-.2em} o_p(1) \label{lem8-eq2}
						$.
					\end{lemma}
					\begin{proof}[Proof of Lemma \ref{lem8}]
						\begin{equation*}
							\begin{aligned}
								& \Vert \mathcal S_n ^{-1}  (\widehat \Gamma_{1,N} - \Gamma_{1,0N}  )\mathcal S_n ^{\prime -1}  \Vert_\HS\leq   \Vert  \frac{1}{n} \sum_{i=1} ^n  \phi^2  (\widehat g_{i,\alpha} ' \widehat \theta_N )\mathcal S_n ^{-1} (\widehat g_{i,\alpha} \widehat g_{i,\alpha} ' -g_{i} g_{i} '  )\mathcal S_n ^{\prime -1}  \Vert_\HS \\
								& +   2\phi(0) \Vert\frac{1}{n}\sum_{i=1} ^n   (\phi    (g_i ' \theta_{0}  )   -\phi  (\widehat  g_{i,\alpha} ' \widehat \theta_N  ) )\mathcal S_n ^{-1} g_{i} g_{i} ' \mathcal S_n ^{\prime -1}  \Vert_\HS+ \Vert  \frac{1}{n} \sum_{i=1} ^n \mathcal S_n ^{-1}  (\phi^2 (g_{i} ' \theta_{0} )g_{i} g_{i} ' - \Gamma_{1,0N}  )\mathcal S_n ^{\prime -1} \Vert_\HS  .
							\end{aligned} 
						\end{equation*}
						The first term in the RHS satisfies
						\begin{align}
							&  n^{-1} \sum_{i=1} ^n \Vert\phi^2  (\widehat g_{i,\alpha} ' \widehat \theta_N )\mathcal S_n ^{-1} (\widehat g_{i,\alpha} \widehat g_{i,\alpha} ' -g_{i} g_{i} '  )\mathcal S_n ^{\prime -1}  \Vert   _\HS \nonumber\\
							&\leq \phi^2 (0)\left( n^{-1}\sum_{i=1} ^n   \Vert \mathcal S_n ^{-1 } (\widehat g_{i,\alpha} - g_i)\Vert^2 + 2     (n^{-1} \sum_{i=1 }^n \Vert  g_{i0} \Vert ^2 )^{1/2}  (n^{-1}\sum_{i=1} ^n \Vert \mathcal S_n^{-1} (\widehat g_{i,\alpha} - g_i) \Vert^2 )^{1/2} \right) , \label{eq:new:1}
						\end{align}
						where the inequality follows from the triangular inequality, the H\"older's inequality, and the boundedness of $\phi  (\cdot)$. The bounds in \eqref{p3:eq1} and \eqref{p3:eq4} tell us that \eqref{eq:new:1} is $ o_p(1)$. By using similar arguments and the Lipschitz continuity of $\phi (\cdot)$, we also have
						\begin{equation}
							n^{-1}\sum_{i=1} ^n \Vert (\phi   (g_i ' \theta_{0}  )   -\phi (\widehat  g_{i,\alpha} ' \widehat \theta_N  ) )\mathcal S_n ^{-1} g_{i} g_{i} ' \mathcal S_n ^{\prime -1}  \Vert _\HS \leq O_p(1) ( n^{-1} \sum_{i=1} ^n |g_i ' \theta_0 - \hat{g}_{i,\alpha} '\hat{\theta}_N|   ^2    )^{1/2}  	 = o_p(1)
							%	&\leq c (    (\frac{1}{n} \sum_{i=1} ^n  \Vert \mathcal S_n ^{-1} (\widehat g_i - g_i) \Vert ^2 )^{1/2} + \Vert \widehat \theta_N - \theta_0  \Vert   ) (\frac{1}{n}\sum_{i=1} ^n  \Vert \mathcal S_n ^{-1} g_i  \Vert ^2 )^{1/2}   
							. \label{eq:new:2}
						\end{equation} 
						%	for a positive constant $c$. The first inequality follows from \eqref{eq:phi}, and the second is obtained by applying the mean value theorem to $\widetilde \phi(\cdot)$ and $\Phi (\cdot)$ and the H\"older's inequality.  \eqref{eq:new:2} is bounded above by $o_p(1)$ because of \eqref{p3:eq1}, \eqref{p3:eq4} and Theorem \ref{thm1:consistency}.
						Lastly, $\mathbb E [\phi^2 (g_i ' \theta_0)\mathcal S_n^{-1}g_i g_i' \mathcal S_n^{\prime-1}] = \Gamma_{1,0N}$ and $\mathbb E [ \Vert \phi^2 (g_i ' \theta_0) \mathcal S_n ^{-1}g_i g_i' \mathcal S_n^{\prime-1}\Vert_\HS^2   ]  =O(1) $ by Assumptions \ref{ass3:finite} and \ref{ass9:finite}, and thus  
						\begin{equation}
							\Vert n^{-1} \sum_{i=1} ^n \mathcal S_n ^{-1} \left(\phi^2\left(g_{i} ' \theta_{0}\right)g_{i} g_{i} ' - \Gamma_{1,0N} \right)\mathcal S_n ^{\prime -1} \Vert_\HS  = o_p (1),\label{eq:new:3}
						\end{equation}
						by the law of large numbers. Hence, Lemma \ref{lem8} follows from \eqref{eq:new:1}-\eqref{eq:new:3}.
					\end{proof}
					\begin{lemma}\label{lem9} \normalfont
						Suppose that the conditions in Theorem \ref{prop4} hold. Then, $
						\Vert \mathcal S_n ^{-1}  \widehat {\mathcal J}_{1,j}\mathcal S_n ^{\prime -1}  -  {\mathcal J}_{1,j}   \Vert_\HS  = o_p(1) $.
					\end{lemma}
					The proof of Lemma~\ref{lem9} is similar to that of Lemma \ref{lem8}. Thus, the details are omitted. 
					\begin{lemma}\label{lem10} \normalfont
						Suppose that the conditions in Theorem \ref{prop4} hold. Then, $ \Vert \mathcal S_n ^{-1} \widehat {\mathcal J}_{2,j} \mathcal S_n ^{\prime -1} - \mathcal J_{2,0j}  \Vert _\HS  = \hspace{-.2em}o_p(1)$.
					\end{lemma}
					\begin{proof}
						Let $\widehat {\sigma}_{  \psi_0}^2 = \frac{1}{n} \sum_{i} (\psi_0 ' V_i )^2$. Then, by the law of large numbers,   $     \widehat {\sigma}_{  \psi_0}^2  -  {\sigma}_{  \psi_0}^2    =o_p(1)$.  Furthermore, by using similar arguments in \eqref{eq:new:1} and the consistency of $\widehat \theta_j$, we have $ 
						\Vert  n^{-1}\sum_i V_iV_i' - n^{-1}\sum_i \widehat V_{i,\alpha} \widehat{V}_{i,\alpha} ' \Vert_{\HS}     = o_p(1) ,$ and $\Vert \widehat{\psi}_j - \psi_0\Vert = o_p(1)$. Therefore, we have $ \widehat\sigma_{\widehat\psi_j} ^2 - \sigma_{  \psi_0} ^2  = o_p(1) $.
						
						Let ${\mathcal V}_{jn} = \frac{1}{n} \sum_{i} \dot{m}_{2j}  (\theta_{0},g_{i} )  Z_i \otimes  g_{0i}$. It satisfies the conditions for Lemma 3.1 in \cite{chen1998central} and thus $ \Vert \mathcal V_{jn} - \mathcal V_{j}   \Vert _{\HS} = O_p(n^{-1/2})$. Because  $\mathcal V_j = \mathcal C_j \mathcal K^{1/2}$, $ 1/(\mu_{m,n}^2\alpha) = o(1)$ and $\Vert\mathcal J_{2,j\alpha} - \mathcal J_{2,j}\Vert_{\HS} = o_p(1)$, we have \begin{align}
							&\Vert  \sigma_{  \psi_0} ^{-2}\mathcal{J}_{2,j} - \mathcal V_{jn} \mathcal K_\alpha^{-1}\mathcal K \mathcal K_\alpha ^{-1} \mathcal V_{jn} ^\ast\Vert_\HS   = \Vert   \mathcal V_j \mathcal K_\alpha^{-1}\mathcal K \mathcal K_\alpha ^{-1} \mathcal V_j ^\ast - \mathcal V_{jn} \mathcal K_\alpha^{-1}\mathcal K \mathcal K_\alpha ^{-1} \mathcal V_{jn} ^\ast\Vert_\HS +o_p(1) \nonumber\\
							&\quad\leq 2\Vert \mathcal C_j \mathcal K^{1/2}\mathcal K_\alpha ^{-1 }\mathcal K \mathcal K_\alpha ^{-1}   \Vert_{\op} \Vert \mathcal V_{jn} -\mathcal V_j \Vert_{\HS} + \Vert \mathcal K_\alpha ^{-1}\mathcal K\mathcal K_\alpha^{-1}\Vert_{\op}\Vert \mathcal V_{jn} - \mathcal V_j\Vert_{\HS}^2 =o_p(1).\nonumber
						\end{align}   
						
						It remains to show that $\Vert \mathcal V_{jn} \mathcal K_\alpha^{-1} \mathcal K \mathcal K_{\alpha} ^{-1} \mathcal V_{jn}^\ast - \widehat {\mathcal V}_{jn} \mathcal K_{n\alpha} ^{-1} \mathcal K_n \mathcal K_{n\alpha} ^{-1} \widehat {\mathcal V}_{jn}\Vert_\HS = o_p(1)$. To this end, we let $\widetilde{\mathcal D}_{j1,\ell}  = n^{-1}\sum_{i=1} ^n  \dot{m}_{2j}( \theta_0, {g}_{i}) Z_i \langle Z_i,  (\widehat{\Pi}_{n,\alpha} - \Pi_n )^\ast\mathcal S_n^{\prime -1}   e_\ell \rangle $, $\widetilde{\mathcal D}_{j2, \ell } = n^{-1} \sum_{i=1} ^n (\dot{m}_{2j}(\hat\theta_j, \hat{g}_{i,\alpha}) - \dot{m}_{2j}(\theta_0, {g}_{i}) ) Z_i  g_{0i} ' e_\ell $ and  $\widetilde{\mathcal D}_{j3,\ell}  = n^{-1}\sum_{i=1} ^n(\dot{m}_{2j}(\hat\theta_j, \hat{g}_{i,\alpha}) - \dot{m}_{2j}(\theta_0, {g}_{i}) ) Z_i \langle Z_i,  (\widehat{\Pi}_{n,\alpha} - \Pi_n )^\ast\mathcal S_n^{\prime -1}   e_\ell \rangle $. Then, we have $
						\widehat{\mathcal V}_{jn} ^\ast - \mathcal V_{jn} ^\ast  =   \sum_{\ell=1} ^{d_e} (\widetilde{\mathcal D}_{j1,\ell} + \widetilde{\mathcal D}_{j2, \ell }+ \widetilde{\mathcal D}_{j3, \ell } )e_\ell'$. Therefore, we have \begin{equation}
							\Vert \mathcal K_{n\alpha} ^{-1/2} (\widehat{\mathcal V}_{jn} ^\ast - \mathcal V_{jn} ^\ast )\Vert_{\HS} ^2  \leq 3\sum_{\ell=1} ^{d_e} (\Vert \mathcal K_{n\alpha}^{-1/2} \widetilde{\mathcal D}_{j1,\ell}\Vert_{\HS} ^2 +  \Vert \mathcal K_{n\alpha}^{-1/2} \widetilde{\mathcal D}_{j2,\ell}\Vert_{\HS} ^2+\Vert \mathcal K_{n\alpha}^{-1/2} \widetilde{\mathcal D}_{j3,\ell}\Vert _{\HS}^2 ). \label{eq: temp00}
						\end{equation}  Below, we will show that each term  in the RHS of \eqref{eq: temp00} is $o_p(1)$.
						For each $\ell$, $\widetilde{\mathcal D}_{j1,\ell}$ satisfies that\begin{equation}
							\widetilde{\mathcal D}_{j1,\ell}  % n^{-1} \sum_{i=1} ^n \dot{m}_{2j} (\theta_0, g_{0i}) Z_i \otimes Z_i (  (\widehat{\Pi}_{n,\alpha} - \Pi_n )^\ast\mathcal S_n^{\prime -1}   e_\ell )
							=(\mathbb E[ \dot{m}_{2j} (\theta_0, g_{i}) Z_i \otimes Z_i ] +O_p(n^{-1/2}))((\widehat{\Pi}_{n,\alpha} - \Pi_n )^\ast\mathcal S_n^{\prime -1}   e_\ell  )     , \label{eq: temp01}
						\end{equation}
						where the  equality is obtained by applying the central limit theorem to $n^{-1}\sum_i \dot{m}_{2j} (\theta_0, g_{i}) Z_i \otimes Z_i$. By using similar arguments in proving \eqref{p3:eq1}-\eqref{p3:eq4}, we have $\Vert(\widehat{\Pi}_{n,\alpha} - \Pi_n )^\ast\mathcal S_n^{\prime -1}  \Vert_{\HS} = o_p(1)$, and because $\mathbb E[ \dot{m}_{2j} (\theta_0, g_{0i}) Z_i \otimes Z_i ] = \mathcal K^{1/2}\dot{\mathcal D}_{1j}$ for a bounded linear operator $\dot{\mathcal D}_{1j}$ (\citealt[Theorem 1]{baker1973joint}), $\Vert \mathcal K_{n\alpha} ^{-1/2}\mathcal K^{1/2}\Vert_{\op}^2 = \Vert \mathcal K_{n\alpha}^{-1/2}\mathcal K\mathcal K_{n\alpha} ^{-1/2}\Vert_{\op} = O_p(1)$ and $n^{-1/2}\alpha^{-1/2} = o(1)$, \eqref{eq: temp01} satisfies  that \begin{equation}
							\Vert \mathcal K_{n\alpha} ^{-1/2}	\widetilde{\mathcal D}_{j1,\ell} \Vert_{\HS} = (\Vert \mathcal K_{n\alpha} ^{-1/2}\mathcal K^{1/2}\Vert_{\op}\Vert \dot{\mathcal D}_{1j}\Vert_{\HS} + O_p((n\alpha)^{-1/2}))o_p(1)   = o_p(1). \label{eq: temp0101}
						\end{equation}
						Next, by  the mean-value theorem, we have that for some bounded linear operators $\dot{\mathcal D}_{j2,\ell}$ and $\ddot{\mathcal D}_{j2,\ell }$ (\citealt[Theorem 1]{baker1973joint}), \begin{align}
							\widetilde{\mathcal D}_{j2,\ell} & =   n^{-1} \sum_{i=1} ^n \widetilde{\dot{m}}_{2i} ^{(1)}   g_{0i}'e_\ell Z_i g_{i}' (\hat \theta_j -\theta_0)  +    n^{-1} \sum_{i=1} ^n \widetilde{\dot{m}}_{2i} ^{(1)}   g_{0i}'e_\ell Z_i (\hat{g}_{i,\alpha} - g_{i})'  \hat \theta_j   \nonumber\\
							&= \mathbb E[ \dot{m}_{2j,0} ^{(1)}   g_{0i}'e_\ell Z_i g_{0i}'] \mathcal S_n'(\hat \theta_j - \theta_0)  + \mathbb E[   \dot{m}_{2j,0} ^{(1)}     g_{0i}'e_\ell Z_i  \otimes Z_i] (\mathcal S_n ^{-1}(\widehat{\Pi}_{n\alpha} - \Pi_n)) ^\ast \mathcal S_n '\hat{\theta}_j   +O_p(n^{-1/2})  \nonumber\\
							& = \mathcal K^{1/2}\dot{\mathcal D}_{j2,\ell}  O_p(\mu_{m,n}^{-1})+ \mathcal K^{1/2}\ddot{\mathcal D}_{j2,\ell}  o_p(1) + O_p(n^{-1/2}),  \label{eq: temp2}
						\end{align}
						where $\widetilde{\dot{m}}_{2j} ^{(1)} $ denotes the first derivative of $ \dot{m}_{2j}(\theta , \mathcal G_i(\pi))   $ evaluated at $ (\theta, \pi)$ satisfying $ \theta'\mathcal G_i(\pi) \in [\widehat \theta_j ' \widehat{g}_{i,\alpha}, \theta_0'g_{i,0}]   $ and $\widetilde{\dot{m}}_{2j,0} ^{(1)}   = \dot{m}_{2j} ^{(1)} (\theta_0 , g_{i}) +o(1)$. Therefore, from similar arguments in obtaining \eqref{eq: temp0101}, we have \begin{equation}
							\Vert  \mathcal K_{n\alpha} ^{-1/2} \widetilde{\mathcal D}_{j2,\ell} \Vert_{\HS}   = o_p(1). \label{eq: temp3}
						\end{equation} 
						Before moving to the last term, note that  $   \Vert\mathcal S_n^{-1}( \widehat \Pi_{n,\alpha} - \Pi_{n})\Vert_{\op} = o_p(1)$ and   $n^{-1}\sum_{i=1} ^n\Vert \mathcal S_n ^{-1}(\widehat{g}_{i,\alpha} - g_i)\Vert ^2=O_p(\alpha^2) + O_p(1/(\mu_{m,n}^2 \alpha^{1/2})) = O_p(\alpha^{1/2})$ under the employed conditions. Then, by using the Lipschitz continuity of $\dot{m}_{2j}(\cdot, \mathcal G_i(\cdot))$ and the Cauchy-Schwarz inequality, we find that \begin{align}
							\Vert \mathcal K_{n\alpha} ^{-1/2}\widetilde{\mathcal D}_{j3,\ell}\Vert_{\HS}  
							&\leq   O_p(\alpha^{-1/4}) \Vert \widehat{\theta}_j - \theta_0\Vert n^{-1}\sum_{i=1} ^n \Vert \mathcal S_n^{-1}(\widehat{g}_{i,\alpha}-g_i)\Vert  \Vert \hat g_{i,\alpha}\Vert\Vert Z_i\Vert _{\mathcal H}  \nonumber\\
							&\quad\quad +  O_p(\alpha^{-1/4})\Vert \widehat{\Pi}_{n\alpha} - \Pi_n\Vert_{\op}n^{-1}\sum_{i=1} ^n\Vert \mathcal S_n^{-1}(\widehat{g}_{i,\alpha}-g_i)\Vert\Vert Z_i\Vert _{\mathcal H} =    o_p(1), \label{eq: temp5}
						\end{align}
						because   $\mu_{m,n}^2 \alpha \to \infty$.
						\eqref{eq: temp00}, \eqref{eq: temp0101}, \eqref{eq: temp3} and \eqref{eq: temp5} give that \begin{equation}
							\Vert \mathcal K_{n\alpha} ^{-1/2} (\widehat{\mathcal V}_{jn} ^\ast - \mathcal V_{jn} ^\ast )\Vert_{\HS} ^2  = o_p(1).\label{eq: temp6}
						\end{equation}
						To conclude the proof, note that \begin{align}
							&	\Vert \widehat{\mathcal V}_{jn} \mathcal K_{n\alpha} ^{-1} \mathcal K_n \mathcal K_{n\alpha} ^{-1} \widehat{\mathcal V}_{jn} ^\ast - \mathcal V_{jn} \mathcal K_\alpha^{-1} \mathcal K \mathcal K_{\alpha} ^{-1} \mathcal V_{jn} ^\ast \Vert _\HS \nonumber\\
							&\leq     \Vert \widehat{\mathcal V}_{jn} \mathcal K_{n\alpha} ^{-1} \mathcal K_n \mathcal K_{n\alpha} ^{-1} \widehat{\mathcal V}_{jn} ^\ast -  {\mathcal V}_{jn} \mathcal K_{n\alpha} ^{-1} \mathcal K_n \mathcal K_{n\alpha} ^{-1} {\mathcal V}_{jn} ^\ast \Vert_{\HS}   + \Vert \mathcal V_{jn} (\mathcal K_{n\alpha} ^{-1}\mathcal K_n \mathcal K_{n\alpha} ^{-1} -\mathcal K_{\alpha} ^{-1}\mathcal K \mathcal K_\alpha^{-1})\mathcal V_{jn}^\ast\Vert_\HS \nonumber\\
							&\leq o_p(1) +\Vert \mathcal C_j\Vert _{\HS} ^2 \Vert\mathcal K^{1/2}(\mathcal K_{n\alpha} ^{-1} \mathcal K_n \mathcal K_{n\alpha} ^{-1} - \mathcal K_\alpha ^{-1}\mathcal K \mathcal K_{\alpha} ^{-1})\mathcal K^{1/2}\Vert_{\op}\nonumber\\
							&\leq o_p(1) +O_p(1)\left(\Vert\mathcal K^{1/2} \mathcal K_{n\alpha} ^{-1/2}\Vert_{\op} ^2 \Vert \mathcal K_{n\alpha} ^{-1/2}\Vert_{\op}^2 \Vert \mathcal K_n - \mathcal K\Vert_{\op} +    \Vert \mathcal K^{1/2}  (\mathcal K_{n\alpha} ^{-1}\mathcal K \mathcal{K}_{n\alpha}^{-1} - \mathcal K_\alpha ^{-1} \mathcal K \mathcal K_{\alpha} ^{-1})\mathcal K^{1/2}\Vert_{\op}\right)  \nonumber\\
							&=o_p(1) ,
						\end{align}
						where the first two inequalities follow from the triangular inequality, \eqref{eq: temp6}, $\Vert \mathcal V_{jn} - \mathcal V_j\Vert_{\HS} = O_p(n^{-1/2})$ and   $\Vert  \mathcal K_{n\alpha} ^{-1/2} \mathcal K_n \mathcal K_{n\alpha} ^{-1} \mathcal V_j ^\ast \Vert_{\op} \leq \Vert  \mathcal K_{n\alpha} ^{-1/2} \mathcal K_n \mathcal K_{n\alpha} ^{-1/2} \Vert_{\op}\Vert \mathcal K_{n\alpha} ^{-1/2}\mathcal V_j ^\ast \Vert_{\op}  = O_p(1) $. The remaining inequalities are obtained from $\Vert \mathcal K_{n\alpha} ^{-1}\Vert_{\op} = O_p(\alpha ^{-1/2})$, $\Vert \mathcal K_n - \mathcal K\Vert_{\HS} = O_p(n^{-1/2})$, $\Vert \mathcal K^{1/2}\mathcal K_{n\alpha} ^{-1/2}\Vert_{\op} = O_p(1)$  and using arguments similar to those in the proof of Lemma 3.(i) in \cite{Carrasco2012}.
					\end{proof}

					\section{Appendix~\ref*{app: mse}: Proofs of the results in Section~\ref{sec: mse}}\label{app: mse}
					
					We employ  \citepos{donald2001choosing} approach to study the MSE of $\overline \theta$. We first simplify the notation by letting $m_{1i,0}$, $\dot m_{2i,0}$ and $\ddot{m}_{2i,0}$ respectively denote $m_{1j}(\theta_0, g_i)$, $\dot m_{2j}(\theta_0, \mathcal G_i (\pi_n))$ and $\ddot m_{2j}(\theta_0, \mathcal G_i (\pi_n))$, regardless of $j\in \{M,N\}$. The $n \times 1 $ vectors $\mathfrak{m}_{1,0}$ and $\mathfrak{g}_{0\ell}$ are given by $(m_{11,0} , \ldots, m_{1n,0})'$ and $(g_{01,\ell}, \ldots, g_{0n,\ell})'$ respectively, where $g_{0i,\ell}$ is the $\ell$th row of $g_{0i }$.  The matrix   $\check{\mathfrak{M}}_{2,0}$   denotes $\text{diag}(\check{m}_{21,0}, \ldots \check{m}_{2n,0})$, where $\check{m}_{2i,0} = \mathbb E[\dot{m}_{2i,0}|\mathfrak{G}]$ and  $\mathfrak{G}$ denote the $\sigma$-field generated by $\{x_i ,V_i\}_{i=1} ^n$.    We use $v_{i\ell}$, $\mathfrak{v}_{\psi_0}$ and $\pi_{\psi_0}$  to denote $[V_i]_{\ell}$, $ V \psi_0$ and $ \psi_0 ' \pi_n$ respectively. For a matrix $A$, its maximum eigenvalue and $(i,j)$th element will be respectively denoted by $\lambda_{\max}(A)$ and $[A]_{ij}$. Let $\mathcal E$ denote $\mathbb R^n$ endowed with the inner product $ \langle  \nu_1 ,\nu_2 \rangle_n  = \nu_1 ' \nu_2/n$ for any $\nu_1, \nu_2 \in \mathbb R^n$. The operation $\mathcal T_n: \mathcal H \to \mathcal E$ and its adjoint $\mathcal T_{n} ^\ast  $ are defined by \begin{equation*}
						\mathcal T_n h =  ( 
						\langle Z_1 , h \rangle \  \cdots \  \langle Z_n ,  h\rangle 
						)' \quad\text{and}\quad\mathcal T_n^\ast \nu = n^{-1}\sum_{i}  \nu_i Z_i ,
					\end{equation*}
					for any $h \in \mathcal H$ and $\nu \in \mathbb R^{n}$. Then, for $j=1,\ldots, n$, $\mathcal T_n \widehat \varphi_j = \widehat \kappa_j^{1/2} \widehat \varpi_j$ and $\mathcal T_n^\ast \widehat {\varpi}_j  = \widehat {\kappa}_j ^{1/2}\widehat\varphi_j$. For notational simplicity, we use $\widehat q_j$  to indicate $q(\widehat\kappa_j,\alpha)$ and let $\mathcal P_{n\alpha} = \mathcal T_n \mathcal K_{n\alpha} ^{-1}\mathcal T_n^\ast=\sum_j \hat{q}_j \widehat{\kappa}_j^{-1} (\mathcal T_n \widehat\varphi_j) \otimes_n(\mathcal T_n\widehat{\varphi}_j) = \sum_{j} \hat{q}_{j} \hat{\varpi}_j \otimes_n \hat\varpi_j   $ with the tensor product acting on $\mathcal E$ satisfying that $(\hat\varpi_j \otimes_n\hat \varpi_j )\nu =  (\hat\varpi_j'\nu/n) \hat\varpi_j$ for any $\nu \in \mathbb R^n$.  The following property will be used sometimes to facilitate  discussions: for any $h \in \mathcal H$, \begin{equation}
						\mathcal T_n (\mathcal K_{n\alpha} ^{-1} \mathcal K_n -\mathcal I)h = (\mathcal P_{n\alpha} - \mathcal I)\mathcal T_nh. \label{eq: mse: 0}
					\end{equation}
					Lastly,  let 
					$$\Delta_1 = \sum_{\ell,\ell'=1} ^{2d_e}\frac{  (\mathcal T_n \pi_{\psi_0})' (\mathcal P_{n\alpha}-\mathcal I)  \check{\mathfrak{M}}_{2,0} \mathfrak{g}_{0\ell}\mathfrak{g}_{0\ell'} ' {\check{\mathfrak{M}}}_{2,0} (\mathcal P_{n\alpha} - \mathcal I) \mathcal T_n \pi_{\psi_0}}{n^2} e_\ell e_{\ell'}', $$
					and \begin{equation}\Delta_2 =  \frac{  (\mathcal T_n \pi_{\psi_0})' (\mathcal P_{n\alpha}-\mathcal I) ^2  \mathcal T_n \pi_{\psi_0}}{n } . \end{equation}
					Then, under the conditions in Proposition~\ref{prop: mse}, we have  \begin{equation}
						\mathbb E[\text{tr}(\Delta_1) |x] \leq %\Delta_2  \sum_{\ell=1} ^{2d_e} \mathbb E[  \check{m}_{20,i} ^2 g_{0i,\ell} ^2 |x     ] =
						\Delta_2\sum_{\ell=1} ^{2d_e}n^{-1}\sum_i ( \mathbb E[\dot{m}_{2i,0} ^2 |x]f_{i,\ell} ^2 1_{\{ \ell \leq d_e \}} + \mathbb E[  \dot{m}_{2i,0} ^2 v_{i\ell }  ^2 |x   ]1_{\{\ell >d_e\}}) = O_p(\Delta_2). \label{eq: delta 1}
					\end{equation}

					\begin{lemma}\label{lem.mse} Suppose that the conditions in Proposition~\ref{prop: mse} hold.  Then,  the following holds for $\ell,\ell' = 1,\ldots, 2d_e$.
						\begin{enumerate}[(a)]
							\item\label{lem.mse.1} $\sum_{i} [\mathcal P_{n\alpha}]_{ii} =   O_p(\alpha^{-1/2})$, $\sum_{i} [\mathcal P_{n\alpha} ^2]_{ii} = O_p(\sum_i {[\mathcal P_{n\alpha}]_{ii}})$ and $\sum_{i} [\mathcal P_{n\alpha}]_{ii} ^2  = O_p(\sum_i {P_{ii}})$. 
							\item\label{lem.mse.1.1} $\Delta_2 = O_p(\alpha^2 )$.
							\item\label{lem.mse.2}   $\mathbb E[(\mathfrak{g}_{0\ell}' \check{\mathfrak{M}}_{2,0} \mathcal P_{n\alpha} \mathfrak{v}_{\psi_0})(\mathfrak{g}_{0\ell'}' \check{\mathfrak{M}}_{2,0} \mathcal P_{n\alpha} \mathfrak{v}_{\psi_0}) |x] =O_p(n\alpha^{-1/2})$.
							\item \label{lem.mse.2.1} $\mathbb E[(\mathcal T_n \pi_{\psi_0})' (\mathcal P_{n\alpha} - \mathcal I) \check{\mathfrak{M}}_{2,0}\mathfrak{g}_{0\ell}\mathfrak{g}_{0\ell'}' \check{\mathfrak{M}}_{2,0} \mathcal P_{n\alpha} \mathfrak{v}_{\psi_0}/n^{2}|x] = O_p(\Delta_2 ^{1/2}/(n\alpha^{1/4})). $
							\item \label{lem.mse.2.2} $n^{-2}\mathbb E[\mathfrak{g}_{0\ell}'\mathfrak{m}_{1,0}\mathfrak{m}_{1,0}'\mathfrak{g}_{0\ell}|x]  = O_p(n^{-1})$,  $\mathbb E[\mathfrak{g}_{0,\ell}' \mathfrak{m}_{1,0}\mathfrak{m}_{1,0}'\mathcal P_{n\alpha}Ve_\ell |x] = O_p(\alpha^{-1/2})$ and $\mathbb E[\mathfrak{g}_{0,\ell}' \mathfrak{m}_{1,0}\mathfrak{m}_{1,0}'(\mathcal P_{n\alpha} - \mathcal I) \mathcal T_n \pi_{0\ell} |x] = O_p(n\Delta_2^{1/2})$.
							\item\label{lem.mse.3}  Let $\xi$ be $n\times 1$ vector consisting of  iid random variables $\xi_1,\ldots, \xi_n$ satisfying $\mathbb E[\xi_i |\mathfrak{G}]= 0$, $   E[\xi_i ^2 |\mathfrak{G}]   \leq c_\xi$ for some positive constant $c_\xi$ a.s.n. Then,   $	\mathbb E[(\xi'\mathcal P_{n\alpha}Ve_\ell) (\xi'\mathcal P_{n\alpha} Ve_{\ell'})|x] = O_p(\alpha^{-1/2})$ and $ n^{-1/2}\xi' (\mathcal P_{n\alpha} - \mathcal I)\mathcal T_n \pi_{\ell,0} = O_p(\Delta_2 ^{1/2}) $.
							\item\label{lem.mse.4} For $\xi$ satisfying the conditions in \ref{lem.mse.3}, $\sum_{i,j}\xi_i v_{i\ell} [\mathcal P_{n\alpha}]_{ij} v_{j\ell ' } = O_p(\alpha^{-1/4}) $.
							\item \label{lem.mse.5} $\mathbb E[ \mathfrak{g}_{0\ell} ' \mathfrak{m}_{1,0} \mathfrak{g}_{0\ell} ' \check{\mathfrak{M}}_{2,0}\mathcal P_{n\alpha}\mathfrak{v}_{\psi_0} |x]$, $\mathbb E[ \mathfrak{g}_{0\ell} ' \mathfrak{m}_{1,0} \mathfrak{g}_{0\ell} ' \check{\mathfrak{M}}_{2,0}( \mathcal P_{n\alpha} - \mathcal I)\mathcal T_n \pi_{\psi_0} |x]$, $\mathbb E[\mathfrak{m}_{1,0} ' \mathfrak{g}_{0\ell }\mathfrak{g}_{0\ell} ' \check{\mathfrak{M}}_{2,0} \mathcal P_{n\alpha} \mathfrak{v}_{\psi_0} |x ]$ and $\mathbb E[\mathfrak{m}_{1,0} ' \mathfrak{g}_{0\ell }\mathfrak{g}_{0\ell} ' \check{\mathfrak{M}}_{2,0}( \mathcal P_{n\alpha} - \mathcal I)\mathcal T_n \pi_{\psi_0} |x ]$ are all $0$.
						\end{enumerate}
					\end{lemma}
					\begin{proof}
						\ref{lem.mse.1}:  $\sum_{i} [\mathcal P_{n\alpha}]_{ii} = \sum_{j}  \hat q_j (n^{-1} \hat\varpi_{j}' \widehat{\varpi}_j) =  \sum_{j} \hat q_j $, since $ \hat\varpi_j'\hat\varpi_j/n = 1$. Because $\mathbb E[\Vert Z_i\Vert_{\mathcal H} ^2 ] $ is finite, $\sum_j \widehat{\kappa}_j =O_p(1)$, from which we have $\sum_j \hat{q}_j \leq \max_{j: \widehat{\kappa}_j >0}(\hat{q}_j \hat \kappa_j ^{-1}) \sum_j \hat{\kappa}_j = O_p(\alpha^{-1/2})$ due to Assumption~\ref{asskinv}. The second part is given from  $\mathcal P_{n\alpha} ^2 = \sum_j \hat{q}_j^2 \hat{\varpi}_j \otimes_n \hat{\varpi}_j$ and $\hat{q}_j \in [0,1]$. Lastly, $\sum_{i} [\mathcal P_{n\alpha}]_{ii} ^2 = \sum_{i} e_i' \mathcal P_{n\alpha} e_i e_i' \mathcal P_{n\alpha}e_i \leq    \sum_{i}  [\mathcal P_{n\alpha} ^2]_{ii}\leq \sum_{i}  \hat q_i   $, due to $\lambda_{\max}(e_i e_i') =1 $ and $(\widehat{\varpi}_j'e_i)^2/n \leq 1$.  
						
						\ref{lem.mse.1.1}: The following is deduced from arguments similar to those in \eqref{p3:eq1} and Assumption~\ref{asskinv}. \begin{equation*}
							\Delta_2 = \sum_{\ell=1} ^{d_e} \sum_{j=1} ^n (\hat q_j -1) ^2 \langle \hat \varpi_j, \mathcal T_n \pi_{\ell,0}\psi_{\ell,0} \rangle_n ^2  = \sup_{j} \hat \kappa_j ^{2\rho+1}(\hat q_j -1) ^2 \sum_{\ell=1} ^{d_e}\sum_{j=1} ^n \frac{\psi_{\ell,0} ^2\langle \hat \varphi_j, \pi_{\ell,0} \rangle_{\mathcal H} ^2}{\hat \kappa_j ^{2\rho}}  = O_p(\alpha^2). 
						\end{equation*} 
						
						\ref{lem.mse.2}: $\mathfrak{g}_{0\ell}'   \check{\mathfrak{M}}_{2,0} \mathcal P_{n\alpha} \mathfrak{v}_{\psi_0}=  \sum_{i,j} \check{m}_{2i,0} g_{0i,\ell} [\mathcal P_{n\alpha}]_{ij}V_{j}'\psi_0 = \sum_{j} \check{m}_{2j,0}[\mathcal P_{n\alpha}]_{ jj}g_{0j,\ell}V_j' \psi_0 + \sum_{j}V_j' \psi_0 L_{-j,\ell} = \mathtt{h}_{1,\ell} + \mathtt{h}_{2,\ell }$ with $L_{-j,\ell} $ being $ \sum_{i\neq j}\check{m}_{2i,0}[\mathcal P_{n\alpha}]_{ ij}g_{0i,\ell}  $. We first consider $\mathbb E[\mathtt{h}_{1,\ell}\mathtt{h}_{1,\ell'}|x]$. Note that \begin{align*}
							&\mathbb E[\mathtt{h}_{1,\ell}\mathtt{h}_{1,\ell'}|x]=	\sum_{j,k} \mathbb E[  \check{m}_{2j,0}[\mathcal P_{n\alpha}]_{ jj} g_{0j,\ell}V_j'\psi_0   \check{m}_{2k,0}[\mathcal P_{n\alpha}]_{kk} g_{0k,\ell'}V_k'\psi_0 |x]  \nonumber\\
							&	= \begin{cases}
								\sum_{j,k} f_{j,\ell}f_{k,\ell'}  [\mathcal P_{n\alpha}]_{jj}[\mathcal P_{n\alpha}]_{kk} \mathbb E[ (\check{m}_{2j,0} V_j'\psi_0)(\check{m}_{2k,0} V_k'\psi_0) |x] &\text{if }1\leq \ell,\ell' \leq d_e,\\
								\sum_{j,k} f_{j,\ell}  [\mathcal P_{n\alpha}]_{jj}[\mathcal P_{n\alpha}]_{kk} \mathbb E[ (\check{m}_{2k,0} v_{k\ell'-d_e} V_k'\psi_0)(\check{m}_{2j,0} V_j'\psi_0) |x]  &\text{if }\ell \leq d_e < \ell' \leq 2d_e, \\
								\sum_{j,k}[\mathcal P_{n\alpha}]_{jj}[\mathcal P_{n\alpha}]_{kk}  \mathbb E[ (\check{m}_{2j,0} v_{j\ell-d_e}V_j'\psi_0)(\check{m}_{2k,0} v_{k\ell-d_e}V_k'\psi_0) |x]    &\text{if }2d_e< \ell,\ell' \leq n,
							\end{cases} 
						\end{align*}
						where the terms in the RHS are all bounded above by $O(1)\sum_{j,k}   [\mathcal P_{n\alpha}]_{jj}[\mathcal P_{n\alpha}]_{kk}   = O_p(\alpha^{-1})$. We then focus on $\mathbb E[\mathtt{h}_{2,\ell} \mathtt{h}_{2,\ell'} |x]$ that can be decomposed into two terms satisfying the following.\\\begin{enumerate*}[\ref*{lem.mse.2}.i]
							\item\label{lem: mse: tmp: 1} $\sum_{j} \mathbb E[  (V_{j}'\psi_0) ^2 L_{-j,\ell} L_{-j,\ell'} |x  ] = O_p(n\alpha^{-1/2})$,
							\item\label{lem: mse: tmp: 2} $ \sum_{j\neq k} \mathbb E[ (V_j'\psi_0)(V_k'\psi_0)L_{-j,\ell}L_{-k,\ell'} |x  ]  = O_p(\alpha^{-1})$.
						\end{enumerate*} 
						
						\textbf{Proof of \ref{lem: mse: tmp: 1}:} Let $\mathtt{h}_{21,\ell\ell'} =\sum_j\sum_{i \neq j} [\mathcal P_{n\alpha}]_{ij}[\mathcal P_{n\alpha}]_{ji} \mathbb E[ \check{m}_{2i,0}^2  g_{0i,\ell}g_{0i,\ell'} |x  ] $ and $\mathtt{h}_{22,\ell\ell'} = \sum_j\sum_{i\neq j}\sum_{s\neq i , j}$ $[\mathcal P_{n\alpha}]_{ij}[\mathcal P_{n\alpha}]_{js} \mathbb E[ \check{m}_{2i,0}  g_{0i,\ell}|x]\mathbb E[ \check{m}_{2s,0}g_{0s,\ell' } |x  ]$. Then, we have  \begin{equation}
							\sum_{j} \mathbb E[  (V_{j}'\psi_0) ^2 L_{-j,\ell} L_{-j,\ell'} |x  ]% = \sigma_{v_{\psi_0}} ^2 \sum_j\sum_{i,s \neq j}[\mathcal P_{n\alpha}]_{ij}[\mathcal P_{n\alpha}]_{js} \mathbb E[ \check{m}_{20,i}\check{m}_{20,s}  g_{0i,\ell} g_{0s,\ell' } |x  ] \nonumber\\
							=\sigma_{{\psi_0}} ^2   (\mathtt{h}_{21,\ell\ell'} + \mathtt{h}_{22,\ell\ell'}). \label{eq: mse: 1}
							% &= \sigma_{v_{\psi_0}} ^2  (\underbrace{\sum_j\sum_{i \neq j} [\mathcal P_{n\alpha}]_{ij}[\mathcal P_{n\alpha}]_{ji} \mathbb E[ \check{m}_{20,i}^2  g_{0i,\ell}g_{0i,\ell'} |x  ] }_{ = \mathtt{h}_{21,\ell\ell' }}+   \underbrace{\sum_j\sum_{i\neq j}\sum_{s\neq i , j}[\mathcal P_{n\alpha}]_{ij}[\mathcal P_{n\alpha}]_{js} \mathbb E[ \check{m}_{20,i}  g_{0i,\ell}|x]\mathbb E[ \check{m}_{20,s}g_{0s,\ell' } |x  ]}_{=\mathtt{h}_{22,\ell\ell' }}   ),
						\end{equation}  
						Because of the boundedness of $\mathbb E[\Vert g_{0i} \Vert^2 |x]$ and $\check m_{2i,0}$, the term $\mathtt{h}_{21,\ell\ell'} $ is bounded as follows. \begin{equation}
							\mathtt{h}_{21,\ell\ell'}  \leq O_p \left( \sum_{j}(\sum_{i\neq j} [\mathcal P_{n\alpha}]_{ji}^2 )^{1/2}(\sum_{i\neq j}  [\mathcal P_{n\alpha}]_{ij} ^2) ^{1/2} \right)  \leq O_p (\sum_j  [\mathcal P_{n\alpha} ^2]_{jj} )  ,\label{eq: mse: 2}
						\end{equation}
						where the inequality follows from H\"older's inequality.
						%where the second inequality follows from Assumption \ref{ass9:finite} and $\sum_{i\neq j} [\mathcal P_{n\alpha}]_{ji} ^2 \leq \sum_{i} [\mathcal P_{n\alpha}]_{ji}^2  = [\mathcal P_{n\alpha} ^2]_{jj}$. Similarly, if $1\leq \ell \leq d_e< \ell'  \leq 2d_e $,  \begin{equation}
							%	\mathtt{h}_{21,\ell\ell'} \leq \mathbb E[\check{m}_{2i,0} ^2 v_{i\ell' -d_e} |x]  \sum_j \sum_{i\neq j} [\mathcal P_{n\alpha}]_{ij}[\mathcal P_{n\alpha}]_{ji} f_{i\ell}\leq  O_p( \sum_j   [\mathcal P_{n\alpha}^2  ]_{jj})  = O_p (\alpha^{-1/2}).\label{eq: mse: 3}
							%\end{equation}
							%astly, if $ d_e <\ell,\ell' \leq   2d_e$, $\mathtt{h}_{21,\ell\ell'}$ satisfies that \begin{equation}
								%	\mathtt{h}_{21,\ell\ell'} =  \mathbb E[\check{m}_{2i,0} ^2 v_{i\ell-d_e}v_{i\ell'-d_e}|x]  \sum_{j}\sum_{i\neq j}[\mathcal P_{n\alpha}]_{ij}[\mathcal P_{n\alpha}]_{ji}  \leq O_p (  \sum_{i}   [ \mathcal P_{n\alpha} ^2]_{ii} ) . \label{eq: mse: 4}
								%\end{equation}
								%From \eqref{eq: mse: 2} to \eqref{eq: mse: 4},  $\mathtt{h}_{21,\ell\ell'} = O(\alpha^{-1/2})$.
								Moreover,  $\mathtt{h}_{22,\ell\ell'}$ in \eqref{eq: mse: 1} is given as follows.
								\begin{equation}
									\mathtt{h}_{22,\ell\ell' } = \begin{cases}
										\sum_{j} \sum_{i \neq j} \sum_{s \neq i , j} \mathbb E[\check{m}_{2i,0} |x]  \mathbb E[\check{m}_{2s,0} |x]f_{i\ell}f_{s\ell'}[\mathcal P_{n\alpha}]_{ij}[\mathcal P_{n\alpha}]_{js} &\text{if }1\leq \ell, \ell' \leq d_e,\\
										\sum_{j} \sum_{i\neq j} \sum_{s \neq i,j}	\mathbb E[\check{m}_{2i,0} |x] \mathbb E[  \check{m}_{2s,0}  v_{s\ell'-d_e } |x   ]  f_{i\ell} [\mathcal P_{n\alpha}]_{ij}[\mathcal P_{n\alpha}]_{js} &\text{if }1\leq \ell\leq d_e<\ell'\leq 2d_e,\\
										\sum_{j} \sum_{i \neq j} \sum_{s \neq i , j}	\mathbb E[  \check{m}_{2i,0}   v_{i\ell-d_e }   |x   ]\mathbb E[  \check{m}_{2s,0}    v_{s\ell' -d_e} |x   ] [\mathcal P_{n\alpha}]_{ij}[\mathcal P_{n\alpha}]_{js} &\text{if }2d_e \leq \ell,\ell' \leq n 
									\end{cases} \nonumber
								\end{equation} 
								Note that  for any sequence $a_j$, $\tilde{a}_j$ and $b_{ij}$ satisfying $\max |a_j|, \max|\tilde{a}_j| \leq c_a$,  $\sum_j (\sum_{i\neq j}a_i b_{ij})(  \sum_{i\neq j}\tilde{a}_i {b}_{ij}) \leq   c_a ^2n\sum_j \sum_{i}b_{ij} ^2$ and $\sum_j \sum_{i\neq j} a_i\tilde{a}_i b_{ij} ^2 \leq c_a^2 \sum_j \sum_i b_{ij} ^2  $. Because $\max_{1\leq j \leq n}\max_{1\leq \ell\leq d_e } 	\mathbb E[  |\check{m}_{2j,0}   g_{0j,\ell}|   |x   ]|  \leq \lambda_{\max}(\mathbb E[\check{{m}}_{2j,0}^2 g_{0j} g_{0j} ' |x ])  = \widetilde{\mathtt{c}}_g ^{1/2}$, we have\begin{align}
									\mathtt{h}_{22,\ell\ell' }  \leq \widetilde{ \mathtt{c}}_g  ( n\sum_{j}   \sum_{i  }[\mathcal P_{n\alpha}]_{ij}^2 + \sum_{i,j} [\mathcal P_{n\alpha}]_{ij} ^2   )   =   O_p(n\alpha^{-1/2}).\label{eq: mse: 6} 
								\end{align}  
								%We then focus on the second line in \eqref{eq: mse: 5}. Because $|a-b| \leq |a| + |b|$ and $\sum_i f_{i\ell} ([\mathcal P_{n\alpha} ^2]_{ii} - [\mathcal P_{n\alpha}]_{ii}) = O_p(\alpha^{-1/2})$, we have\begin{align}
									%	|\sum_{j}\sum_{i\neq j} \sum_{s \neq i,j} f_{i\ell} [\mathcal P_{n\alpha}]_{ij}[\mathcal P_{n\alpha}]_{js}| &\leq | \sum_{j}\sum_{i\neq j}f_{i\ell} [\mathcal P_{n\alpha}]_{ij}\sum_{s\neq j} [\mathcal P_{n\alpha}]_{js} | + O_p(\alpha^{-1/2}) \nonumber\\
									%	&\leq n\sum_{j} (\sum_{i\neq j} f_{i\ell} ^2 [\mathcal P_{n\alpha}]_{ij} ^2 )^{1/2} (\sum_{s\neq j} [\mathcal P_{n\alpha}]_{js} ^2) ^{1/2} +O_p(\alpha^{-1/2}) \nonumber\\
									%	&\leq O_p(n) \sum_{i,j} [\mathcal P_{n\alpha}]_{js} ^2 +O_p(\alpha^{-1/2}) = O_p(n\alpha^{-1/2}), \label{eq: mse: 7} 
									%\end{align}
									%where the last two inequalities follow from H\"older's inequality and the boundedness of $f_i$. From \eqref{eq: mse: 6} to \eqref{eq: mse: 7}, we find that $
									%\mathtt{h}_{22,\ell\ell'} = O_p(n\alpha^{-1/2})$ for all $\ell,\ell'$, and thus, 
									By combining with the bound of $\mathtt{h}_{21,\ell\ell'}$, \eqref{eq: mse: 1} is bounded above by $O_p(n\alpha^{-1/2})$.
									
									\textbf{Proof of \ref{lem: mse: tmp: 2}:} Let $L_{-ij,\ell\ell'}$ being defined without $i$th and $j$th elements. Because $\mathbb E[V_j |x] = 0$, \begin{equation*}
										\sum_j\sum_{s\neq j} \mathbb E[V_{j}'\psi_0V_{s}'\psi_0 L_{-js,\ell}L_{-sj,\ell'} |x ] = 0   \quad\text{and}\quad\sum_j\sum_{s\neq j}  \mathbb E[ V_{j}'\psi_0V_{s}'\psi_0  L_{-sj,\ell'} (L_{-j,\ell}- L_{-js,\ell}) |x] = 0.
									\end{equation*} From the above, we find that  \begin{align*}
										&\sum_j\sum_{s\neq j} \mathbb E[  V_{j}'\psi_0V_s'\psi_0 L_{-j,\ell}L_{-s,\ell'} |x  ]\\
										&	=\begin{cases}
											\sum_j\sum_{s\neq j}  f_{j,\ell} f_{s,\ell'} (\mathbb E[\check{m}_{2j,0} V_j '\psi_0 |x])^2[\mathcal P_{n\alpha}]_{jj}[\mathcal P_{n\alpha}]_{ss} &\text{if }1\leq \ell,\ell' \leq d_e, \\
											\sum_j\sum_{s\neq j} f_{j,\ell} \mathbb E[\check{m}_{2j,0} V_j '\psi_0 |x] \mathbb E[V_s'\psi_0 v_{s\ell'-d_e}\check{m}_{2s,0}|x] [\mathcal P_{n\alpha}]_{jj}[\mathcal P_{n\alpha}]_{ss}  &\text{if }1\leq \ell\leq  d_e<\ell'\leq 2d_e ,\\
											\sum_j\sum_{s\neq j}     \mathbb E[V_j'\psi_0 v_{j\ell- d_e}\check{m}_{2j,0}|x] \mathbb E[V_j'\psi_0 v_{j\ell ' -d_e}\check{m}_{2j,0}|x] [\mathcal P_{n\alpha}]_{jj}[\mathcal P_{n\alpha}]_{ss} &\text{if }2d_e<\ell,\ell'\leq n.
										\end{cases}
									\end{align*}
									Because of the boundednesses of $\check{m}_{2j,0}$  and $\mathbb E[\Vert g_i\Vert ^2 |x]$ and the fact $[\mathcal P_{n\alpha}]_{jj} \geq 0$, the above is bounded above by $O_p(\sum_{j} \sum_{s\neq j}[\mathcal P_{n\alpha}]_{jj}[\mathcal P_{n\alpha}]_{ss} )  = O_p(\alpha^{-1})$, from which the desired result is obtained.  
									
									\ref{lem.mse.2.1}: Let $\gamma_{i,\psi_0}=[\mathcal T_n \pi_{\psi_0}]_i$ and $\mathcal Q_{n\alpha}=\mathcal P_{n\alpha} -\mathcal I$. Due to Assumption~\ref{ass3:finite},  \begin{align*}
										&\mathbb E[(\mathcal T_n \pi_{\psi_0}) ' (\mathcal P_{n\alpha} - \mathcal I)\check{\mathfrak{M}}_{2,0}\mathfrak{g}_{0\ell} \mathfrak{g}_{0\ell'} ' \check{\mathfrak{M}}_{2,0}\mathcal P_{n\alpha} \mathfrak{v}_{\psi_0} |x] \nonumber\\
										&= \sum_{i,j} \gamma_{i,\psi_0} [\mathcal Q_{n\alpha}]_{ij} [\mathcal P_{n\alpha}]_{jj}\mathtt{t}_{1j,\ell\ell'}+ 2 \sum_{i,j}\sum_{k \neq j}  \gamma_{i,\psi_0}  [\mathcal Q_{n\alpha}]_{ij}[\mathcal P_{n\alpha}]_{jk} \mathtt{t}_{2j,\ell}\mathtt{t}_{3k,\ell'}, 
									\end{align*} 
									where $\mathtt{t}_{1 j ,\ell\ell'} = \mathbb E[\check{m}_{2j,0} ^2 g_{0j,\ell} g_{0j,\ell'} V_{j} ' \psi_0 |x  ]$, $\mathtt{t}_{2 j ,\ell} =\mathbb E[\check{m}_{2j,0} g_{0j,\ell} V_j'\psi_0 |x] $ and $\mathtt{t}_{3 j ,\ell'} = \mathbb E[\check{m}_{2j,0} g_{0j,\ell'} |x]$. The quantities are bounded for all $\ell$ and $\ell'$ under the employed assumptions. Hence, we have  \begin{align*}
										\sum_{i,j}  \gamma_{i,\psi_0} [\mathcal Q_{n\alpha}]_{ij} [\mathcal P_{n\alpha}]_{jj}\mathtt{t}_{1j,\ell\ell'}  &\leq (\sum_{j} [\mathcal P_{n\alpha}]_{jj}^2 \mathtt{t}_{1j,\ell\ell'}^2 ) ^{1/2}(\sum_{j}\sum_{i,\ell}\gamma_{i,\psi_0}\gamma_{\ell,\psi_0} [\mathcal Q_{n\alpha}]_{ij}[\mathcal Q_{n\alpha}]_{\ell j} )^{1/2} \nonumber\\
										& \leq O_p(\alpha^{-1/4}) (  \sum_{i,\ell}\gamma_{i,\psi_0}\gamma_{\ell,\psi_0} [\mathcal Q_{n\alpha}^2]_{i\ell}  )^{1/2} = O_p(\alpha^{-1/4}n^{1/2}\Delta_2 ^{1/2}),
									\end{align*}
									where the last bound follows from Lemma \ref{lem.mse}.\ref{lem.mse.1} and   the definition of $\Delta_2$. Similarly, we have \begin{align*}
										\sum_{i,j}\sum_{k\neq j} \gamma_{i,\psi_0} [\mathcal Q_{n\alpha}]_{ij} [\mathcal P_{n\alpha}]_{jk} \mathtt{t}_{2j,\ell} \mathtt{t}_{3k,\ell'}&\leq (\sum_{j} (\sum_{k\neq j} [\mathcal P_{n\alpha}]_{jk} \mathtt{t}_{2j,\ell} \mathtt{t}_{3k,\ell'})^2 )^{1/2} (\sum_{j} (\sum_i \gamma_{i,\psi_0} [\mathcal Q_{n\alpha}]_{ij})^2 )^{1/2} \nonumber\\
										&\leq (\sum_j \sum_{k\neq j} [\mathcal P_{n\alpha}]_{jk} ^2 \mathtt{t}_{2j,\ell}^2 \sum_{k\neq j} \mathtt{t}_{3k,\ell'} ^2   )^{1/2} (n \Delta_2 )^{1/2} \nonumber\\
										&\leq (O(n) \sum_{j}\sum_k [\mathcal P_{n\alpha}]_{jk}^2 )^{1/2}(n \Delta_2 )^{1/2} \leq O_p(\alpha^{-1/4} n\Delta_2 ^{1/2}).
									\end{align*} 
									
									\ref{lem.mse.2.2}: The first part follows from the facts that $\mathbb E[\mathfrak{m}_{1,0}\mathfrak{m}_{1,0}' |\mathfrak{G}] = \text{diag}(\mathbb E[ m_{11,0} ^2  |\mathfrak{G}   ], \ldots, \mathbb E[ m_{1n,0} ^2  |\mathfrak{G}   ])$ and $m_{1i,0} ^2$ is  bounded for all $i$. The second is deduced from  the independence across $i$; specifically, \begin{equation*}
										\mathbb E[\mathfrak{g}_{0\ell}' \mathfrak{m}_{1,0}\mathfrak{m}_{1,0}'\mathcal P_{n\alpha}V e_\ell |x] %= \mathbb E[\sum_{i,j,k}  g_{0i,\ell} \mathbb E[m_{1i,0} m_{1j,0}  |\mathfrak{G}] P_{jk} V_{k\ell} |x  ] 
										= \sum_{i}\mathbb E [\mathbb E[ m_{1i,0} ^2 |\mathfrak{G} ] V_{i\ell} g_{0i,\ell} |x][\mathcal P_{n\alpha}]_{ii} = O_p (\alpha ^{-1/2}).
									\end{equation*} 
									To obtain the last part, let $\mathtt{t}_{4,\ell}$ be the $n\times 1$ vector consisting of $\{E[g_{0i,\ell} \mathbb E[m_{1i,0}^2 |\mathfrak{G}] |x]\}_{i=1} ^n$. Then,  \begin{align*}
										&\mathbb E[\mathfrak{g}_{0\ell}' \mathfrak{m}_{1,0}\mathfrak{m}_{1,0}'(\mathcal P_{n\alpha} - \mathcal I)\mathcal T_n \pi_{\ell,0} |x]   = \sum_{i,j,k} \mathbb E[g_{0i,\ell} \mathbb E[m_{1i,0} m_{1j,0} |\mathfrak{G}] |x] [\mathcal Q_{n\alpha}]_{jk}\langle Z_{k}, \pi_{\ell,0} \rangle \nonumber\\
										&= \mathtt{t}_{4,\ell} ' \mathcal Q_{n\alpha} \mathcal T_n \pi_{\ell,0} \leq (  \mathtt{t}_{4,\ell} ' \mathtt{t}_{4,\ell}  ) ^{1/2} ( (\mathcal T_n \pi_{\ell,0})' (\mathcal P_{n\alpha}-\mathcal I)^2 \mathcal T_n\pi_{\ell,0})^{1/2}  = O_p(n\Delta_2 ^{1/2}).
									\end{align*}

									\ref{lem.mse.3}: The first part follows from the Markov's inequality and the fact that $\mathbb E[\xi ' \mathcal P_{n\alpha} {v}_{i\ell}|x] = 0$ and   \begin{equation*}
										\mathbb E[(\xi'\mathcal P_{n\alpha} Ve_\ell)(\xi'\mathcal P_{n\alpha} Ve_{\ell'})|x] = \mathbb E[  e_{\ell} 'V' \mathcal P_{n\alpha} \mathbb E[\xi \xi' |x] \mathcal P_{n\alpha} Ve_{\ell'} | x   ] \leq c_\xi\text{tr}(\mathcal P_{n\alpha}\mathbb E[Ve_\ell e_{\ell'}'V'    |x  ]  \mathcal P_{n\alpha} ) = O_p(\alpha ^{-1/2}),
									\end{equation*}
									for $\ell,\ell' = 1,\ldots, d_e$. Analogously,  the second part is obtained from  $\mathbb E[\xi' (\mathcal P_{n\alpha} - \mathcal I)\mathcal T_n \pi_{\ell,0} |x] = 0$ and \begin{equation*}
										\mathbb E[   n^{-1}(\xi' (\mathcal P_{n\alpha} - \mathcal I)\mathcal T_n \pi_{\ell,0})^2 |x   ] = n^{-1}(\mathcal T_n \pi_{\ell,0}) ' (\mathcal P_{n\alpha} - \mathcal I) \mathbb E[\xi \xi' |x] (\mathcal P_{n\alpha} - \mathcal I) \mathcal T_n \pi_{\ell,0} \leq c_\xi \Delta_2,
									\end{equation*}
									since $\mathbb E[\xi\xi'|x] =\text{diag}(\mathbb E[\xi_1 ^2 |x ],\ldots, \mathbb E[\xi_n ^2 |x ])$ which is bounded above by $c_\xi\mathcal I_n$.

									\ref{lem.mse.4}: Note that $\mathbb E[\sum_{i,j} \xi_ i v_{i\ell} [\mathcal P_{n\alpha}]_{ij} v_{j\ell ' } | x] = 0$ for all $\ell$ and $\ell'$. Let $\widetilde{L}_{-i,\ell'} = \sum_{j \neq i} [\mathcal P_{n\alpha}]_{ij} v_{j\ell ' }$ and decompose $\sum_{i,j} \xi_ i v_{i\ell} [\mathcal P_{n\alpha}]_{ij} v_{j\ell ' }$ into $\sum_{i} \xi_i v_{i\ell} v_{i\ell'} [\mathcal P_{n\alpha}]_{ii}$ and $\sum_{i} \xi_i v_{i\ell} \widetilde{L}_{-i,\ell'}$. The following holds a.s.\begin{equation*}
										\mathbb E[ (\sum_{i} \xi_i v_{i\ell} v_{i\ell'} [\mathcal P_{n\alpha}]_{ii})^2 |x  ] = \sum_{i} \mathbb E[  \xi_i ^2 v_{i\ell} ^2 v_{i\ell'} ^2|x  ] [\mathcal P_{n\alpha}]_{ii} ^2 \leq c_{\xi} \mathbb E[\Vert V_i\Vert ^4 |x] \sum_{i} [\mathcal P_{n\alpha}]_{ii} ^2  = O_p(\alpha^{-1/2}).
									\end{equation*} 
									Since $\sum_{i}\mathbb E[\xi_i ^2 v_{i\ell} ^2 \widetilde{L}_{-i,\ell'} ^2 |x ]  \leq  c_\xi \sum_{i} \mathbb E[v_{i\ell} ^2|x]\sum_{j \neq i}\mathbb E[v_{j\ell'} ^2|x][\mathcal P_{n\alpha}]_{ji} ^2= O_p(\sum_{i} [\mathcal P_{n\alpha}^2 ]_{ii})  $ and $\mathbb E[\xi_i v_{i\ell} ^2|x] = 0$,\begin{align*}
										\mathbb E[ (\sum_{i} \xi_i v_{i\ell} \widetilde{L}_{-i,\ell'})^2 |x  ] &= \mathbb E[\sum_{i}\sum_{j\neq i} \xi_i \xi_j v_{i\ell}v_{j\ell}\widetilde{L}_{-i,\ell'}\widetilde{L}_{-j,-\ell'}|x]  + O_p(\alpha ^{-1/2})\nonumber\\
										&= \sum_{i} \sum_{j\neq i} \mathbb E[\xi_i v_{i\ell} ^2 [\mathcal P_{n\alpha}]_{ii}|x] \mathbb E[ \xi_j v_{j\ell} ^2 [\mathcal P_{n\alpha}]_{jj}|x] + O_p(\alpha^{-1/2}) = O_p(\alpha^{-1/2}).
									\end{align*} 
									Hence, the desired result is given by applying the Markov's inequality.
									
									\ref{lem.mse.5} is obtained from $\mathbb E[m_{1i,0}|\mathfrak{G}] = 0$ and the law of iterated expectations.
								\end{proof}
								\textbf{Proof of Proposition~\ref{prop: mse}} Let $\overline\theta$ be the solution to the following linearization problem: \begin{equation}
									\mathcal M_n(\theta_0 , g_i) + \Gamma_{1n} (\theta_0, g_i) (\overline\theta - \theta_0) + \Gamma_{2n} (\theta_0, g_i) (\hat{\pi}_n - \pi_n) = 0. \label{thetaoverline: def}
								\end{equation}  Note that  $ \mathcal S_n ^{\prime}(\overline{\theta} - \theta_0) =-\widehat{\mathtt{H}}^{-1}\hat{\mathtt{h}} = -({\mathtt{H}} + \mathtt{T}_1^{H} + \mathtt{T}_2 ^{H} )^{-1} (\mathtt{h} + \mathtt{T} _1 ^h +\mathtt{T} _2 ^h + \mathtt{Z}_{h}),$
								where $\widetilde{\mathtt{H}} = n^{-1}\sum_{i=1} ^n\mathbb E[\dot{m}_{2i,0}   g_{0i} g_{0i}'|x]$, $ \mathtt{Z}_h =  n^{-1}  \sum_{i=1} ^n   \psi_0' (\widehat{\Pi}_{n,\alpha} - \Pi_n)Z_i (\dot{m}_{2i,0} - \check{m}_{2i,0}+ \ddot{m}_{2i,0} ) g_{0i} $  \begin{align*}
									&	  {\mathtt{H}} = n^{-1}\sum_{i=1} ^n\check{m}_{2i,0}   g_{0i} g_{0i}' ,\quad  \mathtt{T}_1^{H} =n^{-1}\sum_{i=1} ^n (\dot{m}_{2i,0} - \check{m}_{2i,0}  )g_{0i} g_{0i}',\quad
									\mathtt{T}_2^{H} =n^{-1}\sum_{i=1} ^n \ddot{m}_{2j,0} g_{0i} g_{0i}' ,\quad\\
									&	   \mathtt{h} = {n^{-1}\sum_{i=1} ^n   \psi_0 '(\widehat{\Pi}_{n,\alpha} - \Pi_n) Z_i  \check{m}_{2i,0} g_{0i}  } ,\quad \mathtt{T}_1 ^h= n^{-1}\sum_{i=1} ^n m_{1i,0} g_{0i}  , \quad\mathtt{T}_2  ^h =  {n^{-1}   \widetilde{\mathcal S}_{n} ^{-1} (\widehat{\Pi}_{n,\alpha} - \Pi_n) \sum_{i=1} ^n  {m}_{1i,0}  Z_i   }  .
								\end{align*}

								For  $\ell=1,\ldots, d_e$,  the $\ell$th row of $\mathtt{Z}_h$ is equal to  $
								n^{-1} \xi' (\mathcal P_{n\alpha} - \mathcal I)\mathcal T_n  \pi_{\psi_0}   + n^{-1}  \xi'  \mathcal P_{n\alpha} \mathfrak{v}_{\psi_0}  ,$
								with $\xi$ being the $n \times 1$ vector whose $i$th element is given by  $  (\dot{m}_{2i,0} - \check{m}_{2i,0}+ \ddot{m}_{2i,0} ) {g}_{0i,\ell}$. Under Assumption~\ref{ass9:finite}, $\xi$ satisfies the conditions in Lemma \ref{lem.mse}.\ref{lem.mse.3}. Therefore, by  Lemma \ref{lem.mse}.\ref{lem.mse.3}, the first $d_e$ rows of $\mathtt{Z}_h $ is $o_p(\varrho_{n\alpha})$, where  $\varrho_{n\alpha}= \Delta_2 + O_p(n^{-1}\alpha^{-1/2})$. For the remaining rows, the same result is obtained by applying  Lemma \ref{lem.mse}.\ref{lem.mse.4} and the second part of Lemma \ref{lem.mse}.\ref{lem.mse.3}. In addition,  if $\ell=1,\ldots,d_e$, the $\ell$th row of  $\mathtt{T}_2 ^h$ is equal to  $\mathtt{z}_{1\ell} ^h+\mathtt{t}_\ell ^h$, while it is equal to $\mathtt{z}_{2\ell} ^h + \mathtt{z}_{3\ell} ^h$ if $ \ell =d_e+1,\ldots, 2d_e$, where \begin{equation*}
									\mathtt{z}_{1\ell} ^h =  {\mathfrak{m}}_{1,0} '    (\mathcal P_{n\alpha} - \mathcal I)\mathcal T_n \pi_{\ell,0}/n ,\quad \mathtt{z}_{2\ell} ^h=     {\mathfrak{m}}_{1,0} '    (\mathcal P_{n\alpha} - \mathcal I)\mathcal T_n \pi_{\ell,n}/n, \quad \mathtt{z}_{3\ell} ^h =    \mathfrak{m}_{1,0} ' \mathcal P_{n\alpha} Ve_\ell/n ,
								\end{equation*}  
								and $\mathtt{t}_{\ell} ^h =   \mathfrak{m}_{1,0} ' \mathcal P_{n\alpha} V\widetilde{\Lambda}e_\ell /(\sqrt{n}\mu_{\ell n} )$.
								From Lemma \ref{lem.mse}.\ref{lem.mse.3},  $\mathtt{z}_{1\ell} ^h$, $\mathtt{z}_{2\ell} ^h$ and $\mathtt{z}_{3\ell} ^h$ are  $O_p( (\Delta_2/n) ^{1/2})$, $O_p( (\Delta_2/n) ^{1/2})$ and $O_p( n^{-1}\alpha^{-1/4})$ which are all $o_p(\varrho_{n\alpha})$. Hence, $\mathtt{T}_2 ^h =  \sum_{\ell=1} ^{d_e}\mathtt{t}_\ell ^h e_\ell   + o_p(\varrho_{n\alpha})$. Lastly,  note that $\mathtt{h}$ can be written by $
								n^{-1} G_0 ' \check{\mathfrak{M}}_{2,0} \mathcal T_n (\mathcal K_{n\alpha} ^{-1}\mathcal K_n - \mathcal I) \pi_{\psi_0} + n^{-1} G_0 ' \check{\mathfrak{M}}_{2,0} \mathcal T_n \mathcal K_{n\alpha} ^{-1} \mathcal T_n ^\ast \mathfrak{v}_{\psi_0},
								$ where $G_0 = \sum_{\ell} \mathfrak{g}_{0,\ell} e_{\ell}'$. Then, from Lemma~\ref{lem.mse} and \eqref{eq: mse: 0}, the following holds.\begin{enumerate}[(a)]
									\item\label{lem.res.1} $ \mathbb E[\mathtt{T}_1 ^{h\prime}   \mathtt{T}_1 ^{h}  |x] = O_p(n^{-1})$.  
									\item\label{lem.res.2} $\max_{1\leq\ell\leq d_e} \mathbb E[\mathtt{T}_1 ^{h} \mathtt{t}_\ell ^{h} |x ] = O_p(n^{-3/2}\mu_{m,n}^{-1} \alpha^{-1/2} )$. 
									\item\label{lem.res.3} $ \mathbb E[\mathtt{T}_1 ^{h} \mathtt{h}  ' |x  ]  = 0$ and $\mathbb E[\mathtt{t}_\ell ^h \mathtt{h} |x ] = 0$.  
									\item\label{lem.res.4} $ \sum_{\ell,\ell' = 1} ^{d_e}  \mathbb E[\mathtt{t}_\ell ^h \mathtt{t}_{\ell'} ^h |x] e_\ell e_{\ell ' }  = O_p( n^{-1}\mu_{m,n} ^{-2} \alpha ^{-1/2}  ) = o_p(\varrho_{n\alpha})$.  
									\item\label{lem.res.5} $\Vert \mathtt{T}_{1} ^H\Vert ^2$ and $\Vert \mathtt{T}_2 ^H\Vert^2 $ are all $o_p(\varrho_{n\alpha})$. Moreover, for   $\ell = 1,\ldots, d_e$ and $k=1,2$, $\Vert \mathtt{T}_k ^H\Vert \Vert \mathtt{t}_\ell  ^h\Vert=O_p(n^{-1}\mu_{m,n}^{-1} \alpha ^{-1/4} ) $ and $\Vert \mathtt{T}_k ^H\Vert \Vert \mathtt{T}_1 ^h\Vert =O_p(n^{-1})$ for $k=1,2$.  
									\item\label{lem.res.6} $\mathbb E[\mathtt{h} \mathtt{h} ' \mathtt{H} ^{-1} \mathtt{T}_{k} ^{H} |x  ] = \mathbb E[\mathtt{h} \mathtt{h} ' \mathtt{H} ^{-1} \mathbb E[ \mathtt{T}_{k} ^{H} |\mathfrak{G}] |x  ] = 0$ for $k=1,2$.
								\end{enumerate}
								Let $\hat{\mathtt{A}}(\alpha) = (\mathtt{h} + \mathtt{T}_1 ^h +  \sum_{\ell=1} ^{d_e } \mathtt{t}_\ell ^h e_\ell  )(\mathtt{h} + \mathtt{T}_1 ^h  +  \sum_{\ell=1} ^{d_e } \mathtt{t}_\ell ^h e_\ell  )'  - \sum_{k=1} ^2( \mathtt{h} \mathtt{h} ' \mathtt{H} ^{-1} \mathtt{T}_{k} ^{H} + \mathtt{T}_k ^h \mathtt{H} ^{-1} \mathtt{h}\mathtt{h}').$ Then, from \ref{lem.res.1} to \ref{lem.res.6}, we find that \begin{align}
									\mathbb E[\hat{\mathtt{A}}(\alpha) |x] &= \mathbb E[ \mathtt{h}\mathtt{h}' |x ] +o_p(\varrho_{n\alpha})\nonumber\\
									%	&=  n^{-2}\sigma_{v_{\psi_0}} ^2 \sum_{\ell,\ell' = 1} ^{2d_e}( \mathtt{h}_{21,\ell\ell'} + \mathtt{h}_{22,\ell\ell'})e_{\ell}e_{\ell'}' + \Delta_1 + O_p(\Delta_2 ^{1/2}/(n\alpha ^{-1/4})) +o_p(\varrho_{n\alpha})\nonumber\\ 
									& =\Delta_1+n^{-2}	\sum_{\ell , \ell' = 1} ^{2d_e} \mathbb E[ \mathfrak{v}_{\psi_0} ' \mathcal P_{n\alpha} \check{\mathfrak{M}}_{2,0} \mathfrak{g}_{0\ell} \mathfrak{g}_{0\ell'} '\check{\mathfrak{M}}_{2,0}  \mathcal P_{n\alpha} \mathfrak{v}_{\psi_0} |x]   e_\ell e_{\ell'} '  +  o_p(\varrho_{n\alpha})   \label{eq: ahat}\\ 
									& =\Delta_1 + O_p(n^{-1}\alpha ^{-1/2}) + o_p(\varrho_{n\alpha}), \label{eq: ahat1}
								\end{align}
								where the second line follows from the definition of $\Delta_1$ and Lemmas \ref{lem.mse}.\ref{lem.mse.2.1} and \ref{lem.mse}.\ref{lem.mse.4}, and the last line is obtained from Lemmas~\ref{lem.mse}.\ref{lem.mse.2}. Moreover, ${\mathtt{H}} = \widetilde{\mathtt{H}} +   o_p(\varrho_{n\alpha}) $ and $\widetilde{\mathtt{H}} = O_p(1)$. Then, by using similar arguments in \cite[Lemma A.1]{donald2001choosing}, it can be shown that $\mathbb E[\Vert \mathcal S_n ' (\overline \theta - \theta_0)\Vert ^2 |x] = \text{tr}(\widetilde{\mathtt{H}}^{-1}\mathbb E[\widehat{\mathtt{A}}(\alpha)|x] \widetilde{\mathtt{H}}^{-1})  + o_p(\rho_{n\alpha})$ and the desired result is deduced from \eqref{eq: delta 1} and \eqref{eq: ahat1}.

								\textbf{\violet{Discussion on \eqref{eq: mse: tmp}}}
								
								\label{r4p6}\violet{It is \eqref{eq: ahat} that determines the convergence rate of the conditional MSE. The criterion for choosing the regularization parameter in Appendix~\ref{stepwise} is designed to utilize this fact and the conditional MSE estimators available in the literature. To see this in detail, we first note that  because of \eqref{eq: delta 1},  $\Delta_1$ in \eqref{eq: ahat}  bounded above by $\widetilde{\mathtt{c}}_g \Delta_2$, where $\widetilde{\mathtt{c}}_g = \lambda_{\max}(\mathbb E[\dot{m}_{2i,0} ^2 g_{0i} g_{0i}'])$. Moreover, the proof of Lemma~\ref{lem.mse}.\ref{lem.mse.2}, in particular \eqref{eq: mse: 1} and \eqref{eq: mse: 6}, tells us that the second term in \eqref{eq: ahat} is bounded above by $  \widetilde{\mathtt{c}}_{g} \sigma_{  \psi_0} ^2 \sum_{i} [\mathcal P_{n\alpha} ^2 ]_{ii} / n  $.  Therefore, if we ignore $o_p(\varrho_{n\alpha})$ terms, \eqref{eq: ahat} is bounded above by   $$\widetilde{\mathtt{c}}_g\left(\psi_0 ' (\mathcal T_n \pi_n) ' (\mathcal P_{n\alpha} - \mathcal I) ^2 \mathcal T_n \pi_n \psi_0 /n  +  \sigma_{  \psi_0} ^2 \text{tr} (\mathcal P_{n\alpha} ^2)/n \right).$$ The term in the parenthesis is similar to $R_\nu (\alpha)$ in \citet[p.\ 389]{Carrasco2012} if $\nu$ in their notation is set to $\psi_0$. The constant $\widetilde{\mathtt{c}}_g$ and the vector $\psi_0$ neither  affect   the convergence rate nor are available in practice. Therefore, in practical estimation, they need to be replaced by reasonable candidates. The criterion in \eqref{eq: mse: tmp} suggests using an arbitrary $h \in \mathbb R^{d_e}$ and $\mathtt{c}_g$ in place of $\psi_0$ and $\widetilde{\mathtt{c}}_g$ respectively, which are readily available in practice.     }

\bibliography{lit}   
\end{document}